%% file: main.tex
\newif\ifarxiv\arxivtrue%

\documentclass[orivec,envcountsect,envcountsame]{llncs}

\usepackage{silence}
\usepackage{authblk}

\makeatletter\renewcommand\paragraph{\@startsection{paragraph}{4}{\z@}{1.6ex \@plus1ex \@minus.2ex}{-.5em}{\normalfont\normalsize\itshape}}\makeatother % chktex 6

\ifarxiv%
\RequirePackage[paperheight=235mm,paperwidth=155mm,textwidth=12.2cm,textheight=19.3cm,inner=55pt]{geometry}
\else%
\RequirePackage[paperheight=235mm,paperwidth=155mm,textwidth=12.2cm,textheight=19.3cm,inner=55pt]{geometry}
\fi%
\makeatletter
\def\@citecolor{blue}%
\def\@urlcolor{blue}%
\def\@linkcolor{blue}%
\ifarxiv\else%

\fi%
\def\orcidID#1{\href{http://orcid.org/#1}{\protect\raisebox{-1.25pt}{\protect\includegraphics{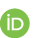}}}}
\makeatother

\usepackage{bbding}
\usepackage{graphicx}

\usepackage{wrapfig}
\usepackage{amssymb}
\usepackage{amsmath}

\usepackage{mathtools}
\usepackage{mathpartir}
\usepackage{stmaryrd}

\usepackage{graphicx}
\usepackage{prftree}
\usepackage{tikz,tikz-cd}
\usetikzlibrary{automata, arrows.meta, positioning}
\tikzset{
}
\usepackage[capitalize]{cleveref}

\usepackage{listings}
\usepackage{etoolbox}
\apptocmd{\sloppy}{\hbadness10000\relax}{}{}

\title{A Complete Inference System for Skip-free Guarded Kleene Algebra with Tests}
\titlerunning{Completeness for skip-free \texorpdfstring{\GKAT}{GKAT}}
\author{Tobias Kapp\'e\inst{1,2}(\Envelope)\orcidID{0000-0002-6068-880X} \and Todd Schmid\inst{3}\orcidID{0000-0002-9838-2363} \and Alexandra Silva\inst{4}\orcidID{0000-0001-5014-9784}} % chktex 8 chktex 36

\institute{
{Open University of the Netherlands \email{tobias.kappe@ou.nl}}
\and{ILLC, University of Amsterdam, NL}
\and{University College London, London, UK}
\and{Cornell University, Ithaca, NY, USA}}

\usepackage{thm-restate}
\input{fix-restatable.tex} % chktex 1

\usepackage{xspace}

\begin{document}

% SYMBOLS ----------
\newcommand{\At}{\mathsf{At}}

\newcommand{\KAT}{\textsf{KAT}\xspace}
\newcommand{\GKAT}{\textsf{GKAT}\xspace}

\newcommand{\BExp}{\mathsf{BExp}}
\newcommand{\GExp}{\mathsf{GExp}}
\newcommand{\SExp}{\mathsf{StExp}}
\newcommand{\Det}{\mathsf{Det}}

\newcommand{\Id}{\operatorname{Id}}
\newcommand{\Set}{\mathbf{Set}}
\newcommand{\Coalg}{\operatorname{Coalg}}

\renewcommand{\P}{\mathcal{P}}
\renewcommand{\L}{\mathcal{L}}
\newcommand{\GT}{\mathsf{GT}}

\newcommand{\bisim}{\mathrel{\raisebox{0.1em}{\(\underline{\leftrightarrow}\)}}}

\newcommand{\grph}{\mathsf{grph}}
\newcommand{\func}{\mathsf{func}}
\newcommand{\gtr}{\operatorname{gtr}}
\newcommand{\rtg}{\operatorname{rtg}}

\newcommand{\id}{\operatorname{id}}

\newcommand{\sem}[1]{\left\lceil\!\!\left\lfloor#1\right\rceil\!\!\right\rfloor}
\newcommand{\semrel}[1]{\llbracket#1\rrbracket}
\newcommand{\tr}[1]{\mathrel{\raisebox{-0.2em}{\(\xrightarrow{#1}\)}}}

\newcommand{\test}[1]{\textsf{\color{blue}#1}}
\newcommand{\action}[1]{\textsf{\color{orange!95!black}#1}}
\newcommand{\constant}[1]{\textsf{\color{gray}#1}}
\newcommand{\method}[1]{\textsf{\color{green!60!black}#1}}
%-------------------

\newcommand{\floor}[1]{\lfloor#1\rfloor}
\newcommand{\while}[1]{{\mathbin{^{(#1)}}}}

\newcommand{\embed}{\mathsf{embed}}

% commands for entry/body labellings
\newcommand{\eo}{\mathrel{\raisebox{-2pt}{\color{blue}\(\to_{\mathsf{e}}\)}}}
\newcommand{\bo}{\mathrel{\raisebox{-2pt}{\(\to_{\mathsf{b}}\)}}}
\newcommand{\diredge}{\mathrel{\curvearrowright}}
\newcommand{\etr}[1]{\mathrel{\raisebox{-.2em}{\color{blue}\(\xrightarrow{{\color{black}#1}}_{\mathsf{e}}\)}}}
\newcommand{\btr}[1]{\mathrel{\raisebox{-.2em}{\(\xrightarrow{#1}_{\mathsf{b}}\)}}}
\newcommand{\N}{\mathbb{N}}
\newcommand{\incl}{\operatorname{incl}}
\newcommand{\ofMil}{\mathsf{ofMil}}

\makeatletter
\newcommand{\bigplus}{%
  \DOTSB\mathop{\mathpalette\mattos@bigplus\relax}\slimits@%
}
\newcommand\mattos@bigplus[2]{%
  \vcenter{\hbox{%
    \sbox\z@{$#1\sum$}%
    \resizebox{!}{0.9\dimexpr\ht\z@+\dp\z@}{\raisebox{\depth}{$\m@th#1+$}}%
  }}%
  \vphantom{\sum}%
}
\makeatother

\maketitle

\begin{abstract}
Guarded Kleene Algebra with Tests (GKAT) is a fragment of Kleene Algebra with Tests (KAT) that was recently introduced to reason efficiently about imperative programs. In contrast to KAT,  GKAT does not have an algebraic axiomatization, but relies on an analogue of Salomaa's axiomatization of Kleene Algebra. In this paper, we present an algebraic axiomatization and prove  two completeness results for a large fragment of GKAT consisting of \emph{skip-free programs}.
\end{abstract}

\begin{quote}
    {\bfseries \color{red} Note.} The original edition of this paper contained a gap in the completeness proof (which did not appear in the main text).  A complete description of the gap and its correction can be found in \cref{app:changes}.
\end{quote}

% SECTION ---------------
\section{Introduction}
%
%    Deciding whether a pair of programs do the same thing can be hard.
%    The task is made much easier if we know the programs are of a certain form and if we abstract away certain details of their execution.
%    A simple example of this is Kleene algebra (\textsf{KA})~\cite{kleene algebra}, which provides a system for deciding equivalences between programs that come in the form of regular expressions.

    Kleene algebra with tests (\KAT)~\cite{katintro} is a logic for reasoning about semantics and equivalence of simple imperative programs.
    It extends Kleene Algebra (\textsf{KA}) with Boolean control flow, which enables encoding of conditionals and while loops.
 %   \KAT is an expressive, widely used calculus for specifying and reasoning about control flow in imperative programs.

    \KAT has been applied to verification tasks.
    For example, it was used in proof-carrying Java programs~\cite{DBLP:journals/entcs/KotK05}, in compiler optimization~\cite{DBLP:conf/cl/KozenP00}, and file system verification~\cite{DBLP:conf/pldi/ChajedTKZ19}.
    More recently, \KAT was used for reasoning about packet-switched networks, serving as a core to {\sffamily NetKAT}~\cite{netkat} and Probabilistic {\sffamily NetKAT}~\cite{probnetkat,cantorscott}.

    The success of \KAT in networking is partly due to its dual nature: it can be used to both specify and verify network properties.
    Moreover, the implementations of {\sffamily NetKAT} and {\sffamily ProbNetKAT} were surprisingly competitive with state-of-the-art tools~\cite{netkat-decision,mcnetkat}.
    Part of the surprise with the efficiency of these implementations is that the decision problem for equivalence in both \KAT and {\sffamily NetKAT} is PSPACE-complete~\cite{katcompleteness,netkat}.
    Further investigations~\cite{gkatpopl} revealed that the tasks performed in {\sffamily NetKAT} only make use of a fragment of \KAT. % chktex 13
    It turns out that the difficulty of deciding equivalence in \KAT can largely be attributed to the non-deterministic nature of \KAT programs.
    If one restricts to \KAT programs that operate deterministically with respect to Boolean control flow, the associated decision problem is almost linear.
    This fragment of \KAT was first identified in~\cite{kozentseng} and further explored as \emph{guarded Kleene algebra with tests} (\GKAT)~\cite{gkatpopl}.

    The study in~\cite{gkatpopl} proved that the decision problem for \GKAT programs is almost linear, and proposed an axiomatization of equivalence.
    However, the axiomatization suffered from a serious drawback: it included a powerful \emph{uniqueness of solutions axiom} (\textsf{UA}), which greatly encumbers algebraic reasoning in practice.
    In order to use (\textsf{UA}) to show that a pair of programs are equivalent, one needs to find a system of equations satisfied by both.
    Even more worryingly, the axiomatization contained a fixed-point axiom with a side condition reminiscent of Salomaa's axiomatization for regular expressions.
    This axiom is known to be non-algebraic, and thus impairs the use of the axiomatic reasoning in context (as substitution of atomic programs is not sound anymore).
    The authors of~\cite{gkatpopl} left as open questions whether (\textsf{UA}) can be derived from the other \GKAT axioms and whether the non-algebraic side condition can be removed.
    Despite the attention \GKAT has received in recent literature~\cite{gkaticalp,gkatlearning,processesparametrised}, these questions remain open.

    \smallskip
    In the present work, we offer a partial answer to the questions posed in~\cite{gkatpopl}.
    We show that proving the validity of an equivalence in \GKAT does not require (\textsf{UA}) if the pair of programs in question are of a particular form, what we call \emph{skip-free}.
    This fragment of \GKAT is expressive enough to capture a large class of programs, and it also provides a better basis for algebraic reasoning: we show that the side condition of the fixed-point axiom can be removed.
    Our inspiration to look at this fragment came from recent work by Grabmayer and Fokkink on the axiomatization of \emph{one-free star expressions modulo bisimulation}~\cite{onefreeregexlics,regexlics}, an important stepping stone to solve a decades-open problem posed by Milner~\cite{milner}.

    % Compare the problem to a similar problem from process algebra

    % The problem of showing that the \textsf{UA} follows from other axioms is not entirely new.
    % In 1984, Robin Milner proposed a nonstandard interpretation of regular expressions~\cite{milner}, and along with it a nonstandard form of program equivalence between regular expressions.
    % He suggests a plausible axiomatization of this equivalence in the same paper and shows it is sound, but leaves completeness as an open problem.

    In a nutshell, our contribution is to identify a large fragment of \GKAT, what we call the \emph{skip-free fragment}, that admits an algebraic axiomatization.
    We axiomatize both bisimilarity and language semantics and provide two completeness proofs.
    The first proves completeness of skip-free \GKAT \emph{modulo bisimulation}~\cite{gkaticalp}, via a reduction to completeness of Grabmayer and Fokkink's system~\cite{onefreeregexlics}.
    The second proves completeness of skip-free \GKAT w.r.t.\ language semantics via a reduction to skip-free \GKAT modulo bisimulation.
    We also show that equivalence proofs of skip-free \GKAT expressions (for both semantics) embed in full \GKAT. % chktex 13

    The next section contains an introduction to \GKAT and an overview of the open problems we tackle in the technical sections of the paper.

    % Say somewhere that we are focusing on bisimulation \GKAT

% SECTION ---------------
\section{Overview}\label{sec:overview}
    In this section we provide an overview of our results. We start with a motivating example of two imperative programs to discuss program equivalence as a verification technology. We then show how \GKAT can be used to solve this problem and explore the open questions that we tackle in this paper.

    \paragraph{\bf Equivalence for Verification. }
    In the game \emph{Fizz! Buzz!}~\cite{fizzbuzz}, players sit in a circle taking turns counting up from one.
    Instead of saying any number that is a multiple of \(3\), players must say ``fizz'', and multiples of \(5\) are replaced with ``buzz''.
    If the number is a multiple both \(3\) and \(5\), the player must say ``fizz buzz''.% (with twice the enthusiasm).

    \begin{figure}[t]
        \begin{center}
            \sffamily\scriptsize
            \fbox{\parbox{0.41\textwidth}{
                def \method{fizzbuzz1} = \hfill{\rmfamily (i)}\\
                \hspace*{1em}\action{\(n := 1\)}; \\
                \hspace*{1em}while \test{\(n \leq 100\)} do \\
                \hspace*{2em}if \test{\(3|n\)} then \\
                \hspace*{3em}if \test{not \(5|n\)} then \\
                \hspace*{4em}\action{print \textit{fizz}}; %\\
               % \hspace*{4em}
               \action{\(n\)\texttt{++}}; \\
                \hspace*{3em}else \\
                \hspace*{4em}\action{print \textit{fizzbuzz}}; %\\
               % \hspace*{4em}
               \action{\(n\)\texttt{++}}; \\
                \hspace*{2em}else if \test{\(5|n\)} then \\
                \hspace*{3em}\action{print \textit{buzz}}; %\\
                %\hspace*{3em}
                \action{\(n\)\texttt{++}}; \\
                \hspace*{2em}else \\
                \hspace*{3em}\action{print \(n\)}; %\\
               % \hspace*{3em}
               \action{\(n\texttt{++}\)}; \\
                \hspace*{1em}\action{print \textit{done!}};
            }}
            \quad
            \fbox{\parbox{0.41\textwidth}{
                def \method{fizzbuzz2} = \hfill {\rmfamily (ii)}\\
                \hspace*{1em}\action{\(n:= 1\)}; \\
                \hspace*{1em}while \test{\(n \leq 100\)} do \\
                \hspace*{2em}if \test{\(5 | n\) and \(3 | n\)} then \\
                \hspace*{3em}\action{print \textit{fizzbuzz}}; \\
                \hspace*{2em}else if \test{\(3 | n\)} then \\
                \hspace*{3em}\action{print \textit{fizz}}; \\
                \hspace*{2em}else if \test{\(5 | n\)} then \\
                \hspace*{3em}\action{print \textit{buzz}}; \\
                \hspace*{2em}else \\
                \hspace*{3em}\action{print \(n\)}; \\
                \hspace*{2em}\action{\(n\texttt{++}\)}; \\
                \hspace*{1em}\action{print \textit{done!}};
            }}
        \end{center}
        \caption{\label{fig:example}
            Two possible specifications of the ideal \emph{Fizz! Buzz!} player.
        }
    \end{figure}

    Imagine you are asked in a job interview to write a program that prints out the first \(100\) rounds of a perfect game of \emph{Fizz! Buzz!}.
    You write the function \method{fizzbuzz1} as given in \Cref{fig:example}(i). % chktex 36
    Thinking about the interview later that day, you look up a solution, and you find \method{fizzbuzz2}, depicted in \Cref{fig:example}(ii). % chktex 36
    You suspect that \method{fizzbuzz2} should do the same thing as \method{fizzbuzz1}, and after thinking it over for a few minutes, you realize your program could be transformed into the reference solution by a series of transformations that do not change its semantics:
    \begin{enumerate}
        \item
        Place the common action \action{\(n\)\texttt{++}} at the end of the loop.
        \item
        Replace \test{not \(5|n\)} with \test{\(5|n\)} and swap \action{print \textit{fizz}} with \action{print \textit{fizzbuzz}}.
        \item
        Merge the nested branches of \test{\(3|n\)} and \test{\(5|n\)} into one.
    \end{enumerate}

    Feeling somewhat more reassured, you ponder the three steps above.
    It seems like their validity is independent of the actual tests and actions performed by the code; for example, swapping the branches of an \textsf{if~-~then~-~else~-} block while negating the test should be valid under \emph{any} circumstances.
    This raises the question: is there a family of primitive transformations that can be used to derive valid ways of rearranging imperative programs?
    Furthermore, is there an algorithm to decide whether two programs are equivalent under these laws?

    \paragraph{\bf Enter \GKAT.}
    Guarded Kleene Algebra with Tests (\GKAT)~\cite{gkatpopl} has been proposed as a way of answering the questions above.
    Expressions in the language of \GKAT model skeletons of imperative programs, where the exact meaning of tests and actions is abstracted.
    The laws of \GKAT correspond to program transformations that are valid regardless of the semantics of tests and actions.

    Formally, \GKAT expressions are captured by a two-level grammar, generated by a finite set of tests $T$ and a finite set of actions $\Sigma$, as follows:
    \begin{align*}
        \BExp \ni b, c &::= 0 \mid 1 \mid t \in T \mid b \vee c \mid b \wedge c \mid \overline{b} \\
        \GExp \ni e, f &::= p \in \Sigma \mid b \mid e +_b f \mid e \cdot f \mid e^{(b)}
    \end{align*}
    \(\BExp\) is the set of \emph{Boolean expressions}, built from $0$ (\test{false}), $1$ (\test{true}), and primitive tests from $T$, and composed using $\vee$ (\test{or}), $\wedge$ (\test{and}) and $\overline{\phantom{a}}$ (\test{not}).
    \(\GExp\) is the set of \GKAT \emph{expressions}, built from tests (\textsf{assert} statements) and primitive actions $p \in \Sigma$.
    Here, $e +_b f$ is a condensed way of writing \textsf{`if~\test{b}~then~\action{e}~else~\action{f}'}, and $e^{(b)}$ is shorthand for \textsf{`while~\test{b}~do~\action{e}'}; the operator $\cdot$ models sequential composition.
    By convention, the sequence operator $\cdot$ takes precedence over the operator $+_b$.

    \begin{example}%
        \label{example:gkat-syntax-example}
        Abbreviating statements of the form \(\action{print \textsl{foo}}\) by simply writing \(\action{\textsl{foo}}\),
        \Cref{fig:example}(i) can be rendered as the \GKAT expression % chktex 36
        \begin{equation}%
        \label{example-compressed-left}
        (\action{\(n := 1\)}) \cdot
            \left(
                \begin{array}{c}
                    ( \action{\textsl{fizz}} \cdot \action{\(n\)\texttt{++}} +_{\test{\(\overline{5|n}\)}} \action{\textsl{fizzbuzz}} \cdot \action{\(n\)\texttt{++}}) +_{\test{\(3|n\)}} {} \\
                    (\action{\textsl{buzz}} \cdot \action{\(n\)\texttt{++}} +_{\test{\(5|n\)}} \action{\(n\)} \cdot \action{\(n\)\texttt{++}})
                \end{array}
            \right)^{(\test{\(n \leq 100\)})}
            \cdot \action{\textsl{done!}}
        \end{equation}
        Similarly, the program in \Cref{fig:example}(ii) gives the \GKAT expression % chktex 36
        \begin{equation}%
        \label{example-compressed-right}
            (\action{\(n := 1\)}) \cdot ( ( \action{\textsl{fizzbuzz}} +_{\test{\(5|n \wedge 3|n\)}} (\action{\textsl{fizz}} +_{\test{\(3|n\)}} (\action{\textsl{buzz}} +_{\test{\(5|n\)}} \action{\(n\)}))) \cdot \action{\(n\)\texttt{++}} )^{(\test{\(n \leq 100\)})} \cdot \action{\textsl{done!}}
        \end{equation}
    \end{example}

    \paragraph{\bf Semantics.}
    A moment ago, we stated that \GKAT equivalences are intended to witness program equivalence, regardless of how primitive tests and actions are interpreted.
    We make this more precise by recalling the \emph{relational} semantics of \GKAT programs~\cite{gkatpopl}.\footnote{A probabilistic semantics in terms of sub-Markov kernels is also possible~\cite{gkatpopl}.}
    The intuition behind this semantics is that if the possible states of the machine being programmed are modelled by some set $S$, then tests are predicates on $S$ (comprised of all states where the test succeeds), and actions are relations on $S$ (encoding the changes in state affected by the action).

    \begin{definition}[\cite{gkatpopl}]%
    \label{def:relational interpretation}
    A \emph{(relational) interpretation} is a triple $\sigma = (S, \mathsf{eval}, \mathsf{sat})$ where $S$ is a set, $\mathsf{eval}: \Sigma \to \mathcal{P}(S \times S)$ and $\mathsf{sat}: T \to \mathcal{P}(S)$.
    Each interpretation \(\sigma\) gives rise to a semantics $\semrel{-}_\sigma: \GExp \to \mathcal{P}(S \times S)$, as follows:
    \begin{align*}
    \semrel{0}_\sigma &= \emptyset
        & \semrel{\overline{a}}_\sigma &= \semrel{1}_\sigma \setminus \semrel{a}_\sigma \\
    \semrel{1}_\sigma &= \{ (s, s) : s \in S \}
        & \semrel{p}_\sigma &= \mathsf{eval}(p) \\
    \semrel{t}_\sigma &= \{ (s, s) : s \in \mathsf{sat}(t) \}
        & \semrel{e +_b f}_\sigma &= \semrel{b}_\sigma \circ \semrel{e}_\sigma \cup \semrel{\overline{b}}_\sigma \circ \semrel{f}_\sigma \\
    \semrel{b \wedge c}_\sigma &= \semrel{b}_\sigma \cap \semrel{c}_\sigma
        & \semrel{e \cdot f}_\sigma &= \semrel{e}_\sigma \circ \semrel{f}_\sigma \\
    \semrel{b \vee c}_\sigma &= \semrel{b}_\sigma \cup \semrel{c}_\sigma
        & \semrel{e^{(b)}}_\sigma &= {(\semrel{b}_\sigma \circ \semrel{e}_\sigma)}^* \circ \semrel{\overline{b}}_\sigma
    \end{align*}
    Here we use \(\circ\) for relation composition and ${}^*$ for reflexive transitive closure.
    \end{definition}

    \begin{remark}
    If $\mathsf{eval}(p)$ is a partial function for every $p \in \Sigma$, then so is $\semrel{e}_\sigma$ for each $e$.
    The above therefore also yields a semantics in terms of partial functions.
    \end{remark}

    The relation $\semrel{e}_\sigma$ contains the possible pairs of start and end states of the program $e$.
    For instance, the input-output relation of $\semrel{e +_b f}$ consists of the pairs in $\semrel{e}_\sigma$ (resp.\ $\semrel{f}_\sigma$) where the start state satisfies $b$ (resp.\ violates $b$).

    \begin{example}
    We could model the states of the machine running \emph{Fizz!\@ Buzz!\@} as pairs $(m, \ell)$, where $m$ is the current value of the counter $n$, and $\ell$ is a list of words printed so far; the accompanying maps $\mathsf{sat}$ and $\mathsf{eval}$ are given by:
    \begin{align*}
        \mathsf{sat}(\test{$k|n$}) &= \{ (m, \ell) \in S : m \equiv 0\ \text{mod}\ k \} \\
        \mathsf{sat}(\test{$n \leq k$}) &= \{ (m, \ell) \in S : m \leq k \} \\ % chktex 1
        \mathsf{eval}(\action{$n\texttt{++}$}) &= \{ ((m, \ell), (m+1, \ell) : (m, \ell) \in S \} \\
        \mathsf{eval}(\action{$n := k$}) &= \{ ((m, \ell), (k, \ell)) : (m, \ell) \in S) \} \\ % chktex 26
        \mathsf{eval}(\action{\textsf{\textsl{w}}}) &= \{ ((m, \ell), (m, \ell\textsf{\textsl{w}})) : (m, \ell) \in S \} & \tag{$\action{\textsf{\textsl{w}}} \in \{ \action{\textsf{\textsl{fizz}}}, \action{\textsf{\textsl{buzz}}}, \action{\textsf{\textsl{fizzbuzz}}} \}$} \\ % chktex 1
        \mathsf{eval}(\action{$n$}) &= \{ ((m, \ell), (m, \ell{}m)) : (m, \ell) \in S \}
    \end{align*}
    For instance, the interpretation of $\action{$n\texttt{++}$}$ connects states of the form $(m, \ell)$ to states of the form $(m + 1, \ell)$---incrementing the counter by one, and leaving the output unchanged.
    Similarly, \action{$\textsf{print}$} statements append the given string to the output.
    \end{example}

    On the one hand, this parameterized semantics shows that programs in the \GKAT syntax can be given a semantics that corresponds to the intended meaning of their actions and tests.
    On the other hand, it allows us to quantify over all possible interpretations, and thus abstract from the meaning of the primitives.

    As it happens, two expressions have the same relational semantics under any interpretation if and only if they have the same \emph{language semantics}~\cite{gkatpopl}, i.e., in terms of languages of \emph{guarded strings} as used in \KAT~\cite{katintro}.
    Since equivalence under the language semantics is efficiently decidable~\cite{gkatpopl}, so is equivalence under all relational interpretations.
    The decision procedure in~\cite{gkatpopl} uses bisimulation and known results from automata theory.
    These techniques are good for mechanization but hide the algebraic structure of programs.
    To expose them, algebraic laws of \GKAT program equivalence were studied.

    \paragraph{\bf Program transformations.} \GKAT programs are (generalized) regular expressions, which are intuitive to reason about and for which many syntactic equivalences are known and explored.
    In~\cite{gkatpopl}, a set of sound axioms $e \equiv f$ such that $\semrel{e}_\sigma = \semrel{f}_\sigma$ for all $\sigma$ was proposed, and it was shown that these can be used to prove a number of useful facts about programs. For instance, the following two equivalences are axioms of \GKAT:\@
    \begin{mathpar}
        e \cdot g +_b f \cdot g \equiv (e +_b f) \cdot g
        \and
        f +_{\overline{b}} e \equiv e +_b f
    \end{mathpar}
    The first axioms says that common code at the end of branches can be factored out, while the second says that the code in branches of a conditional can be swapped, as long as we negate the test.
    Returning to our running example, if we apply the first law to~\eqref{example-compressed-left} three times (once for each guarded choice), we obtain
    \begin{equation}%
    \label{example-compressed-left-transformed}
        (\action{\(n := 1\)}) \cdot
            \left(
                \left(
                \begin{array}{c}
                    ( \action{\textsl{fizzbuzz}} +_{\test{\(5|n\)}} \action{\textsl{fizz}}) +_{\test{\(3|n\)}} {} \\
                    (\action{\textsl{buzz}} +_{\test{\(5|n\)}} \action{\(n\)})
                \end{array}
                \right)
                \cdot \action{\(n\)\texttt{++}}
            \right)^{(\test{\(n \leq 100\)})}
            \cdot \action{\textsl{done!}}
    \end{equation}
    Finally, we can apply $(e +_b f) +_c (g +_b h) \equiv e +_{b \wedge c} (f +_c (g +_b h))$, which is provable from the axioms of \GKAT, to transform~\eqref{example-compressed-left-transformed} into~\eqref{example-compressed-right}.

    Being able to transform one \GKAT program into another using the axioms of \GKAT is useful, but the question arises: do the axioms capture \emph{all} equivalences that hold?
    More specifically, are the axioms of \GKAT powerful enough to prove that $e \equiv f$ whenever $\semrel{e}_\sigma = \semrel{f}_\sigma$ holds for all $\sigma$?

    In~\cite{gkatpopl}, a partial answer to the above question is provided: if we extend the laws of \GKAT with the \emph{uniqueness axiom} (\textsf{UA}), then the resulting set of axioms is sound and complete w.r.t.\ the language semantics.
    The problem with this is that (\textsf{UA}) is not really a single axiom, but rather an \emph{axiom scheme}, which makes both its presentation and application somewhat unwieldy.

    To properly introduce (\textsf{UA}), we need the following notion.
    \begin{definition}
        A \emph{left-affine system} is defined by expressions $e_{11}, \dots, e_{nn} \in \GExp$ and $f_1, \dots, f_n \in \GExp$, along with tests $b_{11}, \dots, b_{nn} \in \BExp$.
        A sequence of expressions $s_1, \dots, s_n \in \GExp$ is said to be a \emph{solution} to this system if
        \[
            s_i \equiv e_{i1} \cdot s_1 +_{b_{i1}} e_{i2} \cdot s_2 +_{b_{i2}} \cdots +_{b_{i(n-1)}} s_{n} \cdot e_{in} +_{b_{in}} f_i \quad (\forall i \leq n)
        \]
        Here, the operations $+_{b_{ij}}$ associate to the right.

        A left-affine system called \emph{guarded} if no \(e_{ij}\) that appears in the system successfully terminates after reading an atomic test.
        In other words, each coefficient denotes a productive program, meaning it must execute some action before successfully terminating---we refer to \Cref{section:lift-equivalences} for more details.
    \end{definition}
    Stated fully, (\textsf{UA}) says that if expressions  $s_1, \dots, s_n$ and $t_1, \dots, t_n$ are solutions to the same guarded left-affine system, then $s_i \equiv t_i$ for $1 \leq i \leq n$.

    On top of the infinitary nature of (\textsf{UA}), the side condition demanding guardedness prevents purely algebraic reasoning: replacing action symbols in a valid \GKAT equation with arbitrary \GKAT expressions might yield an invalid equation!
    The situation is analogous to the \emph{empty word property} used by Salomaa~\cite{salomaa} to axiomatize equivalence of regular expressions.
    The side condition of guardedness appearing in (\textsf{UA}) is inherited from another axiom of \GKAT, the fixed-point axiom, which in essence is the unary version of this axiom scheme and explicitly defines the solution of one guarded left-affine equation as a while loop.
    \[
    g \equiv eg +_b f \Longrightarrow g \equiv e^{(b)} f \qquad \text{if $e$ is guarded}.
    \]

%    To make this precise we  define a left-affine system of expression equations.

    \begin{remark}
    Part of the problem of the uniqueness axiom is that the case for general $n$ does not seem to follow easily from the case where $n = 1$.
    The problem here is that, unlike the analogous situation for Kleene algebra, there is no general method to transform a left-affine system with $n+1$ unknowns into one with $n$ unknowns~\cite{kozentseng}, even though this is possible in certain cases~\cite{gkatpopl}.
    \end{remark}

   \paragraph{\bf The open questions.}
   We are motivated by two open questions from~\cite{gkatpopl}:
    \begin{itemize}
        \item
        First, can the uniqueness axiom be eliminated?
        The other axioms of \GKAT contain the instantiation of (\textsf{UA}) for $n = 1$, which has so far been sufficient in all handwritten proofs of equivalence that we know.
        Yet (\textsf{UA}) seems to be necessary in both known completeness proofs.

        \item
        Second, can we eliminate the guardedness side condition?
        Kozen~\cite{kozen} showed that Salomaa's axiomatization is subsumed by a set of axioms that together imply existence and uniqueness of \emph{least} solutions to systems of equations, but this approach has not yet borne fruit in \GKAT. % chktex 13
    \end{itemize}
    \paragraph{\bf This paper.}
    Our main contribution is to show that, in a particular fragment of \GKAT, both questions can be answered in the positive (see \Cref{fig:gkat axioms}).
    %, and a rather neat axiomatization of equivalence --- without (\textsf{UA}) or the Salomaa-style side condition of guardedness --- emerges (see \Cref{fig:gkat axioms}).
    \begin{figure}[!t]
        \hspace{-4mm}
        \begin{tabular}{c @{\qquad} c @{\qquad} c}
            \underline{\bf Guarded Union} & \underline{\bf Sequencing} & \underline{\bf Loops} \\[2mm]
            \(\begin{array}{lrl}
                \text{(\textsf{G0})} \phantom{x} & x &= x +_{1} y \\[1mm]
                \text{(\textsf{G1})} & x &= x +_b x \\[1mm]
                \text{(\textsf{G2})} & x +_b y &= y +_{\bar b} x \\[1mm]
                \text{(\textsf{G3})} & \mathrlap{\hspace*{-9mm} x +_b (y +_c z) = } \\
                & \mathrlap{\hspace*{-5mm}(x +_b y) +_{b \vee c} z}
            \end{array}\)&
            \(\begin{array}{lrl}
                \text{(\textsf{G6})} \phantom{x} & 0 x &= 0 \\[1mm]
                \text{(\textdagger)} & x 0 &= 0 \\[1mm]
                \text{(\textsf{G7})} & x (y z) &= (x y) z \\[1mm]
                \text{(\textsf{G8})} & (x +_b y) z &= x z +_b y z
            \end{array}\)&
            \(
            \begin{array}{lrl}
            \text{(\textsf{FP})} & \mathrlap{\hspace*{-17mm} x^{(b)}y =} \\
            & \mathrlap{\hspace*{-14mm}x (x^{(b)}y) +_b y} \\[1mm]
            \text{(\textsf{RSP})} \phantom{x} & \inferrule{z = x z +_b y}{z = x^{(b)}y}
            \end{array}
            \)
        \end{tabular}
        \caption{\label{fig:gkat axioms} Axioms for language semantics skip-free  \GKAT (in addition to Boolean algebra axioms for tests, see \cref{fig:boolean algebra axioms}). If the axiom marked $\dagger$ is omitted, these axiomatize a finer semantics, bisimilarity.
        }
    \end{figure}

   In \Cref{sec:skipfree}, we present what we call the skip-free fragment of \GKAT, consisting of programs that do not contain {\sffamily assert} statements in the body (other than {\sffamily assert \test{false}}); in other words, Boolean statements are restricted to control statements.
   For this fragment, we show that the axiom scheme (\textsf{UA}) can be avoided entirely. In fact, this is true for language semantics (as first introduced in~\cite{gkatpopl}) as well as for the bisimulation semantics of~\cite{gkaticalp}.

    In \Cref{sec:regexp}, we provide a bridge to a recent result in process algebra.
    In the 80s, Milner offered an alternative interpretation of regular expressions~\cite{milner}, as what he called \emph{star behaviours}.
    Based on work of Salomaa from the 1960s~\cite{salomaa}, Milner proposed a sound axiomatization of the \emph{algebra of star behaviours}, but left completeness an open problem.
    After 38 years, it was recently solved by Clemens Grabmayer~\cite{regexlics} following up on his joint work with Wan Fokkink showing that a suitable restriction of Milner's axioms is complete for the \emph{one-free fragment} of regular expressions modulo bisimulation~\cite{onefreeregexlics}.
    We leverage their work with an interesting embedding of skip-free \GKAT into the one-free regular expressions.

    This leads to two completeness results.
    In \Cref{sec:bisim complete}, we start by focusing on the \emph{bisimulation semantics} of the skip-free fragment, and then in \Cref{sec:completeness for skip-free GKAT} expand our argument to its \emph{language semantics}.  More precisely, we first provide a reduction of the completeness of skip-free \GKAT \emph{up to bisimulation} to the completeness of Grabmayer and Fokkink's one-free regular expressions modulo bisimulation~\cite{onefreeregexlics}.
    We then provide a reduction of the completeness of skip-free  \GKAT modulo language semantics to the completeness of skip-free  \GKAT modulo bisimulation via a technique inspired by the \emph{tree pruning} approach of~\cite{gkaticalp}.

    Finally, in \Cref{sec:relation to GKAT}, we connect our semantics of skip-free \GKAT expressions to the established semantics of full \GKAT. % chktex 13
    We also connect the syntactic proofs between skip-free \GKAT expressions in both our axiomatization and the existing one.
    In conjunction with the results of \cref{sec:bisim complete,sec:completeness for skip-free GKAT}, the results in \cref{sec:relation to GKAT} make a significant step towards answering the question of whether the axioms of \GKAT give a complete description of program equivalence, in the positive.

    \ifarxiv%
    Proofs are deferred to \hyperref[appendices]{the appendices}.
    \else%
    Proofs appear in the full version~\cite{fullversion}.
    \fi%

% SECTION ---------------
\section{Introducing Skip-free  \GKAT}\label{sec:skipfree}
    The axiom scheme (\textsf{UA}) can be avoided entirely in a certain fragment of  \GKAT, both for determining bisimilarity and language equivalence.
    In this section, we give a formal description of the expressions in this fragment and their semantics.
    %, and how the theory of this fragment relates to  \GKAT as it is studied in~\cite{gkatpopl,gkaticalp}.

    \paragraph*{\bf Skip-free expressions.}
    The fragment of \GKAT in focus is the one that \emph{excludes sub-programs that may accept immediately}, without performing any action.
    Since these programs can be ``skipped'' under certain conditions, we call the fragment that avoids them \emph{skip-free}.
    Among others, it prohibits sub-programs of the form {\sffamily assert \test{b}} for $\test{b} \neq \test{false}$, but also {\sffamily while \test{false} do $p$}, which is equivalent to {\sffamily assert \test{true}}.

    \begin{definition}
        Given a set \(\Sigma\) of \emph{atomic actions}, the set \(\GExp^-\) of \emph{skip-free \GKAT expressions} is given by the grammar
        \[
            \GExp^- \ni e_1, e_2 ::= 0 \mid p \in \Sigma \mid e_1 +_b e_2 \mid e_1\cdot e_2 \mid e_1^{(b)} e_2
        \]
        where \(b\) ranges over the Boolean algebra expressions \(\BExp\).
    \end{definition}

    Unlike full \GKAT, in skip-free \GKAT the loop construct is treated as a binary operation, analogous to Kleene's original star operation~\cite{kleene}, which was also binary.
    This helps us avoid loops of the form $e^{(b)}$, which can be skipped when $b$ does not hold.
    The expression \(e_1 ^{(b)} e_2\) corresponds to \(e_1 ^{(b)} \cdot e_2\) in  \GKAT. % chktex 13

    \begin{example}\label{eg:example expression}
        Using the same notational shorthand as in \Cref{example:gkat-syntax-example}, the block of code in \Cref{fig:example}(ii) can be cast as the skip-free \GKAT expression % chktex 36
        \[
            (\action{$n$ := 1}) \cdot ((\action{\textit{fizzbuzz}} +_{\test{$3|n$} \wedge \test{$5|n$}} (\action{\textit{fizz}} +_{\test{$3|n$}} (\action{\textit{buzz}} +_{\test{$5|n$}} \action{$n$}))) \cdot \action{$n$\texttt{++}})^{(\test{$n \leq 100$})} (\action{\textit{done!}}) % chktex 1
        \]
        Note how we use a skip-free loop of the form $e_1 \while{b} e_2$ instead of the looping construct $e_1^{(b)}$ before concatenating with $e_2$, as was done for \GKAT. % chktex 13
    \end{example}

    \subsection{Skip-free Semantics}
    There are three natural ways to interpret skip-free \GKAT expressions: as \emph{automata}, as \emph{behaviours}, and as \emph{languages}.\footnote{We will connect these to the relational semantics from \Cref{def:relational interpretation} in \Cref{sec:relation to GKAT}.}
    After a short note on Boolean algebra, we shall begin with the automaton interpretation, also known as the \emph{small-step semantics}, from which the other two can be derived.
    \begin{figure}[!t]
        \begin{gather*}
            x \vee 0 = x
            \qquad
            x \vee \bar x = 1
            \qquad
            x \vee y = y \vee x
            \qquad
            x\vee (y \wedge z) = (x \vee y)\wedge (x \vee z)
            \\
            x \wedge 1 = x
            \qquad
            x \wedge \bar x = 0
            \qquad
            x\wedge y = y \wedge x
            \qquad
            x\wedge (y \vee z) = (x \wedge y)\vee (x \wedge z)
        \end{gather*}
        \caption{\label{fig:boolean algebra axioms} The axioms of Boolean algebra~\cite{huntington}.}
    \end{figure}

    \paragraph{\bf Boolean algebra.}
    To properly present our automata, we need to introduce one more notion.
    Boolean expressions $\BExp$ are a syntax for elements of a \emph{Boolean algebra}, an algebraic structure satisfying the equations in \cref{fig:boolean algebra axioms}.
    When a Boolean algebra is freely generated from a finite set of basic tests ($T$ in the case of $\BExp$), it has a finite set \(\At\) of nonzero minimal elements called \emph{atoms}.
    Atoms are in one-to-one correspondence with sets of tests, and the free Boolean algebra is isomorphic to \(\mathcal P(\At)\), the sets of subsets of \(\At\), equipped with \({\vee} = {\cup}\), \({\wedge} = {\cap}\), and \(\overline{(-)} = \At\setminus (-)\).
    In the context of programming, one can think of an atom as a complete description of the machine state, saying which tests are true and which are false.
    We will denote atoms by the Greek letters $\alpha$ and $\beta$, sometimes with indices.
    Given a Boolean expression $b\in \BExp$ and an atom $\alpha \in \At$ we say that $\alpha$ entails $b$, written $\alpha \leq b$, whenever $\overline\alpha \vee b = 1$, or equivalently $\alpha \vee b = b$.
    \paragraph{\bf Automata.}
    Throughout the paper, we use the notation $\bullet + S$ where $S$ is a set and $\bullet$ is a symbol to denote the disjoint union (coproduct) of $\{ \bullet \}$ and $S$.
    When $X$ and $Y$ are sets, we write $Y^X$ for the set of functions from $X$ to $Y$.

    The small-step semantics of a skip-free  \GKAT expression uses a special type of deterministic automaton.
    A \emph{skip-free automaton} is a pair $(X,h)$, where $X$ is a set of \emph{states} and $h\colon X \to (\bot + \Sigma\times(\checkmark + X))^\At$ is a \emph{transition structure}.
    At every \(x \in X\) and for any $\alpha \in \At$, one of three things can happen:
    \begin{enumerate}
        \item \(h(x)(\alpha) = (p, y)\), which we write as \(x \tr{\alpha \mid p}_h y\), means the state $x$ under $\alpha$ makes a transition to a new state $y$, after performing the action $p$;
        \item  \(h(x)(\alpha) = (p, \checkmark)\), which we write as \(x \tr{\alpha \mid p}_h \checkmark\), means the state $x$ under $\alpha$ successfully terminates with action $p$;
        \item \(h(x)(\alpha) = \bot\), which we write as \(x \downarrow_h \alpha\), means the state $x$ under $\alpha$ terminates with failure.
        Often we will leave these outputs implicit.
    \end{enumerate}
    We often drop the subscript from the notations above when no confusion is likely.

    \begin{definition}[Automaton of expressions]\label{def:autexp}
        We equip the set \(\GExp^-\) of all skip-free  \GKAT expressions with an automaton structure \((\GExp^-, \partial)\) given in \cref{fig:small-step GKAT}, representing step-by-step execution.
        Given \(e \in \GExp^-\),  we denote the set of states reachable from \(e\) by \(\langle e\rangle\) and call this the \emph{small-step semantics} of \(e\).
    \end{definition}

    \begin{figure}[!t]
        \begin{gather*}
            \prftree{}{p \tr{\alpha\mid p} \checkmark}
            \hspace{1.5em}
            \prftree{e_1 \tr{\alpha \mid p} e' \quad \alpha \leq b}{e_1 +_b e_2 \tr{\alpha \mid p} e'}
            \hspace{1.5em}
            \prftree{e_2 \tr{\alpha \mid p} e' \quad \alpha \not\leq b}{e_1 +_b e_2 \tr{\alpha \mid p} e'} \hspace{1.5em}
            \prftree{e_1 \tr{\alpha \mid p} e'}{e_1e_2 \tr{\alpha \mid p} e'e_2} \\
            \prftree{e_1 \tr{\alpha \mid p} \checkmark}{e_1e_2 \tr{\alpha \mid p} e_2}
            \hspace{1.5em}
            \prftree{e_1 \tr{\alpha \mid p} e' \quad \alpha \leq b}{e_1^{(b)}e_2 \tr{\alpha \mid p} e'(e_1^{(b)}e_2)}
            \hspace{1.5em}
            \prftree{e_1 \tr{\alpha \mid p} \checkmark \quad \alpha \leq b}{e_1^{(b)}e_2 \tr{\alpha \mid p} e_1^{(b)}e_2}
            \hspace{1.5em}
            \prftree{e_2 \tr{\alpha \mid p} e' \quad \alpha \not\leq b}{e_1^{(b)}e_2 \tr{\alpha \mid p} e'}
        \end{gather*}
        \caption{\label{fig:small-step GKAT}
        The small-step semantics of skip-free  \GKAT expressions.
        }
    \end{figure}
    The small-step semantics of skip-free  \GKAT expressions is inspired by Brzozowski's \emph{derivatives}~\cite{brzozowski}, which provide an automata-theoretic description of the step-by-step execution of a regular expression.
    Our first lemma tells us that, like regular expressions, skip-free  \GKAT expressions correspond to finite automata.

    \begin{restatable}{lemma}{lemlocalfiniteness}\label{lem:local finiteness}
        For any \(e \in \GExp^-\), \(\langle e\rangle\) has finitely many states.
    \end{restatable}

   \begin{example}\label{eg:example automaton}
        The automaton that arises from the program \method{fizzbuzz2} is below, with \(\test{\(a\)} = \test{\(n \le 100\)}\), \(\test{\(b\)} = \test{\(3|n\)}\), and \(\test{\(c\)} = \test{\(5 | n\)}\).
        The expression \(e\) is the same as in \Cref{eg:example expression}, $e_1$ is the same as $e$ but without the action \action{$n := 0$} in front, and $e_2 = \action{\(n\)\texttt{++}} \cdot e_1$.
        We also adopt the convention of writing $x \tr{b|p} x'$ where $b \in \BExp$ to represent all transitions $x \tr{\alpha|p} x'$ where $\alpha \leq b$.
        %\begin{wrapfigure}{r}{.58\textwidth}
            %\vspace{-35pt}
            \[\begin{tikzpicture}[scale=.7,font=\scriptsize]
                \node[state,minimum size=.64cm] (e) at (0,0) {\(e\)};
                \node[state,minimum size=0cm] (e1) at (3,0) {\(e_1\)};
                \node[state,minimum size=0cm] (e2) at (9,0) {\(e_2\)};
                \node (final) at (3,1.8) {\(\checkmark\)};

                \draw (e) edge[->] node[below] {\(1 \mid \action{$n := 1$}\)} (e1);
                \draw (e1) edge[->, bend left,looseness=0.6] node[above]
                    {\(\begin{aligned}
                        \test{$abc$} &\mid \action{\textsl{fizzbuzz}}, &
                        \test{$ab\bar c$} &\mid \action{\textsl{fizz}}, \\
                        \test{$a\bar b c$} &\mid \action{\textsl{buzz}}, &
                        \test{$a \bar b \bar c$} &\mid \action{$n$}
                    \end{aligned}\)} (e2);
                \draw (e2) edge[->] node[below] {\(1 \mid \action{$n$\texttt{++}}\)} (e1);
                \draw (e1) edge[->] node[left] {\(\bar{\test{$a$}} \mid \action{\textit{done!}}\)} (final);
            \end{tikzpicture}
            \]
            %\vspace{-50pt}
        %\end{wrapfigure}
    \end{example}

    % \begin{remark}
    %     The coalgebraic semantics of  \GKAT uses a different functor, \(G' = (\bot + \checkmark + \Sigma \times \Id)^{\At}\)~\cite{gkatpopl,gkaticalp}.
    %     Interpreting skip-free programs as \(G'\)-coalgebras, every state either accepts \(\At\) or no atomic test at all.
    %     What we have done is replace the state of a skip-free program that accepts an atomic test (since it must accept every atomic test) with \(\checkmark\).
    %     See \cref{app:relation to \GKAT} for details.
    % \end{remark}

    The automaton interpretation of a skip-free  \GKAT expression (its small-step semantics) provides an intuitive visual depiction of the details of its execution. This is a useful view on the operational semantics of expressions, but sometimes one might want to have a more precise description of the global behaviour of the program.   The remaining two interpretations of skip-free  \GKAT expressions aim to capture two different semantics of expressions: one finer, bisimilarity, that makes a distinction on the branching created by how its states respond to atomic tests, which actions can be performed, and when successful termination and crashes occur; another coarser, language semantics, that assigns a language of traces to each expression capturing all sequences of actions that lead to successful termination. The key difference between these two semantics will be their ability to distinguish programs that crash early in the execution from programs that crash later---this will become evident in the axiomatizations of both semantics. We start by presenting the language semantics as this is the more traditional one associated with \GKAT (and regular) expressions.

     \paragraph{\bf Language semantics.}
    Formally, a \emph{(skip-free) guarded trace} is a nonempty string of the form \(\alpha_1p_1 \cdots \alpha_n p_n\), where each \(\alpha_i \in \At\) and \(p_i \in \Sigma\). Intuitively, each $\alpha_i$ captures the state of program variables needed to execute program action $p_i$ and the execution of each $p_i$ except the last yields a new program state $\alpha_{i+1}$.
    A \emph{skip-free guarded language} is a set of guarded traces.

    Skip-free guarded languages should be thought of as sets of strings denoting successfully terminating computations.

    \begin{definition}[Language acceptance]
        In a skip-free automaton  \((X, h)\) with a state \(x \in X\), the \emph{language accepted by \(x\)} is the skip-free guarded language
        \[
            \L(x, (X, h)) = \{\alpha_1p_1\cdots \alpha_{n}p_n \mid x \tr{\alpha_1 \mid p_1} x_1 \tr{} \cdots \tr{} x_n \tr{\alpha_n \mid p_n} \checkmark\}
        \]
        If \((X, h)\) is clear from context, we will simply write \(\L(x)\) instead of \(\L(x, (X, h))\).
        If \(\L(x) = \L(y)\), we write \(x \sim_\L y\) and say that \(x\) and \(y\) are \emph{language equivalent}.
    \end{definition}

    Each skip-free \GKAT expression is a state in the automaton of expressions (\Cref{def:autexp}) and therefore accepts a language.
    The language accepted by a skip-free \GKAT expression is the set of successful runs of the program it denotes.
    Analogously to $\GKAT$, we can describe this language inductively.
    %The next lemma formally connects the inductive definition of language to acceptance in the automaton of expressions.

    \begin{restatable}{lemma}{lemlanguagecomputation}\label{lem:language}
       Given an expression $e\in \GExp^-$, the language accepted by $e$ in \((\GExp^-, \partial)\), i.e., $\L(e) = \L(e, (\GExp^-, \partial))$ can be characterized as follows:
        \begin{gather*}
            \L(0) = \emptyset
            \quad
            \L(p) = \{\alpha p \mid \alpha \in \At\}
            \quad
            \L(e_1 +_b e_2) = b\L(e_1) \cup \bar b\L(e_2)
            \\
            \L(e_1\cdot e_2) = \L(e_1)\cdot \L(e_2)
            \quad
            \L(e_1^{(b)}e_2) = \bigcup_{n \in \mathbb N} (b\L(e_1))^n\cdot\bar b\L(e_2)
        \end{gather*}
        Here, we write \(b L = \{\alpha p w \in L \mid \alpha \leq b\}\) and \(L_1 \cdot L_2 = \{ wx : w \in L_1, x \in L_2 \} \), while $L^0 = \{ \epsilon \}$ (where $\epsilon$ denotes the empty word) and $L^{n+1} = L \cdot L^n$.
    \end{restatable}

    \cref{lem:language} provides a way of computing the language of an expression $e$ without having to generate the automaton for $e$.

    \paragraph{\bf Bisimulation semantics.}
       Another, finer, notion of equivalence that we can associate with skip-free automata is bisimilarity. %Bisimilarity can be defined by unrolling automata into trees of a certain form, an approach explored in \cref{sec:completeness for skip-free GKAT}.
       %Another way, of more immediate relevance, is to construct a bisimulation relation that witnesses the entanglement of the behaviours of the states in two skip-free automata.

    \begin{definition}\label{def:g bisim}
        Given skip-free automata \((X, h)\) and \((Y, k)\), a \emph{bisimulation} is a relation \(R \subseteq X \times Y\) such that for any \(x \mathrel{R} y\), \(\alpha \in \At\) and \(p \in \Sigma\):
        \begin{enumerate}
            \item \(x \downarrow \alpha\) if and only if \(y \downarrow \alpha\),
            \item \(x \tr{\alpha \mid p} \checkmark\) if and only if \(y \tr{\alpha \mid p} \checkmark\), and
            \item if \(x \tr{\alpha \mid p} x'\), then there is a \(y' \in Y\) such that \(x'\mathrel{R} y'\) and \(y \tr{\alpha \mid p} y'\).\footnote{Together with the first two constraints, this condition implies that if $y \tr{\alpha \mid p} y'$, then there is an $x' \in X$ such that $x' \mathrel{R} y'$ and $x \tr{\alpha \mid p} x'$.}
        \end{enumerate}
        We call \(x\) and \(y\) \emph{bisimilar} if \(x \mathrel{R} y\) for some bisimulation $R$ and write
        %$(X, h), x \bisim (Y, k), y$ or simply
        $x \bisim y$.
        %when $(X, h)$ and $(Y, k)$ are clear from context.
    \end{definition}

    In a fixed skip-free automaton \((X, h)\), we define ${\bisim} \subseteq X\times X$ as the largest bisimulation, called \emph{bisimilarity}.
    This is an equivalence relation and a bisimulation.\footnote{
        This follows directly from seeing skip-free automata as a special type of coalgebra and the fact that the functor involved preserves weak pullbacks~\cite{coalgebra}.
        In fact, coalgebra has been an indispensable tool in the production of the current paper, guiding us to the correct definitions and simplifying many of the proofs.%
        %We present more details in \Cref{app:coalgebra}.
        \label{footnote:coalg}
   }
    The bisimilarity equivalence class of a state is often called its \emph{behaviour}.

    \begin{example}
    %\begin{wrapfigure}{r}{.4\textwidth}\vspace{-40pt}
        In the automaton below, \(x_1\) and \(x_2\) are bisimilar.
        This is witnessed by the bisimulation \(\{(x_1, x_2), (x_2, x_2)\}\).
        \[
        \begin{tikzpicture}[scale=.62,font=\scriptsize]
            \node[state,minimum size=0cm] (x1) at (0,0) {\(x_1\)};
            \node[state,minimum size=0cm] (x2) at (3,0) {\(x_2\)};
            \node (final) at (1.5, -1) {\(\checkmark\)};
            \draw (x1) edge[->] node[above] {\(a \mid p\)} (x2);
            \draw (x1) edge[->] node[below left] {\(\bar a \mid q\)} (final);
            \draw[loop] (x2) edge[->, loop right] node[right] {\(a \mid p\)} (x2);
            \draw (x2) edge[->] node[below right] {\(\bar a \mid q\)} (final);
        \end{tikzpicture}
        \]%\vspace{-45pt}
        %\end{wrapfigure}
    \end{example}

    We can also use bisimulations to witness language equivalence.

    \begin{lemma}\label{lem:finer}
        Let \(e_1, e_2 \in \GExp^-\).
        If \(e_1 \bisim e_2\), then \(\L(e_1) = \L(e_2)\).
    \end{lemma}

    The converse of \cref{lem:finer} is not true.
    Consider, for example, the program \(p^{(1)}q\) that repeats the atomic action \(p \in \Sigma\) indefinitely, never reaching \(q\).
    Since
    \[\L(p^{(1)}q) = \bigcup_{n \in \mathbb N}\L(p)^n\cdot \emptyset = \emptyset = \L(0)\]
    we know that \(p^{(1)}q \sim_\L 0\).
    But \(p^{(1)}q\) and \(0\) are not bisimilar, since \cref{fig:small-step GKAT} tells us that \(p^{(1)}q \tr{\alpha \mid p} p^{(1)}q\) and \(0 \downarrow \alpha\), which together refute \cref{def:g bisim}.1.

    \subsection{Axioms}
    Next, we give an inference system for bisimilarity and language equivalence consisting of equations and equational inference rules.
    The axioms of skip-free  \GKAT are given in \cref{fig:gkat axioms}.
    They include the equation (\(\dagger\)), which says that early deadlock is the same as late deadlock.
    This is \emph{sound} with respect to the language interpretation, meaning that (\(\dagger\)) is true if \(x\) is replaced with a skip-free guarded language, but it is not sound with respect to the bisimulation semantics.
    For example, the expressions \(p \cdot 0\) and \(0\) are not bisimilar for any \(p \in \Sigma\).
    Interestingly, this is the only axiomatic difference between bisimilarity and language equivalence.

%    \begin{figure}[!t]
%        \centering
%        \begin{tabular}{c  c  c}
%            Guarded Union & Sequencing & Loops \\
%            \(\begin{aligned}
%                x &= x +_b x \\
%                x &= x +_{\At} y \\
%                x +_b y &= y +_{\bar b} x \\
%                x +_b (y +_c z) &= (x +_b y) +_{b \vee c} z
%            \end{aligned}\)&
%            \(\begin{aligned}
%                0x &= 0 \\
%                x0 &\stackrel{(\dagger)}{=} 0 \\
%                x(yz) &= (xy)z \\
%                (x +_b y)z &= xz +_b yz
%            \end{aligned}\)&
%            \(\begin{gathered}
%                x^{(b)}y = x(x^{(b)}y) +_b y \\
%                \prftree{z = xz +_b y}{z = x^{(b)}y}
%            \end{gathered}\)
%        \end{tabular}
%        \caption{\label{fig:gkat axioms} Axioms for skip-free  \GKAT.}
%    \end{figure}
%
    \begin{remark}
        The underlying logical structure of our inference systems is equational logic~\cite{birkhoff}, meaning that provable equivalence is an equivalence relation that is preserved by the algebraic operations.
    \end{remark}

    Given expressions \(e_1, e_2 \in \GExp^-\), we write \(e_1 \equiv_\dagger e_2\) and say that \(e_1\) and \(e_2\) are \emph{\(\equiv_\dagger\)-equivalent} if the equation \(e_1 = e_2\) can be derived from the axioms in \cref{fig:gkat axioms} without the axiom marked (\(\dagger\)).
    We write \(e_1 \equiv e_2\) and say that \(e_1\) and \(e_2\) are \emph{\(\equiv\)-equivalent} if \(e_1 = e_2\) can be derived from the whole set of axioms in \cref{fig:gkat axioms}.

    % \begin{proof}
    %     Both can be shown by induction on the proof of \(e_1 = e_2\).
    %     For \emph{1.}, the goal is to show that provable equivalence is a bisimulation on \(\GExp\).
    %     This is largely routine, so we just show the soundness of the Horn rule for clarity.
    %     Suppose \(R\) is a bisimulation on \((\GExp, \partial)\) containing \((g, e_1g +_b e_2)\).
    %     If \(\alpha \in b\), then we cannot have \(g \tr{\alpha \mid p} \checkmark\), as this would imply \(e_1g \tr{\alpha \mid p} \checkmark\) and \(e_1g\) is the concatenation of two skip-free  \GKAT expressions.
    %     If \(\alpha \notin b\), then \(g \tr{\alpha \mid p} \checkmark\) if and only if \(\)
    % \end{proof}

    The axioms in \cref{fig:gkat axioms} are sound with respect to the respective semantics they axiomatize.
    The only axiom that is not sound w.r.t.\ bisimilarity is $e\cdot 0\equiv 0$, as this would relate automata with different behaviours ($e$ may permit some action to be performed, and this is observable in the bisimulation).
    \begin{restatable}[Soundness]{theorem}{thmsoundness}\label{thm:soundness}
        For any \(e_1, e_2 \in \GExp^-\),
        \begin{enumerate}
            \item If \(e_1 \equiv_\dagger e_2\), then \(e_1 \bisim e_2\).
            \item If \(e_1 \equiv e_2\), then \(e_1 \sim_\L e_2\).
        \end{enumerate}
    \end{restatable}

    We consider the next two results, which are jointly converse to \cref{thm:soundness}, to be the main theorems of this paper.
    They state that the axioms in \cref{fig:gkat axioms} are \emph{complete} for bisimilarity and language equivalence respectively, i.e., they describe a complete set of program transformations for skip-free  \GKAT. % chktex 13

    \begin{restatable}[Completeness I]{theorem}{restatecompletenessone}\label{thm:completeness I}
        If \(e_1 \bisim e_2\), then \(e_1 \equiv_\dagger e_2\).
    \end{restatable}
    \begin{restatable}[Completeness II]{theorem}{restatecompletenesstwo}\label{thm:completeness II}
        If \(e_1 \sim_\L e_2\), then \(e_1 \equiv e_2\).
    \end{restatable}

    We prove \cref{thm:completeness I} in \cref{sec:bisim complete} by drawing a formal analogy between skip-free \GKAT and a recent study of regular expressions in the context of process algebra~\cite{onefreeregexlics}.
    We include a short overview of the latter in the next section.

    We delay the proof of \cref{thm:completeness II} to \cref{sec:completeness for skip-free GKAT}, which uses a separate technique based on the pruning method introduced in~\cite{gkaticalp}.

% SECTION ---------------
\section{One-free Star Expressions}\label{sec:regexp}
    Regular expressions were introduced by Kleene~\cite{kleene} as a syntax for the algebra of regular events.
    Milner offered an alternative interpretation of regular expressions~\cite{milner}, as what he called \emph{star behaviours}.
    Based on work of Salomaa~\cite{salomaa}, Milner proposed a sound axiomatization of the algebra of star behaviours, but left completeness an open problem.
    %It turned out to be difficult!
    After nearly 40 years of active research from the process algebra community, a solution was finally found by Grabmayer~\cite{regexlics}.

    A few years before this result, Grabmayer and Fokkink proved that a suitable restriction of Milner's axioms gives a complete inference system for the behaviour interpretation of a fragment of regular expressions, called the \emph{one-free fragment}~\cite{onefreeregexlics}.
    In this section, we give a quick overview of Grabmayer and Fokkink's one-free fragment~\cite{onefreeregexlics}, slightly adapted to use an alphabet that will be suitable to later use in one of the completeness proofs of skip-free \GKAT. % chktex 13
    %~\cite{bergstraetaliteration,fokkinkperpetual,fokkinkzantema,onefreeregexlics}.
    %Of particular interest to us is the algebra of \emph{one-free star expressions} that is completely axiomatized in~\cite{onefreeregexlics}.
    %One-free \textsf{StEx} consists of a constant \(0\) for deadlock, atomic actions, a binary operation \(+\) for nondeterministic branching, sequential composition, and a binary operation \(*\) interpreted as iteration in its first component sequentially composed with its second component.

    \paragraph*{\bf Syntax.}
    In the process algebra literature~\cite{milner,onefreeregexlics,regexlics}, regular expressions generated by a fixed alphabet \(A\) are called \emph{star expressions}, and denote labelled transition systems (LTSs) with labels drawn from \(A\).
    As was mentioned in \cref{sec:skipfree}, skip-free automata can be seen as certain LTSs where the labels are atomic test/atomic action pairs.
    In \cref{sec:bisim complete}, we encode skip-free \GKAT expressions as one-free regular expressions and skip-free automata as LTSs with labels drawn from \(\At \cdot \Sigma\).
    We instantiate the construction from~\cite{onefreeregexlics} of the set of star expressions generated by the label set \(\At \cdot \Sigma\).

    \begin{definition}
        The set \(\SExp\) of \emph{one-free star expressions} is given by
        \[
            \SExp \ni r_1, r_2 ::= 0 \mid \alpha p \in \At \cdot \Sigma \mid r_1 + r_2 \mid r_1r_2 \mid r_1 * r_2
        \]
     %   Notice that the appearance of the set of atomic test/atomic action pairs.
    \end{definition}

    \paragraph*{\bf Semantics.}
    The semantics of $\SExp$ is now an instance of the labelled transition systems that originally appeared in~\cite{onefreeregexlics}, with atomic test/atomic action pairs as labels and a (synthetic) output state \(\checkmark\) denoting successful termination.

    For the rest of this paper, we call a pair $(S,t)$ a \emph{labelled transition system} when $S$ is a set of states and $t\colon S \to \mathcal P(\At\cdot\Sigma \times (\checkmark + S))$ is a transition structure.
    We write \(x \tr{\alpha p} y\) if \((\alpha p, y) \in t(x)\) and \(x \tr{\alpha p} \checkmark\) if \((\alpha p, \checkmark) \in t(x)\).

    The set \(\SExp\) can be given the structure of a labelled transition system \((\SExp, \tau)\), defined in \cref{fig:small-step stex}.
    If \(r \in \SExp\), we write \(\langle r\rangle\) for the transition system obtained by restricting \(\tau\) to the one-free star expressions reachable from \(r\) and call \(\langle r \rangle\) the \emph{small-step semantics} of \(r\).

    \begin{figure}[!t]
        \begin{gather*}
            \prftree{}{\alpha p \tr{\alpha p} \checkmark}
            \qquad
            \prftree{r_1 \tr{\alpha p} r'}{r_1 + r_2 \tr{\alpha p} r'}
            \qquad
            \prftree{r_2 \tr{\alpha p} r'}{r_1 + r_2 \tr{\alpha p} r'}
            \qquad
            \prftree{r_1 \tr{\alpha p} r'}{r_1r_2 \tr{\alpha p} r'r_2}
            \\
            \prftree{r_1 \tr{\alpha p} \checkmark}{r_1r_2 \tr{\alpha p} r_2}
            \qquad
            \prftree{r_1 \tr{\alpha p} r'}{r_1 * r_2 \tr{\alpha p} r'(r_1 * r_2)}
            \qquad
            \prftree{r_1 \tr{\alpha p} \checkmark}{r_1 * r_2 \tr{\alpha p} r_1 * r_2}
            \qquad
            \prftree{r_2 \tr{\alpha p} x}{r_1 * r_2 \tr{\alpha p} x}
        \end{gather*}
        \caption{\label{fig:small-step stex}The small-step semantics of one-free star expressions.}
    \end{figure}

    The bisimulation interpretation of one-free star expressions is subtler than the bisimulation interpretation of skip-free  \GKAT expressions.
    The issue is that labelled transition systems (LTSs) are \emph{nondeterministic} in general: it is possible for an LTS to have both a \(x \tr{\alpha p} y\) and a \(x \tr{\alpha q} z\) transition for \(p \neq q\) or \(y \neq z\).
    The appropriate notion of bisimilarity for LTSs can be given as follows.

    \begin{definition}\label{def:p bisim}
        Given labelled transition systems \((S, t)\) and \((T, u)\), a \emph{bisimulation} between them is a relation \(R \subseteq S \times T\) s.t.\ for any \(x \mathrel{R} y\) and \(\alpha p \in \At\cdot \Sigma\),
        \begin{enumerate}
            \item \(x \tr{\alpha p} \checkmark\) if and only if \(y \tr{\alpha  p} \checkmark\),
            \item if \(x \tr{\alpha p} x'\), then there exists a $y'$ such that \(x' \mathrel{R} y'\) and \(y \tr{\alpha p} y'\), and
            \item if \(y \tr{\alpha p} y'\), then there exists an $x'$ such that \(x' \mathrel{R} y'\) and \(x \tr{\alpha p} x'\).
        \end{enumerate}
       As before, we denote the largest bisimulation by $\bisim$.
       That is, we call \(x\) and \(y\) \emph{bisimilar} and write \(x \bisim y\) if \(x \mathrel{R} y\) for some bisimulation $R$.
    \end{definition}

    The following closure properties of bisimulations of LTSs are useful later.
    They also imply that bisimilarity is an equivalence relation.
    Like in the skip-free case, the bisimilarity equivalence class of a state is called its \emph{behaviour}.

    \begin{lemma}\label{lem:bisim facts}
        Let \((S, t)\), \((T, u)\), and \((U, v)\) be labelled transition systems.
        Furthermore, let \(R_1, R_2 \subseteq S \times T\) and \(R_3 \subseteq T \times U\) be bisimulations.
        Then \(R_1 ^{op} = \{(y, x) \mid x \mathrel{R_1} y\}\), \(R_1 \cup R_2\) and \(R_1 \circ R_3\) are bisimulations.
      %  \end{enumerate}
    \end{lemma}

    % As is the case for \(G\)-coalgebras, general coalgebraic results imply that bisimilarity is an equivalence relation in a fixed \(P\)-coalgebra.
    % The bismilarity equivalence class of a state \(x\) is called its \emph{behaviour}.

    \paragraph*{\bf Axiomatization.}
    We follow~\cite{onefreeregexlics}, where it was shown that the axiomatization given in \cref{fig:stex axioms} is complete with respect to bisimilarity for one-free star expressions.
    Given a pair \(r_1, r_2 \in \SExp\), we write \(r_1 \equiv_* r_2\) and say that \(r_1\) and \(r_2\) are \emph{\(\equiv_*\)-equivalent} if the equation \(r_1 = r_2\) can be derived from the axioms in \cref{fig:stex axioms}.

    \begin{figure}[!t]
        \centering
        \begin{tabular}{c  c  c}
           \underline{\bf Union} & \underline{\bf Sequencing} & \underline{\bf Loops} \\
            \(\begin{aligned}
                x &= x + x \\
                x &= x + 0 \\
                x + y &= y + x \\
                x + (y + z) &= (x + y) + z
            \end{aligned}\)&
            \(\begin{aligned}
                0x &= 0 \\
                x(yz) &= (xy)z \\
                (x + y)z &= xz + yz
            \end{aligned}\)&
            \(\begin{gathered}
                x * y = x(x * y) + y \\
                \prftree{z = xz + y}{z = x * y}
            \end{gathered}\)
        \end{tabular}

        \caption{\label{fig:stex axioms} Axioms for equivalence for one-free star expressions.}
    \end{figure}

    The following result is crucial to the next section, where we prove that the axioms of \(\equiv_\dagger\) are complete with respect to bisimilarity in skip-free  \GKAT. % chktex 13

    \begin{theorem}[{\cite[Theorem.~7.1]{onefreeregexlics}}]\label{thm:grabmayer-fokkink}
        \(r_1\bisim r_2\) if and only if \(r_1 \equiv_* r_2\).
    \end{theorem}

% SECTION ---------------
\section{Completeness for Skip-free Bisimulation  \GKAT}\label{sec:bisim complete}
This section is dedicated to the proof of our first completeness result, \cref{thm:completeness I}, which says that the axioms of \cref{fig:gkat axioms} (excluding $\dagger$) are complete with respect to bisimilarity in skip-free  \GKAT. % chktex 13
Our proof strategy is a reduction of our completeness claim to the completeness result for $\SExp$ (\cref{thm:grabmayer-fokkink}).

The key objects of interest in the reduction are a pair of translations: one translation turns skip-free  \GKAT expressions into one-free star expressions and maintains bisimilarity, and the other translation turns (certain) one-free star expressions into skip-free  \GKAT expressions and maintains provable bisimilarity.

We first discuss the translation between automata and labelled transition systems, which preserves and reflects bisimilarity.
 We then introduce the syntactic translations and present the completeness proof.

\subsection{Transforming skip-free automata to labelled transition systems}
We can easily transform a skip-free automaton into an LTS by essentially turning \(\tr{\alpha \mid p}\) transitions into \(\tr{\alpha p}\) transitions. This can be formalized, as follows.

\begin{definition}
    Given a set \(X\), we define \(\grph_X : (\bot + \Sigma\times(\checkmark + X))^\At \to \mathcal P(\At\cdot\Sigma \times (\checkmark + X)) \) to be
    \(
        \grph_X(\theta) = \{(\alpha p, x) \mid \theta(\alpha) = (p, x)\}
    \).
    Given a skip-free automaton \((X, h)\), we define \(\grph_*(X, h) = {(X, \grph_X \circ h)}\)
\end{definition}

The function \(\grph_X\) is injective: as its name suggests, \(\grph_X(\theta)\) is essentially the graph of $\theta$ when viewed as a partial function from $\At$ to ${\Sigma \times (\checkmark + X)}$.
This implies that the transformation \(\grph_*\) of skip-free automata into LTSs preserves and reflects bisimilarity.

\begin{lemma}\label{lem:bisim iff bisim}
    Let \(x, y \in X\), and \((X, h)\) be a skip-free automaton.
    Then \(x \bisim y\) in \((X, h)\) if and only if \(x \bisim y\) in \(\grph_*(X, h)\).
\end{lemma}

Leading up to the proof of \cref{thm:completeness I}, we also need to undo the effect of \(\grph_*\) on skip-free automata with a transformation that takes every LTS of the form \(\grph_*(X, h)\) to its underlying skip-free automaton \((X, h)\).

The LTSs that can be written in the form \(\grph_*(X, h)\) for some skip-free automaton \((X, h)\) can be described as follows.
Call a set \(U \in  \P(\At \cdot \Sigma \times (\checkmark + X))\) \emph{graph-like} if whenever \((\alpha p, x)\in U\) and \((\alpha q, y) \in U\), then \(p = q\) and \(x = y\).
An LTS \((S, t)\) is \emph{deterministic} if \(t(s)\) is graph-like for every \(s \in S\).

\begin{restatable}{lemma}{lemmadeterministiccharacterisation}
    An LTS \((S, t)\) is deterministic if and only if \((S, t) = \grph_*(X, h)\) for some skip-free automaton \((X, h)\).
\end{restatable}

\begin{remark}
    As mentioned in \Cref{footnote:coalg}, there is a coalgebraic outlook in many of the technical details in the present paper.
    For the interested reader, \(\grph\) and the map that transforms graph-like relations into functions \(\func\) are actually natural transformations between the functors whose coalgebras correspond to skip-free automata and deterministic LTSs, and are furthermore inverse to one another.
    This implies that \(\grph_*\) and \(\func_*\) witness an isomorphism between the categories of skip-free automata and deterministic LTSs.
\end{remark}

\subsection{Translating Syntax}
We can mimic the transformation of skip-free automata into deterministic labelled transition systems and vice-versa by a pair of syntactic translations going back and forth between skip-free \textsf{GKAT} expressions and certain one-free star expressions. Similar to how only some labelled transition systems can be turned into skip-free automata, only some one-free star expressions have corresponding skip-free \GKAT expressions---the deterministic ones.

The definition of deterministic expressions requires the following notation: given a test \(b \in \BExp\), we define \(b\cdot r\) inductively on \(r \in \SExp\) as follows:
\begin{gather*}
    b\cdot 0 = 0
    \qquad
    b\cdot \alpha p =
    \begin{cases}
        \alpha p & \alpha \leq b\\
        0 & \alpha \not\leq b
    \end{cases}
    \qquad
    b\cdot (r_1 + r_2) = b\cdot r_1 + b\cdot r_2
    \\
    b \cdot (r_1r_2) = (b\cdot r_1)r_2
    \qquad
    b\cdot (r_1 * r_2) = (b\cdot r_1)(r_1*r_2) + b \cdot r_2
\end{gather*}
for any \(\alpha p \in \At\cdot \Sigma\) and \(r_1,r_2 \in \SExp\).

%{\color{red} % Comment out for final

\begin{definition}\label{def:deterministic sexp}
    The set of \emph{deterministic} one-free star expressions is the smallest subset \(\Det\subseteq \SExp\) such that (i) \(0 \in \Det\), (ii) \(\alpha p \in \Det\) for any \(\alpha \in \At\) and \(p \in \Sigma\), (iii) \(r_1r_2 \in \Det\) for any \(r_1,r_2 \in \Det\), and (iv) for any \(r_1, r_2 \in \Det\) such that for some $b \in \BExp$ we have \(r_1 \equiv_* b\cdot r_1\) and \(r_2 \equiv_* \bar b \cdot r_2\), we also have \(r_1 + r_2, r_1 *r_2 \in \Det\).
\end{definition}

%} % Comment out for final

\paragraph*{\bf From \(\GExp^-\) to $\Det$.} We can now present the translations of skip-free expressions to deterministic one-free star expressions.
\begin{definition}\label{def:gtr}
    We define the \emph{translation} function \(\gtr : \GExp^- \to \Det\) by
    \begin{gather*}
        \gtr(0) = 0
        \qquad
        \gtr(p) = \sum_{\alpha \in \At} \alpha p
        \qquad
        \gtr(e_1 +_b e_2) = b\cdot \gtr(e_1) + \bar b\cdot \gtr(e_2)
        \\
        \gtr(e_1\cdot e_2) = \gtr(e_1)\gtr(e_2)
        \qquad
        \gtr(e_1^{(b)}e_2) = (b \cdot \gtr(e_1))*(\bar b \cdot \gtr(e_2))
    \end{gather*}
    for any \(b \in \BExp\), \(p \in \Sigma\), \(e_1,e_2 \in \GExp\).
\end{definition}
\begin{remark}
    In \cref{def:gtr}, we make use of a generalized sum \(\sum_{\alpha \in \At}\).
    Technically, this requires we fix an enumeration of \(\At\) ahead of time, say \(\At = \{\alpha_1, \dots, \alpha_{n}\}\), at which point we can define \(\sum_{\alpha \in \At} r_\alpha = r_{\alpha_1} + \cdots + r_{\alpha_n}\).
    Of course, $+$ is commutative and associative up to $\equiv_*$, so the actual ordering of this sum does not matter as far as equivalence is concerned.
\end{remark}

The most prescient feature of this translation is that it respects bisimilarity.

\begin{restatable}{lemma}{lemmagtrahomom}\label{lem:gtr is good stuff}
    The graph of the translation function \(\gtr\) is a bisimulation of labelled transition systems between $\grph_*(\GExp^-, \partial)$ and $(\SExp, \tau)$.
    Consequently, \(e_1 \bisim e_2\) in $\grph_*(\GExp^-, \partial)$ if and only if \(\gtr(e_1) \bisim \gtr(e_2)\) in $(\SExp, \tau)$.
    %Also, if \(e_1 \equiv_\dagger e_2\), then \(\gtr(e_1) \equiv_* \gtr(e_2)\).
\end{restatable}

\paragraph{\bf From \(\Det\) to $\GExp^-$.} %Next we translate deterministic one-free star expressions into skip-free  \GKAT expressions.
We would now like to define a \emph{back translation} function \(\rtg : \Det \to \GExp^-\) by induction on its argument.
Looking at \cref{def:deterministic sexp}, one might be tempted to write \(\rtg(b \cdot r_1 + \bar b \cdot r_2) = \rtg(r_1) +_b \rtg(r_2)\), but the fact of the matter is that it is possible for there to be distinct \(b,c \in \BExp\) such that \(b \cdot r_1 + \bar b \cdot r_2 = c \cdot r_1 + \bar c \cdot r_2\), even when $b$ and $c$ have different atoms.

\begin{definition}
    Say that \(r_1, r_2 \in \SExp\) are \emph{separated by} \(b \in \BExp\) if \(r_1 = b\cdot r_1\) and \(r_2 = \bar b \cdot r_2\).
    If such a \(b\) exists we say that \(r_1\) and \(r_2\) are \emph{separated}.
\end{definition}

Another way to define \(\Det\) is to say that \(\Det\) is the smallest subset of \(\SExp\) containing \(0\) and \(\At\cdot\Sigma\) that is closed under sequential composition and closed under unions and stars of separated one-free star expressions.

Suppose \(r_1\) and \(r_2\) are separated by both \(b\) and \(c\).
Then one can prove that \((b \vee c)\cdot r_1 \equiv_* b\cdot r_1 + c\cdot r_1 \equiv_* r_1\) and \(\overline{(b \vee c)}\cdot r_2 = (\bar b \wedge \bar c)\cdot r_2 \equiv_* \bar b\cdot (\bar c \cdot r_2) \equiv_* r_2\), so \(r_1\) and \(r_2\) are separated by \(b \vee c\) as well.
Since there are only finitely many Boolean expressions up to equivalence, there is a maximal (weakest) test \(b(r_1,r_2) \in \BExp\) such that \(r_1\) and \(r_2\) are separated by \(b(r_1, r_2)\).

\begin{definition}\label{def:rtg}
    The \emph{back translation} \(\rtg : \Det \to \GExp^-\) is defined by
    \begin{gather*}
        \rtg(0) = 0
        \qquad
        \rtg(\alpha p) = p+_\alpha 0
        \qquad
        \rtg(r_1 + r_2) = \rtg(r_1) +_{b(r_1, r_2)} \rtg(r_2)
        \\
        \rtg(r_1r_2) = \rtg(r_1)\cdot \rtg(r_2)
        \qquad
        \rtg(r_1*r_2) = \rtg( r_1)^{(b(r_1, r_2))}\rtg(r_2)
    \end{gather*}
    for any \(r_1,r_2 \in \SExp\).
    In the union and star cases, we may use that \(r_1\) and \(r_2\) are separated (by definition of $\Det$), so that \(b(r_1, r_2)\) is well-defined.
\end{definition}

The most important property of \(\rtg\) in the completeness proof is that it preserves provable equivalence.

%{\color{red} % Comment out for final

\begin{restatable}{theorem}{lemconservativity}\label{lem:rtg is also good stuff}
    Let \(r, s \in \Det\).
    If \(r \equiv_* s\), then \(\rtg(r) \equiv_\dagger \rtg(s)\).
\end{restatable}

An intuitive approach to proving \cref{lem:rtg is also good stuff} proceeds by induction on the derivation of \(r \equiv_* s\).
However, if one-free regular expressions appear in the derivation of \(r \equiv_* s\) that are not deterministic, then the induction hypothesis cannot be applied, because the translation map \(\rtg\) is only defined on \(\Det\).
In other words, the  induction hypothesis must be strengthened, as a proof by induction on derivations can only go through if whenever \(r, s \in \Det\) and \(r \equiv_* s\), there is a derivation of \(r \equiv_* s\) in which only deterministic one-free regular expressions appear.
Such a derivation is what we call a \emph{deterministic proof}.

\begin{definition}
    Given \(r,s \in \Det\), we call a proof of \(r \equiv_* s\) a \emph{deterministic proof} if every expression that appears in the proof is a deterministic one-free regular expression (i.e., is in \(\Det\)).
    We write \(r \equiv_*^{\det} s\) if there is a deterministic proof of \(r \equiv_* s\).
\end{definition}

To proceed with the completeness proof sketched above, we have to show that the induction hypothesis is sound.

%
%As it turns out, provable equivalences between deterministic one-free regular expressions always have deterministic proofs.

\begin{restatable}{theorem}{restatedetproofthm}\label{thm:sf_gkat_fix}
    For any \(r,s \in \Det\), \(r \equiv_* s\) if and only if \(r \equiv_*^{\det} s\).
\end{restatable}

The proof of \cref{thm:sf_gkat_fix} requires an in-depth dive into the completeness proof of Grabmayer and Fokkink~\cite{onefreeregexlics}.
This can be found in \cref{app:proof of fix}.
With \cref{thm:sf_gkat_fix} in hand, \cref{lem:rtg is also good stuff} is proven by induction on the deterministic proof of \(r \equiv_*^{\det} s\).

%} % Comment out for final

The last fact needed in the proof of completeness is that, up to provable equivalence, every skip-free  \GKAT expression is equivalent to its back-translation.

\begin{restatable}{lemma}{lemprovableretract}\label{lem:provable inverse}
    For any \(e \in \GExp^-\), \(e \equiv_\dagger \rtg(\gtr(e))\).
\end{restatable}

We are now ready to prove \cref{thm:completeness I}, that $\equiv_\dagger$ is complete with respect to behavioural equivalence in skip-free  \GKAT. % chktex 13

\restatecompletenessone*
\begin{proof}
    Let \(e_1,e_2 \in \GExp^-\) be a bisimilar pair of skip-free  \GKAT expressions.
    By \cref{lem:bisim iff bisim}, \(e_1\) and \(e_2\) are bisimilar in \(\grph_*(\GExp^-, \partial)\).
    By \cref{lem:gtr is good stuff}, the translation \(\gtr : \grph_*(\GExp^-, \partial) \to (\SExp, \tau)\) preserves bisimilarity, so \(\gtr(e_1)\) and \(\gtr(e_2)\) are bisimilar in \((\SExp, \tau)\) as well.
    By \cref{thm:grabmayer-fokkink}, \(\gtr(e_1) \equiv_* \gtr(e_2)\).
    Therefore, by \cref{lem:rtg is also good stuff}, \(\rtg(\gtr(e_1)) \equiv_\dagger \rtg(\gtr(e_2))\).
    Finally, by \cref{lem:provable inverse}, we have
    \(
        e_1 \equiv_\dagger \rtg(\gtr(e_1)) \equiv_\dagger \rtg(\gtr(e_2)) \equiv_\dagger e_2
    \).
\end{proof}

% SECTION ---------------
\section{Completeness for Skip-free  \GKAT}\label{sec:completeness for skip-free GKAT}

The previous section establishes that \(\equiv_\dagger\)-equivalence coincides with bisimilarity for skip-free \GKAT expressions by reducing the completeness problem of skip-free \GKAT up to bisimilarity to a solved completeness problem, namely that of one-free star expressions up to bisimilarity.
In this section we prove a completeness result for skip-free \GKAT up to language equivalence.
We show this can be achieved by reducing it to the completeness problem of skip-free \GKAT up to bisimilarity, which we just solved in the previous section.
%We now turn our attention to

Despite bisimilarity being a less traditional equivalence in the context of Kleene algebra, this reduction simplifies the completeness proof greatly, and justifies the study of bisimilarity in the pursuit of completeness for \GKAT. % chktex 13

The axiom \(x \cdot 0 = 0\) (which is the only difference between skip-free \GKAT up to language equivalence and skip-free \GKAT up to bisimilarity) indicates that the only semantic difference between bisimilarity and language equivalence in skip-free \GKAT is early termination.
This motivates our reduction to skip-free \GKAT up to bisimilarity below, which involves reducing each skip-free expression to an expression representing only the successfully terminating branches of execution.
% Given \(e, f \in \GExp^-\) the reduction consists of finding ``early-termination versions'' of \(e\) and \(f\), call them \(\floor{e}\) and \(\floor{f}\) respectively, such that \(e\) and \(f\) are language equivalent if and only if \(\floor{e}\) and \(\floor{f}\) are bisimilar.
% Applying completeness of
%The reduction to skip-free \GKAT up to bisimilarity employs this perspective by \emph{pruning} the behaviours of states that do not end in successful termination and matching this semantic transformation with a syntactic construction.

Now let us turn to the formal proof of \Cref{thm:completeness II}, which says that if $e, f \in \GExp^-$ are such that $\L(e) = \L(f)$, then $e \equiv f$.
In a nutshell, our strategy is to produce two terms $\floor{e}, \floor{f} \in \GExp^-$ such that $e \equiv \floor{e}$, $f \equiv \floor{f}$ and $\floor{e} \bisim \floor{f}$ in \((\GExp^-, \partial)\).
The latter property tells us that $\floor{e} \equiv_\dagger \floor{f}$ by \Cref{thm:completeness I}, which together with the other properties allows us to conclude $e \equiv f$.
The expression \(\floor{e}\) can be thought of as the \emph{early termination version} of \(e\), obtained by pruning the branches of its execution that cannot end in successful termination.

To properly define the transformation \(\floor{-}\) on expressions, we need the notion of a \emph{dead} state in a skip-free automaton, analogous to a similar notion from~\cite{gkatpopl}.

\begin{definition}
Let $(X, h)$ be a skip-free automaton.
The set $D(X, h)$ is the largest subset of $X$ such for all $x \in D(X, h)$ and $\alpha \in \At$, either $h(x)(\alpha) = \bot$ or $h(x)(\alpha) \in \Sigma \times D(X,h)$.
When $x \in D(X, h)$, $x$ is \emph{dead}; otherwise, it is \emph{live}.
\end{definition}

In the sequel, we say $e \in \GExp^-$ is dead when $e$ is a dead state in $(\GExp^-, \partial)$, i.e., when $e \in D(\GExp^-, \partial)$.
Whether $e$ is dead can be determined by a simple depth-first search, since $e$ can reach only finitely many expressions by $\partial$.
The axioms of skip-free \GKAT can also tell when a skip-free expression is dead.

\begin{restatable}{lemma}{lemdeadequalszero}\label{lemma:dead-equals-zero}%
    Let $e \in \GExp$.
    If $e$ is dead, then $e \equiv 0$.
\end{restatable}

We are now ready to define $\floor{-}$, the transformation on expressions promised above.
The intuition here is to prune the dead subterms of $e$ by recursive descent; whenever we find a part that will inevitably lead to an expression that is never going to lead to acceptance, we set it to $0$.

\begin{definition}
    Let $e \in \GExp^-$ and $a \in \BExp$.
    In the sequel we use $ae$ as a shorthand for $e +_a 0$.
    We now define $\floor{e}$ inductively, as follows
    \begin{mathpar}
        \floor{0} = 0
        \and
        \floor{p} = p
        \and
        \floor{e_1 +_b e_2} = \floor{e_1} +_b \floor{e_2}
        \and
        \floor{e_1 \cdot e_2} =
            \begin{cases}
            0 & \text{$e_2$ is dead} \\
            \floor{e_1} \cdot \floor{e_2} & \text{otherwise}
            \end{cases}
        \and
        \floor{e_1 \while{b} e_2} =
            \begin{cases}
            0 & \text{$\overline{b}e_2$ is dead} \\
            \floor{e_1} \while{b} \floor{e_2} & \text{otherwise}
            \end{cases}
    \end{mathpar}
\end{definition}

The transformation defined above yields a term that is \(\equiv\)-equivalent to $e$, because $\equiv$ includes the early termination axiom $e \cdot 0 \equiv 0$.
The proof is a simple induction on $e$, using \cref{lemma:dead-equals-zero}. %the fact that if $e$ is dead, then $e \equiv 0$.

\begin{restatable}{lemma}{restateprunepreservesequiv}%
\label{lem:prune-preserves-equiv}
    For any $e \in \GExp^-$, $e \equiv \floor{e}$.
\end{restatable}

It remains to show that if $\L(e) = \L(f)$, then $\floor{e}$ and $\floor{f}$ are bisimilar.
To this end, we need to relate the language semantics of $e$ and $f$ to their behaviour.
As a first step, we note that behaviour that never leads to acceptance can be pruned from a skip-free automaton by removing transitions into dead states.

\begin{definition}
    Let $(X, h)$ be a skip-free automaton.
    Define $\floor{h}: X \to GX$ by
    \[
        \floor{h}(x)(\alpha) =
            \begin{cases}
            \bot & h(x)(\alpha) = (p, x'),\ \text{$x'$ is dead} \\
            h(x)(\alpha) & \text{otherwise}
            \end{cases}
    \]
\end{definition}

Moreover, language equivalence of two states in a skip-free automaton implies bisimilarity of those states, but only in the pruned version of that skip-free automaton.
The proof works by showing that the relation on $X$ that connects states with the same language is, in fact, a bisimulation in $(X, \floor{h})$.

\begin{restatable}{lemma}{restateprunebisimilar}%
\label{lemma:prune-bisimilar}
    Let $(X, h)$ be a skip-free automaton and $x,y \in X$.
    We have
    \[
        \L(x, (X, h)) = \L(y, (X, h))
        \implies
        x \bisim y ~\text{in}~(X, \floor{h})
    \]
\end{restatable}

The final intermediate property relates the behaviour of states in the pruned skip-free automaton of expressions to those in the syntactic skip-free automaton.

\begin{restatable}{lemma}{restateprunecommutes}%
\label{lemma:prune-commutes}
    % $\floor{-}: \GExp^- \to \GExp^-$ is a morphism of skip-free automata, from $(\GExp^-, \floor{\partial})$ to $(\GExp^-, \partial)$.
    The graph \(\{(e, \floor{e}) \mid e \in \GExp^-\}\) of \(\floor{-}\) is a bisimulation of skip-free automata between $(\GExp^-, \floor{\partial})$ and $(\GExp^-, \partial)$.
\end{restatable}

We now have all the ingredients necessary to prove \Cref{thm:completeness II}.

\restatecompletenesstwo*
\begin{proof}
% We derive as follows:
% \begin{align*}
%     e_1 \sim_\L e_2
%         &\implies \L(e, (\GExp^-, \partial)) = \L(f, (\GExp^-, \partial)) \tag{def. $\sim_\L$} \\
%         &\implies (\GExp^-, \floor{\partial}), e \bisim (\GExp^-, \floor{\partial}), f  \tag{\Cref{lemma:prune-bisimilar}}\\
%         &\implies (\GExp^-, \partial), \floor{e} \bisim (\GExp^-, \partial), \floor{f} \tag{\Cref{lemma:prune-commutes}} \\
%         &\implies \floor{e} \equiv_\dagger \floor{f} \tag{\Cref{thm:completeness I}} \\
%         &\implies e \equiv f \tag{\Cref{lemma:prune-preserves-equiv}}
% \end{align*}
If \(e_1 \sim_\L e_2\), then by definition \(\L(e_1) = \L(e_2)\).
By \cref{lemma:prune-bisimilar}, \(e_1 \bisim e_2\) in \((\GExp^-, \floor{\partial})\), which by \Cref{lemma:prune-commutes} implies that \(\floor{e_1} \bisim \floor{e_2}\) in \((\GExp^-, \partial)\).
From \Cref{thm:completeness I} we know that \(\floor{e_1} \equiv_\dagger \floor{e_2}\), and therefore \(e_1 \equiv e_2\) by \Cref{lem:prune-preserves-equiv}.
\end{proof}

% SECTION ---------------
\section{Relation to  \GKAT}%
\label{sec:relation to GKAT}

So far we have seen the technical development of skip-free \textsf{GKAT} without much reference to the original development of \GKAT as it was presented in~\cite{gkatpopl} and~\cite{gkaticalp}.
In this section, we make the case that the semantics of skip-free \GKAT is merely a simplified version of the semantics of \GKAT, and that the two agree on which expressions are equivalent after embedding skip-free \GKAT into \GKAT. % chktex 13
More precisely, we identify the bisimulation and language semantics of skip-free \GKAT given in \Cref{sec:skipfree} with instances of the existing bisimulation~\cite{gkaticalp} and language~\cite{gkatpopl} semantics of \GKAT proper.
The main takeaway is that two skip-free \GKAT expressions are equivalent in our semantics precisely when they are equivalent when interpreted as proper \GKAT expressions in the existing semantics.

\subsection{Bisimulation semantics}

To connect the bisimulation semantics of skip-free \GKAT to \GKAT at large, we start by recalling the latter.
To do this, we need to define $\GKAT$ automata.

\begin{definition}
A (\GKAT) \emph{automaton} is a pair $(X, d)$ such that $X$ is a set and $d: X \to (\bot + \checkmark + \Sigma \times X)^\At$ is a function called the \emph{transition function}.
We write \(x \tr{\alpha\mid p} y\) to denote \(d(x)(\alpha) = (p, y)\), \(x \Rightarrow \alpha\) to denote \(d(x)(\alpha) = \checkmark\), and \(x\downarrow\alpha\) if \(d(x)(\alpha) = \bot\).
\end{definition}

Automata can be equipped with their own notion of bisimulation.\footnote{As in previous sections, automata can be studied as coalgebras for a given functor and the notions below are instances of general abstract notions~\cite{gumm,coalgebra}.}
%We spell out the definition that results from coalgebra below.

\begin{definition}
Given automata $(X, h)$ and $(Y, k)$, a \emph{bisimulation} between them is a relation $R \subseteq X \times Y$ such that if $x \mathrel{R} y$, $\alpha \in \At$ and $p \in \Sigma$,:
\begin{enumerate}
    \item
    if $h(x)(\alpha) = \bot$, then $k(y)(\alpha) = \bot$; and
    \item
    if $h(x)(\alpha) = \checkmark$, then $k(y)(\alpha) = \checkmark$; and
    \item
    if $h(x)(\alpha) = (p, x')$, then $k(y)(\alpha) = (p, y')$ such that $x' \mathrel{R} y'$.
\end{enumerate}
We call $x$ and \(y\) \emph{bisimilar} and write \(x \bisim y\) if \(x \mathrel{R} y\) for some bisimulation \(R\).% to $y \in Y$, denoted $x \bisim y$, if there exists a bisimulation $R \subseteq X \times Y$ such that $x \mathrel{R} y$.
\end{definition}

\begin{remark}
The properties listed above are implications, but it is not hard to show that if all three properties hold for $R$, then so do all of their symmetric counterparts.
For instance, if $k(y)(\alpha) = (p, y')$, then certainly $h(x)(\alpha)$ must be of the form $(q, x')$, which then implies that $q = p$ while $x' \mathrel{R} y'$.
\end{remark}

Two \GKAT expressions are bisimilar when they are bisimilar as states in the \emph{syntactic automaton}~\cite{gkaticalp}, \((\GExp, \delta)\), summarized in \cref{fig:gkat syntactic automaton}.

%\begin{definition}
%We define $\delta: \GExp \to (\bot + \checkmark + \Sigma \times \GExp)^\At$ inductively:

\begin{figure}[!t]
\begin{gather*}
    \infer{\alpha \le b}
		{b \Rightarrow a}
	\quad
    \infer{\alpha\le b \quad e_1 \Rightarrow \alpha}
		{e_1 +_b e_2 \Rightarrow \alpha}
	\quad
    \infer{\alpha \le \bar b \quad e_2 \Rightarrow \alpha}
		{e_1 +_b e_2 \Rightarrow a}
        \quad
    \infer{\alpha \le b \quad e_1 \tr{\alpha \mid p} e'}
		{e_1 +_b e_2 \tr{\alpha \mid p} e'}
	\quad
	\infer{\alpha \le \bar b \quad e_2 \tr{\alpha \mid p} e'}
		{e_1 +_b e_2 \tr{\alpha \mid p} e'}
    \quad
    \\
    \infer{\ }
        {p \tr{\alpha \mid p} 1}
    \quad
	\infer{e \Rightarrow \alpha \quad e_2 \Rightarrow \alpha}
		{e_1 \cdot e_2 \Rightarrow a}
	\quad
	\infer{e \Rightarrow \alpha \quad f \tr{\alpha \mid p} e'}
		{e_1 \cdot e_2 \tr{\alpha \mid p} e'}
	\quad
	\infer{e \tr{\alpha \mid p} e'}
		{e_1 \cdot e_2 \tr{\alpha \mid p} e' \fatsemi e_2}
	\\
	\infer{\alpha \le b \quad e \tr{\alpha \mid p} e'}
		{e^{(b)} \tr{\alpha \mid p} e' \fatsemi e^{(b)}}
		\quad
	\infer{\alpha \le \bar b}
		{e^{(b)} \Rightarrow a}
\end{gather*}
\caption{\label{fig:gkat syntactic automaton}The transition function $\delta: \GExp \to (\bot + \checkmark + \Sigma \times \GExp)^\At$ defined inductively. Here, $e_1 \fatsemi e_2$ is $e_2$ when $e = 1$ and $e_1 \cdot e_2$ otherwise, \(b \in \mathsf{BExp}\), \(p \in \Sigma\), and \(e,e',e_i \in \GExp\).}
\end{figure}
% \begin{mathpar}
%     \delta(b)(\alpha) =
%         \begin{cases}
%         \checkmark & \alpha \leq b \\
%         \bot & \alpha \not\leq b
%         \end{cases}
%     \and
%     \delta(e +_b f)(\alpha) =
%         \begin{cases}
%         \delta(e)(\alpha) & \alpha \leq b \\
%         \delta(f)(\alpha) & \alpha \not\leq b
%         \end{cases}
%     \\
%     \delta(p)(\alpha) = (p, 1)
%     \and
%     \delta(e \cdot f)(\alpha) =
%         \begin{cases}
%         (p, e' \fatsemi f) & \delta(e)(\alpha) = (p, e') \\
%         \bot & \delta(e)(\alpha) = \bot \\
%         \delta(f)(\alpha) & \delta(e)(\alpha) = \checkmark
%         \end{cases}
%     \and
%     \delta(e^{(b)})(\alpha) =
%         \begin{cases}
%         (p, e' \fatsemi e^{(b)}) & \alpha \leq b,\ \delta(e)(\alpha) = (p, e') \\
%         \checkmark & \alpha \not\leq b \\
%         \bot & \text{otherwise}
%         \end{cases}
% \end{mathpar}
%Here, we define $e \fatsemi f$ to be $f$ when $e = 1$, and $e \cdot f$ otherwise.
%\end{definition}

\begin{remark}
The definition of $\delta$ given above diverges slightly from the definition in~\cite{gkaticalp}.
Fortunately, this does not make a difference in terms of the bisimulation semantics: two expressions are bisimilar in $(\GExp, \delta)$ if and only if they are bisimilar in the original semantics.
\ifarxiv%
We refer to \Cref{appendix:relation to GKAT} for a detailed account.
\else%
The full version~\cite{fullversion} contains a detailed account.
\fi%
\end{remark}

There is a fairly easy way to convert a skip-free automaton into a $\GKAT$ automaton: simply reroute all accepting transitions into a new state $\top$, that accepts immediately, and leave the other transitions the same.

\begin{definition}
Given a skip-free automaton $(X, d)$, we define the automaton $\embed(X, d) = (X + \top, \tilde{d})$, where $\tilde{d}$ is defined by
\[
    \tilde{d}(x)(\alpha) =
        \begin{cases}
        \checkmark & x = \top \\
        (p, \top) & d(x)(\alpha) = (p, \checkmark) \\
        d(x)(\alpha) & \text{otherwise}
        \end{cases}
\]
\end{definition}

We can show that two states are bisimilar in a skip-free automaton if and only if these same states are bisimilar in the corresponding $\GKAT$ automaton.

\begin{restatable}{lemma}{restateembedbisimilar}%
\label{lemma:embed-bisimilar}
Let $(X, d)$ be a skip-free automaton, and let $x, y \in X$.
\[
    x \bisim y~\text{in}~(X, d) \iff x \bisim y~\text{in}~\embed(X, d)
\]
\end{restatable}

The syntactic skip-free automaton $(\GExp^-, \partial)$ can of course be converted to a $\GKAT$ automaton in this way.
It turns out that there is a very natural way of correlating this automaton to the syntactic $\GKAT$ automaton $(\GExp, \delta)$.

\begin{restatable}{lemma}{restateembedexpmorphism}%
\label{lemma:embed-exp-morphism}
The relation $\{ (e, e) : e \in \GExp^- \} \cup \{ (\top, 1) \}$ is a bisimulation between $\embed(\GExp^-, \partial)$ and $(\GExp, \delta)$.
\end{restatable}

We now have everything to relate the bisimulation semantics of skip-free \GKAT expressions to the bisimulation semantics of \GKAT expressions at large.

\begin{lemma}%
\label{lemma:relate-bisim-semantics}
Let $e, f \in \GExp^-$.
The following holds:
\[
    e \bisim f~\text{in}~(\GExp^-, \partial)
    \iff
    e \bisim  f~\text{in}~(\GExp, \delta)
\]
\end{lemma}
\begin{proof}
We derive using \Cref{lemma:embed-bisimilar,lemma:embed-exp-morphism}, as follows: since the graph of \(\embed\) is a bisimulation, \(e \bisim f\) in \((\GExp^-, \partial)\) iff \(e \bisim f\) in \(\embed(\GExp^-, \partial)\) if and only if \(e \bisim f\) in \((\GExp, \delta)\).
% \begin{align*}
%     e \bisim f~\text{in}~(\GExp^-, \partial)
%         &\iff \embed(\GExp^-, \partial), e \bisim \embed(\GExp^-, \partial), f \\
%         &\iff (\GExp, \delta), e \bisim (\GExp, \delta), f
% \end{align*}
In the last step, we use the fact that if $R$ is a bisimulation (of automata) between $(X, h)$ and $(Y, k)$, and $S$ is a bisimulation between $(Y, k)$ and $(Z, \ell)$, then $R \circ S$ is a bisimulation between $(X, h)$ and $(Z, \ell)$ (see \cref{lem:general bisim facts}).
\end{proof}

\subsection{Language semantics}

We now recall the language semantics of \GKAT, which is defined in terms of \emph{guarded strings}~\cite{katcompleteness}, i.e., words in the set $\At \cdot (\Sigma \cdot \At)^*$, where atoms and actions alternate.
In \GKAT, successful termination occurs with a trailing associated test, representing the state of the machine at termination.
In an execution of the sequential composition of two programs \(e \cdot f\), the test trailing the execution of \(e\) needs to match up with an input test compatible with \(f\), otherwise the program crashes at the end of executing \(e\).
The following operations on languages of guarded strings record this behaviour by matching the ends of traces on the left with the beginnings of traces on the right.

\begin{definition}
For $L, K \subseteq At \cdot (\Sigma \cdot \At)^*$, define
\(
    L \diamond K = \{ w\alpha{}x : w\alpha \in L, \alpha{}x \in K \}
\) and
\(
    L^{(*)} = \bigcup_{n \in \mathbb{N}} L^{(n)}
\),
where $L^{(n)}$ is defined inductively by setting $L^{(0)} = \At$ and $L^{(n+1)} = L \diamond L^{(n)}$.
\end{definition}

The language semantics of a \GKAT expression is now defined in terms of the composition operators above, as follows.

\begin{definition}
We define $\widehat{\L}: \GExp \to \mathcal{P}(\At \cdot (\Sigma \cdot \At)^*)$ inductively, as follows:
\begin{mathpar}
    \widehat{\L}(b) = \{ \alpha \in \At \mid \alpha \leq b \}
    \and
    \widehat{\L}(p) = \{ \alpha{}p\beta \mid \alpha, \beta \in \At \}
    \and
    \widehat{\L}(e \cdot f) = \widehat{\L}(e) \diamond \widehat{\L}(f)
    \\
    \widehat{\L}(e +_b f) = \widehat{\L}(b) \diamond \widehat{\L}(e) \cup \widehat{\L}(\overline{b}) \diamond \widehat{\L}(f)
    \and
    \widehat{\L}(e^{(b)}) = (\widehat{\L}(b) \diamond \widehat{\L}(e))^{(*)} \diamond \widehat{\L}(\overline{b})
\end{mathpar}
\end{definition}

This semantics is connected to the relational semantics from \Cref{def:relational interpretation}:
\begin{theorem}[\cite{gkatpopl}]
For $e, f \in \GExp$, we have $\widehat{\L}(e) = \widehat{\L}(f)$ if and only if $\semrel{e}_\sigma = \semrel{f}_\sigma$ for all relational interpretations $\sigma$
\end{theorem}

Moreover, since skip-free \GKAT expressions are also \GKAT expressions, this means that we now have two language interpretations of the former, given by $\widehat{\L}$ and $\L$.
Fortunately, one can easily be expressed in terms of the other.

\begin{restatable}{lemma}{restatelanguagerecover}%
\label{lemma:language-recover}
For $e \in \GExp^-$, it holds that $\widehat{\L}(e) = \L(e) \cdot \At$.
\end{restatable}

As an easy consequence of the above, we find that the two semantics must identify the same skip-free \GKAT-expressions.

\begin{lemma}
For $e, f \in \GExp^-$, we have $\L(e) = \L(f)$ iff $\widehat{\L}(e) = \widehat{\L}(f)$.
\end{lemma}

By \Cref{thm:completeness II}, these properties imply that $\equiv$ also axiomatizes relational equivalence of skip-free GKAT-expressions, as a result.

\begin{corollary}
Let $e, f \in \GExp^-$, we have $e \equiv f$ if and only if $\semrel{e}_\sigma = \semrel{f}_\sigma$ for all relational interpretations $\sigma$.
\end{corollary}

\subsection{Equivalences}%
\label{section:lift-equivalences}
Finally, we relate equivalences as proved for skip-free \GKAT expressions to those provable for \GKAT expressions, showing that proofs of equivalence for skip-free \GKAT expressions can be replayed in the larger calculus, without (\textsf{UA}).

\begin{figure}[!t]
    \centering
    \begin{tabular}{c  c  c}
        \underline{\bf Guarded Union} &  \underline{\bf Sequencing} &  \underline{\bf Loops} \\
        \(\begin{aligned}
            x &= x +_b x \\
            x +_b y &= y +_{\bar b} x \\
            x +_b (y +_c z) &= (x +_b y) +_{b \vee c} z \\
            x +_b y &=  bx +_b y \\
            (x +_b y)z &= xz +_b yz
        \end{aligned}\)&
        \(\begin{aligned}
            x(yz) &= (xy)z \\
            0x &= 0 \\
            x0 &\stackrel{(\dagger)}{=} 0 \\
            1x &= x \\
            x1 &= x
        \end{aligned}\)&
        \(
        \begin{gathered}
            \begin{aligned}
            xx^{(b)} +_b 1 &= x^{(b)} \\
            (x +_a 1)^{(b)} &= (ax)^{(b)}
            \end{aligned} \\[2mm]
            \inferrule{z = xz +_b y \quad   E(x) = 0}{z = x^{(b)}y}
        \end{gathered}
        \)
    \end{tabular}
    \caption{
        Axioms for language semantics \GKAT (without the Boolean algebra axioms for tests).
        The function $E: \GExp \to \BExp$ is defined below.
        If the axiom marked ($\dagger$) is omitted, the above is conjectured to axiomatize bisimilarity.
    }%
    \label{fig:gkat-full-axioms}
\end{figure}

The axioms of \GKAT as presented in~\cite{gkatpopl,gkaticalp} are provided in \Cref{fig:gkat-full-axioms}.
We write $e \approx_\dagger f$ when $e = f$ is derivable from the axioms in \Cref{fig:gkat-full-axioms} with the exception of ($\dagger$), and $e \approx f$ when $e = f$ is derivable from the full set.

The last axiom of \GKAT is not really a single axiom, but rather an \emph{axiom scheme}, parameterized by the function $E: \GExp \to \BExp$ defined as follows:
\begin{gather*}
E(b) = b
\qquad
E(p) = 0
\qquad
E(e +_b f) = (b \wedge E(e)) \vee (\overline{b} \wedge E(f))
\\
E(e \cdot f) = E(e) \wedge E(f)
\qquad
E(e^{(b)}) = \overline{b}
\end{gather*}
The function $E$ models the analogue of Salomaa's \emph{empty word property}~\cite{salomaa}: we say $e$ is \emph{guarded} when $E(b)$ is equivalent to $0$ by to the laws of Boolean algebra.
Notice that as \GKAT expressions, skip-free \GKAT expressions are always guarded.

Since skip-free \GKAT expressions are also \GKAT expressions, we have four notions of equivalence for \GKAT expressions: as skip-free expressions or \GKAT expressions in general, either with or without ($\dagger$).
These are related as follows.

\begin{theorem}
    Let $e, f \in \GExp^-$.
    Then (1) $e \approx_\dagger f$ if and only if $e \equiv_\dagger f$, and (2) $e \approx f$ if and only if $e \equiv f$.
\end{theorem}

\begin{proof}
    For the forward direction of (1), we note that if $e \approx_\dagger f$, then $e \bisim f$ in \((\GExp, \delta)\) by \Cref{thm:soundness}.
    By \Cref{lemma:relate-bisim-semantics}, $e \bisim f$ in \((\GExp^-, \delta)\) and therefore $e \equiv_\dagger f$ by \Cref{thm:completeness I}.
    Conversely, note that any proof of $e = f$ by the axioms of \Cref{fig:gkat axioms} can be replayed using the rules from \Cref{fig:gkat-full-axioms}.
    In particular, the guardedness condition required for the last skip-free \GKAT axiom using the last \GKAT axiom is always true, because $E(g) \approx_\dagger 0$ for any $g \in \GExp^-$.

    The proof of the second claim is similar, but uses \Cref{thm:completeness I} instead.
\end{proof}

% SECTION -----------------------------
\section{Related Work}

This paper fits into a larger research program focused on understanding the logical and algebraic content of programming.
Kleene's paper introducing the algebra of regular languages~\cite{kleene} was a foundational contribution to this research program, containing an algebraic account of mechanical programming and some of its sound equational laws.
The paper also contained an interesting completeness problem: give a complete description of the equations satisfied by the algebra of regular languages.
Salomaa was the first to provide a sound and complete axiomatization of language equivalence for regular expressions~\cite{salomaa}.

The axiomatization in op.\ cit.\ included an inference rule with a side condition that prevented it from being \emph{algebraic} in the sense that the validity of an equation is not preserved when substituting letters for arbitrary regular expressions.
Nevertheless, this inspired axiomatizations of several variations and extensions of Kleene algebra~\cite{ska,gkatpopl,processesparametrised}, as well as Milner's axiomatization of the algebra of star behaviours~\cite{milner}.
The side condition introduced by Salomaa is often called the \emph{empty word property}, an early version of a concept from process theory called \emph{guardedness}\footnote{This is a different use of the word ``guarded'' than in ``guarded Kleene algebra with tests''. In the context of process theory, a recursive specification is guarded if every of its function calls occurs within the scope of an operation.} that is also fundamental to the theory of iteration~\cite{bloomesik}.

Our axiomatization of skip-free \GKAT \emph{is} algebraic due to the lack of a guardedness side-condition (it is an equational \emph{Horn theory}~\cite{horntheory}).
This is particularly desirable because it allows for an abundance of other models of the axioms.
Kozen proposed an algebraic axiomatization of Kleene algebra that is sound and complete for language equivalence~\cite{kozen}, which has become the basis for a number of axiomatizations of other Kleene algebra variants~\cite{netkat-decision,kao,ckao,pocka} including Kleene algebra with tests~\cite{katintro}.
\textsf{KAT} also has a plethora of relational models, which are desirable for reasons we hinted at in~\cref{sec:overview}.

\textsf{GKAT} is a fragment of \textsf{KAT} that was first identified in~\cite{kozentseng}.
It was later given a sound and complete axiomatization in~\cite{gkatpopl}, although the axiomatization is neither algebraic nor finite (it includes (\textsf{UA}), an axiom scheme that stands for infinitely many axioms).
It was later shown that dropping \(x\cdot 0 = 0\) (called (S3) in~\cite{gkatpopl}) from this axiomatization gives a sound and complete axiomatization of bisimilarity~\cite{gkaticalp}.
The inspiration for our pruning technique is also in~\cite{gkaticalp}, where a reduction of the language equivalence case to the bisimilarity case is discussed.

Despite the existence of an algebraic axiomatization of language equivalence in \textsf{KAT}, \textsf{GKAT} has resisted algebraic axiomatization so far.
Skip-free \textsf{GKAT} happens to be a fragment of \textsf{GKAT} in which every expression is guarded, thus eliminating the need for the side condition in \cref{fig:gkat-full-axioms} and allowing for an algebraic axiomatization.
An inequational axiomatization resembling that of \textsf{KAT} might be gleaned from the recent preprint~\cite{orderedprocesses}, but we have not investigated this carefully.
The \GKAT axioms for bisimilarity of ground terms can also likely be obtained from the small-step semantics of \GKAT using~\cite{aceto94complete,aceto11gsos,aceto11preg}, but unfortunately this does not appear to help with the larger completeness problem.

The idea of reducing one completeness problem in Kleene algebra to another is common in Kleene algebra; for instance, it is behind the completeness proof of \KAT~\cite{katcompleteness}.
Cohen also reduced \emph{weak Kleene algebra} as an axiomatization of star expressions up to simulation to \emph{monodic trees}~\cite{ernie}, whose completeness was conjectured by Takai and Furusawa~\cite{monodic}.
Grabmayer's solution to the completeness problem of regular expressions modulo bisimulation~\cite{regexlics} can also be seen as a reduction to the one-free case~\cite{onefreeregexlics}, since his \emph{crystallization} procedure produces an automaton that can be solved using the technique found in op.\ cit.
Other instances of reductions include~\cite{cohen-1994,netkat,hypotheses,pocka,kao,cka,cka-precursor,ka-top,kaequations}.
Recent work has started to study reductions and their compositionality properties~\cite{hypotheses,ckao,kah-tools}.

% SECTION -----------------------------
\section{Discussion}

We continue the study of efficient fragments of Kleene Algebra with Tests (\KAT) initiated in~\cite{gkatpopl}, where the authors introduced Guarded Kleene Algebra with Tests (\GKAT) and provided an efficient decision procedure for equivalence. They also proposed a candidate axiomatization, but left open two questions.
\begin{itemize}
\item
The first question concerned the existence of an algebraic axiomatization, which is an axiomatization that is closed under substitution---i.e., where one can prove properties about a certain program $p$ and then use $p$ as a variable in the context of a larger program, being able to substitute as needed.
This is essential to enable compositional analysis.
\item
The second question left open in~\cite{gkatpopl} was whether an axiomatization that did not require an axiom scheme was possible.
Having a completeness proof that does not require an axiom scheme to reason about mutually dependent loops is again essential for scalability: we should be able to axiomatize single loops and generalize this behaviour to multiple, potentially nested, loops.
\end{itemize}
In this paper, we identified a large fragment of \GKAT, which we call \emph{skip-free \GKAT} ($\GKAT^-$), that can be axiomatized algebraically without relying on an axiom scheme.
We show how the axiomatization works well for two types of equivalence: bisimilarity and language equivalence, by proving completeness results for both semantics. Having the two semantics is interesting from a verification point of view as it gives access to different levels of precision when analyzing program behaviour, but also enables a layered approach to the completeness proofs.

We provide a reduction of the completeness proof for language semantics to the one for bisimilarity. Moreover, the latter is connected to a recently solved~\cite{regexlics} problem proposed by Milner. This approach enabled two things: it breaks down the completeness proofs and reuses some of the techniques while also highlighting the exact difference between the two equivalences (captured by the axiom $e\cdot 0\equiv 0$ which does not hold for bisimilarity). We also showed that proofs of equivalence in skip-free \GKAT transfer without any loss to proofs of equivalence in \GKAT. % chktex 13

There are several directions for future work. The bridge between process algebra and Kleene algebra has not been exploited to its full potential. The fact that we could reuse results by Grabmayer and Fokkink~\cite{regexlics,onefreeregexlics} was a major step towards completeness. An independent proof would have been much more complex and very likely required the development of technical tools resembling those in~\cite{regexlics,onefreeregexlics}. We hope the results in this paper can be taken further and more results can be exchanged between the two communities to solve open problems.

The completeness problem for full \GKAT remains open, but our completeness results for skip-free \GKAT are encouraging.
We believe they show a path towards studying whether an algebraic axiomatization can be devised or a negative result can be proved.
A first step in exploring a completeness result would be to try extending Grabmayer's completeness result~\cite{regexlics} to a setting with output variables---this is a non-trivial exploration, but we are hopeful that it will yield new tools for completeness.
As mentioned in the introduction, \textsf{NetKAT}~\cite{netkat} (and its probabilistic variants~\cite{probnetkat,cantorscott}) have been one of the most successful extensions of \KAT. We believe the step from skip-free \GKAT to a skip-free guarded version of \textsf{NetKAT} is also a worthwhile exploration. Following~\cite{kmt}, we hope to be able to explore these extensions in a modular and parametric way. % chktex 13

\subsubsection*{Acknowledgements}{A.~Silva and T.~Schmid were partially funded by ERC grant Autoprobe (grant agreement 101002697). T.~Kapp\'{e} was supported by the EU’s Horizon 2020 research and innovation programme under Marie Sk\l{}odowska-Curie grant agreement No. 101027412 (VERLAN).}

\bibliographystyle{splncs04}
\bibliography{main.bib}

\ifarxiv%
\clearpage
\input{appendix.tex}

\fi%

{\small\medskip\noindent{\bf Open Access} This chapter is licensed under the terms of the Creative Commons Attribution 4.0 International License (\url{http://creativecommons.org/licenses/by/4.0/}), which permits use, sharing, adaptation, distribution and reproduction in any medium or format, as long as you give appropriate credit to the original author(s) and the source, provide a link to the Creative Commons license and indicate if changes were made.} % chktex 36

{\small The images or other third party material in this chapter are included in the chapter's Creative Commons license, unless indicated otherwise in a credit line to the material. If material is not included in the chapter's Creative Commons license and your intended use is not permitted by statutory regulation or exceeds the permitted use, you will need to obtain permission directly from the copyright holder.}

\medskip\noindent\includegraphics{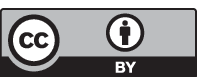} % chktex 8

\end{document}

%% file: appendix.tex
\appendix
\allowdisplaybreaks%

\section{Coalgebra}%
\label{appendices}%
\label{app:coalgebra}

In the main text of the paper, we avoided using the language of universal coalgebra~\cite{coalgebra} in the presentation to not deter from the main concepts which can be described concretely. We have however used coalgebra in our development and as mentioned in \Cref{footnote:coalg} the concrete definitions in the main text are instances of abstract notions. This is helpful in simplifying proofs and so, in the appendix, we will present the proofs of the results using coalgebra.

In this first section of the appendix, we present the coalgebraic results relevant to our studies of skip-free \GKAT, full \GKAT, and one-free star expressions as they appear in the paper.

Coalgebra makes heavy use of the language of category theory, which we assume the reader is somewhat familiar with (see~\cite{awodey} for an introduction).
For our purposes, we only need the category \(\Set\) of sets and functions, so when we refer to a \emph{functor} \(F\), we really just mean a functor \(F : \Set \to \Set\).

\begin{definition}
	Given a functor \(F\), an \emph{\(F\)-coalgebra} is a pair \((X, h)\) consisting of a set \(X\) of \emph{states} and a \emph{transition} function \(h : X \to FX\).
	An \emph{\(F\)-coalgebra homomorphism} \(f : (X, h) \to (Y, k)\) consists of a function \(f : X \to Y\) such that \(k \circ f = F(f) \circ h\), i.e.\ the diagram below commutes.
	\[\begin{tikzcd}
		X \ar[r, "f"] \ar[d, "h"] & Y \ar[d, "k"] \\
		FX \ar[r, "F(f)"] & FY
	\end{tikzcd}
	\]
	The functor \(F\) is the \emph{coalgebraic signature} of an \(F\)-coalgebra \((X, h)\).
\end{definition}

Many models of computation can be captured as \(F\)-coalgebras for some \(F\), and (universal) coalgebra provides a framework for studying them all at once~\cite{coalgebra}.
Everything that we refer to in the main text as an \emph{automaton} (of some sort) is also an \(F\)-coalgebra for some \(F\).

\begin{example}\label{tab:coalgebraic signatures}
	The coalgebraic signature of skip-free \GKAT is the functor \(G\), where for any set \(X\), any \(f : X \to Y\), any \(\theta \in GX\), and any atomic test \(\alpha \in \At\),
	\begin{gather*}
		GX = (\bot + \Sigma\times(\checkmark + X))^\At
		\qquad
		G(f)(\theta)(\alpha) = \begin{cases}
			(p, f(x)) & \theta(\alpha) = (p, x) \\
			(p, \checkmark) &\theta(\alpha) = (p, \checkmark) \\
			\bot &\theta(\alpha) = \bot
		\end{cases}
	\end{gather*}
	The coalgebraic signature of one-free star behaviours is \(P\)~\cite{starmfps}, where for any set \(X\) and any \(U \subseteq (\At \cdot \Sigma) \times (\checkmark + X)\),
	\begin{gather*}
		PX = \mathcal P((\At \cdot \Sigma)\times(\checkmark + X))
		\\
		P(f)(U) = \{(p, \checkmark) \mid (p, \checkmark) \in U\} \cup \{(p, f(x)) \mid (p, x) \in U\}
	\end{gather*}
	We denote the coalgebraic signature of \GKAT with \(H\), where
	\begin{gather*}
		HX = (\bot + \checkmark + \Sigma\times X)^\At
		\qquad
		H(f)(\psi)(\alpha) = \begin{cases}
			(p, f(x)) & \psi(\alpha) = (p, x) \\
			\checkmark &\psi(\alpha) = \checkmark \\
			\bot &\psi(\alpha) = \bot
		\end{cases}
	\end{gather*}
\end{example}

Since multiple kinds of automata appear in the paper, results from coalgebra help us avoid duplicating proofs.
For example, all the notions of bisimulation that appear in this paper are instances of a general notion of bisimulation suggested by universal coalgebra.

\begin{definition}\label{def:general bisim}
	A \emph{bisimulation} between \(F\)-coalgebras \((X, h)\) and \((Y, k)\) is a relation \(R \subseteq X \times Y\) with an \(F\)-coalgebra structure \((R, j)\) such that the projection maps \(\pi_1 : R \to X\) and \(\pi_2 : R \to Y\) are coalgebra homomorphisms~\cite{coalgebra}.
	\[\begin{tikzcd}
		X \ar[d, "h"]
		& R \ar[d, "j"] \ar[l, "{\pi_1}"'] \ar[r, "{\pi_2}"]
		& Y \ar[d,"k"] \\
		FX
		& FR \ar[l, "{F(\pi_1)}"'] \ar[r, "{F(\pi_2)}"]
		& FY
	\end{tikzcd}
	\]
\end{definition}

\begin{lemma}
	\cref{def:g bisim} is a restatement of \cref{def:general bisim} in the case of \(F = G\), \cref{def:p bisim} is a restatement of \cref{def:general bisim} in the case of \(F = P\), and the definition of a bisimulation for \GKAT automata is a restatement of \cref{def:general bisim} in the case of \(F = H\).
\end{lemma}

\begin{proof}
	We briefly sketch the case for $F = G$.
	The other two cases have been covered elsewhere~\cite{starmfps,gkatpopl}.
	Suppose \((X, h)\), \((Y, k)\) are \(G\)-coalgebras and \(R \subseteq X \times Y\).
	Then the function \(j : R \to GR\) defined
	\[
	j(x, y)(\alpha) = \begin{cases}
		(p, (x', y')) & h(x)(\alpha) = (p, x'),~k(y)(\alpha)=(p, y'),~\text{and}~x' \mathrel{R} y' \\
		(p, \checkmark) & h(x)(\alpha) = (p, \checkmark)~\text{and}~k(y)(\alpha)=(p, \checkmark) \\
		\bot & \text{otherwise}
	\end{cases}
	\]
	defines a bisimulation of \(G\)-coalgebras iff \emph{1.}-\emph{3.} of \cref{def:g bisim} are met.
\end{proof}

\(F\)-coalgebras and \(F\)-coalgebra homomorphisms form a category \(\Coalg(F)\), the properties of which can vary wildly depending on the properties of \(F\).
We follow~\cite{gumm} in describing the desired structure of \(\Coalg(G)\), \(\Coalg(P)\), and \(\Coalg(H)\) based on the coalgebraic signature.

\begin{definition}
	% Let \(D = \{f : X_i \to Y \mid i \in I\}\) be an set of functions into \(Y\) indexed by some set \(I\).
	% A \emph{cone of \(D\)} is a set of functions \(\{q_i : X \to X_i \mid i \in I\}\), written \(q : X \Rightarrow D\), such that \(f_i \circ q_i = f_j \circ q_j\) for all \(i,j \in I\).
	% A \emph{weak pullback of \(D\)} is a cone \(p : Z \Rightarrow D\) such that for any other cone \(q : X \Rightarrow D\), there is a function \(t : X \to Z\) such that \(p_i \circ t = q_i\) for all \(i \in I\).

	% Let \(D = (f : X \to Z, g : Y \to Z)\) be a pair of functions to \(Y\).
	Let \(I\) be a set and \(D = \{f_i \colon X_i \to Y \mid i \in I\}\) be a set of functions to a set \(Y\).
	A \emph{weak pullback} of \(D\) is a set \(P\) and a set of functions \(\{p_i \colon P \to X_i \mid i \in I\}\) from \(P\) such that the following two conditions are met:
	\begin{enumerate}
		\item For all \(i, j \in I\), we have \(f_i \circ p_i = f_j \circ p_j\).
		\item Suppose \(\{q_i \colon Q \to P \mid i \in I\}\) is another set of functions such that for all \(i,j \in I\), \(f_i \circ q_i = f_j \circ q_j\).
		Then there is a function \(\ell \colon Q \to P\) such that for all \(i \in I\), \(q_i = p_i \circ \ell\).
	\end{enumerate}
	Note that \(\ell\) is not necessarily unique.\footnote{This is often called a \emph{generalized} weak pullback, and the typical notion of weak pullback is \emph{binary}, i.e., specifically the \(I = \{1, 2\}\) case. Note that there are functors that preserve binary weak pullbacks that do not preserve generalized weak pullbacks~\cite{gumm}.}

	\paragraph{Weak pullback preservation.} A functor \(F\) is said to \emph{preserve weak pullbacks} if \(\{F(p_i) \colon FP \to FX_i \mid i \in I\}\) is a weak pullback of \(\{F(p_i) \colon FX_i \to FY \mid i \in I\}\) whenever \(\{p_i \colon P \to X_i \mid i \in I\}\) is a weak pullback of \(\{f_i \colon X_i \to Y\}\).
	In the binary case, weak pullback preservation looks like this:
	\[
		\begin{tikzcd}
			Q \ar[dr, dashed, "\ell"] \ar[drr, "q_1", bend left] \ar[ddr, "q_2", bend right] & \\
			& P \ar[r, "p_1"] \ar[d, "p_2"] & X \ar[d, "f_1"] \\
			& Y \ar[r, "f_2"] & Y
		\end{tikzcd}
		\quad\Rightarrow\quad
		\begin{tikzcd}
			Q \ar[dr, dashed, "\ell"] \ar[drr, "q_1", bend left] \ar[ddr, "q_2", bend right] & \\
			& FP \ar[r, "F(p_1)"] \ar[d, "F(p_2)"] & FX \ar[d, "F(f_1)"] \\
			& FY \ar[r, "F(f_2)"] & FY
		\end{tikzcd}
	\]

	% A \emph{cone of \(D\)} is a set of functions \(\{q_i : X \to X_i \mid i \in I\}\), written \(q : X \Rightarrow D\), such that \(f_i \circ q_i = f_j \circ q_j\) for all \(i,j \in I\).
	% A \emph{weak pullback of \(D\)} is a cone \(p : Z \Rightarrow D\) such that for any other cone \(q : X \Rightarrow D\), there is a function \(t : X \to Z\) such that \(p_i \circ t = q_i\) for all \(i \in I\).

\end{definition}

\begin{lemma}
	\(G\), \(P\), and \(H\) are weak pullback preserving functors.
\end{lemma}

\begin{proof}
	These are all examples of \emph{Kripke polynomial functors}, which are known to preserve weak pullbacks~\cite{gummelements}.
\end{proof}

Functors that preserve weak pullbacks are the subject of the seminal work of Rutten~\cite{coalgebra}.
Bisimulations are particularly well-behaved in this setting.

\begin{lemma}\label{lem:general bisim facts}
	Let \(F\) be a functor, and \((X, h)\) and \((Y, k)\) be \(F\)-coalgebras.
	\begin{enumerate}
		\item If \(R_i \subseteq X \times Y\) is a bisimulation for each \(i \in I\), then \(\bigcup_{i \in I} R_i\) is a bisimulation.
		\item If \(R \subseteq X \times Y\) is a bisimulation, then so is its reverse \(R^{op}\).
		\item Given a third \(F\)-coalgebra \((Z, l)\), if \(F\) preserves pullbacks and \(R \subseteq X \times Y\) and \(Q \subseteq Y \times Z\) are bisimulations, then the relational composition \(Q \circ R\) is also a bisimulation.
		\item If \(F\) preserves weak pullbacks and \(R_i \subseteq X \times Y\) is a bisimulation for each \(i \in I\), then so is their intersection \(\bigcap_{i \in I} R_i\).
	\end{enumerate}
	% \begin{enumerate}
	% 	\item
	% 	\item If \(R \subseteq X \times Y\) is a bisimulation, then so is \(R ^{op} = \{(y, x) \mid x \mathrel{R} y\}\).
	% 	\item Let \((Z, j)\) be a third \(F\)-coalgebra.
	% 	If \(R_1 \subseteq X \times Y\) and \(R_2 \subseteq Y \times Z\) are bisimulations, then so is \(R_1 \circ R_2\).
	% \end{enumerate}
\end{lemma}

\begin{proof}
	See~\cite[Theorems 5.2, 5.4, 5.5 and 6.4]{coalgebra}.
\end{proof}

Together with~\cite[Theorem 5.4]{coalgebra}, this subsumes \cref{lem:bisim facts}.
The first statement, in particular, implies that the union of all bisimulations \({\bisim} \subseteq X \times Y\) is always a bisimulation between \(F\)-coalgebras.
The first three statements together imply that if \(F\) preserves weak pullbacks, then \(\bisim\) is also an equivalence relation.

In the main text of the paper, we also talk about graphs of certain functions being bisimulations.
See, for example, \cref{lem:gtr is good stuff,lemma:prune-commutes}.
This is talk of coalgebra homomorphisms in disguise.

\begin{lemma}[{\cite[Theorem~2.5]{coalgebra}}]\label{lem:same as homom}
	Let \(F\) be a functor that preserves weak pullbacks.
	If \((X, h)\) and \((Y, k)\) are \(F\)-coalgebras and \(f : X \to Y\) is any function, then \(f\) is an \(F\)-coalgebra homomorphism if and only if
	\[
	\operatorname{Grph}(f) = \{(x, f(x)) \mid x \in X\}
	\]
	is a bisimulation.
\end{lemma}

The small-step semantics of expressions can also be given in coalgebraic terms.
Recall that the small-step semantics \(\langle e\rangle\) of a skip-free \GKAT expression is the skip-free automaton obtained by restricting transitions to the set of states reachable from \(e\).
This is otherwise known as the \emph{subcoalgebra generated by \(e\) in \((\GExp^-, \partial)\)}.
A subset \(U \subseteq X\) of an \(F\)-coalgebra \((X, h)\) is said to be \emph{closed (under transitions)} if its diagonal \(\Delta_U = \{(x, x) \mid x \in U\}\) is a bisimulation.
Equivalently, there is a (necessarily unique) \(F\)-coalgebra structure \((U, h')\) on \(U\) s.t.\ the inclusion \(U \hookrightarrow X\) is a homomorphism, i.e.\ \((U, h')\) is a \emph{subcoalgebra} of \((X, h)\).

The subcoalgebra \(\langle e\rangle\) generated by \(e\) in \((\GExp^-, \partial)\) is the smallest subcoalgebra containing \(e\).
Here, \(\langle e\rangle\) is defined explicitly, but in general, if \(F\) preserves weak pullbacks, every state of an \(F\)-coalgebra is contained in a smallest closed subset.\footnote{This follows directly from \cref{lem:general bisim facts} item 4.}

% \begin{lemma}%
% 	\label{lem:smallest subcoalgebra}
% 	Let \(F\) be a functor that preserves weak pullbacks and let \((X, h)\) be an \(F\)-coalgebra.
% 	For any \(x \in X\), then there is a \emph{smallest} closed subset of \(X\) containing \(x\).
% \end{lemma}

% \begin{proof}
% 	Let \(\mathcal U = \{U \subseteq X \mid \text{\(U\) closed in \((X, h)\) and \(x \in U\)}\}\).
% 	By \cref{lem:general bisim facts}, the set \(V = \bigcap_{U \in \mathcal{U}} U\subseteq X\) is such that \(\Delta_V = \bigcap_{U \in \mathcal U} \Delta_{U}\) is a bisimulation.
% 	Thus, \(V\) is by definition the smallest closed subset containing \(x\).
% \end{proof}

% Every \(e \in \GExp^-\) is contained in a closed subset of \((\GExp^-, \partial)\) by \cref{lem:local finiteness} (and the same goes for each \(r \in \SExp\)), so the lemma tells us that \(\langle e\rangle\) is well-defined.
% Given \(e \in \GExp^-\) (and similarly for \(r \in \SExp\)), the set of expressions \(f\) such that \(e \tr{\alpha_1\mid p_1} \cdots \tr{\alpha_n\mid p_n} f\) is the smallest closed subset containing \(e\), and its corresponding subcoalgebra is denoted \(\langle e\rangle\).
% Generally, the smallest closed subset of an \(F\)-coalgebra \((X, h)\) containing \(x \in X\) carries a unique coalgebra structure \(\langle x \rangle = (V, h')\) making it a subcoalgebra of \((X, h)\).

\paragraph*{Restructuring functors.}
The last bit of coalgebra that we make use of in the main text is naturality.
Recall that a \emph{natural transformation} between functors \(\eta : F_1 \Rightarrow F_2\) is a family of functions \(\{\eta_X : F_1X \to F_2X\}\) indexed by (the class of) all sets and such that \(F_2(f) \circ \eta_X = \eta_Y \circ F_1(f)\) for any \(f : X \to Y\).

\begin{definition}
	Let \(\eta : F_1 \Rightarrow F_2\) be a natural transformation.
	We define the \emph{restructuring} functor \(\eta_* : \Coalg(F_1) \to \Coalg(F_2)\) by setting \(\eta_*(X, h) = (X, \eta_X \circ h)\) and \(\eta_*(f) = f\) for \(f \colon (X, h) \to (Y, k)\).
\end{definition}

Restructuring functors preserve bisimilarity.

\begin{lemma}[{\cite[Theorem~15.1]{coalgebra}}]
	Let \((X, h)\) and \((Y, k)\) be \(F_1\)-coalgebras and \(\eta : F_1 \Rightarrow F_2\).
	If $(X, h), x \bisim (Y, k), y$, then $\eta_*(X, h), x \bisim \eta_*(Y, k), y$.
\end{lemma}

Here we use the notation \((X,h), x \bisim (Y, k), y\) to emphasize the two coalgebras involved in the relation \(x \bisim y\).

\begin{proof}
	Let \((R,j)\) be a coalgebra structure on \(R\) such that \(\pi_1,\pi_2\) are \(F_1\)-coalgebra homomorphisms.
	By naturality, \(\eta_*(R, j)\) is a coalgebra structure on \(R\) such that \(\pi_1, \pi_2\) are \(F_2\)-coalgebra homomorphisms.
\end{proof}

A natural transformation \(\eta : F_1 \Rightarrow F_2\) is a \emph{natural isomorphism} if there is another natural transformation \(\kappa : F_2 \Rightarrow F_1\) such that \(\eta_X \circ \kappa_X = \id_{F_2X}\) and \(\kappa_X \circ \eta_X = \id_{F_1X}\) for all \(X\).

\begin{lemma}
	Suppose \(\eta \colon F_1 \Rightarrow F_2\) is a natural isomorphism.
	Then \[\eta_* \colon \Coalg(F_1) \to \Coalg(F_2)\] is a \emph{concrete} isomorphism of categories, i.e.\ if \(U_i : \Coalg(F_i) \to \Set\) is the forgetful functor \(U_i(X, h) = X\), then \(U_2 \circ \eta_* = U_1\).
\end{lemma}

Define \(P_{\det}(X) = \{U \subseteq \checkmark + (\At \cdot \Sigma) \times X \mid U \text{ is graph-like}\}\) for any set \(X\).
Note that \(P_{\det}\) is a subfunctor of \(P\), which in this context means that $P_{\det}(X) \subseteq P(X)$ for all sets $X$.
In \cref{sec:bisim complete}, we introduced two natural transformations, \(\grph : G \Rightarrow P_{\det}\) and \(\func : P_{\det} \Rightarrow G\), and we indicated that they were inverse to one another.
This means that the category of \(G\)-coalgebras (the skip-free automata) is isomorphic to the category of \(P_{\det}\)-coalgebras (the deterministic LTSs).
This observation is the basis of \cref{thm:completeness I}, the completeness theorem of \cref{sec:bisim complete}.

In \cref{sec:completeness for skip-free GKAT}, we introduced a different kind of functor \(\embed : \Coalg(G) \to \Coalg(H)\), only this one was not quite concrete (it adds a state).
This functor is
\[
\embed(X, h) = (X + \top, \tilde h)
\]
where
\[
\tilde{d}(x)(\alpha) =
\begin{cases}
	\checkmark & x = \top \\
	(p, \top) & d(x)(\alpha) = (p, \checkmark) \\
	d(x)(\alpha) & \text{otherwise}
\end{cases}
\]
However, as we saw in \cref{lemma:embed-bisimilar}, \(\embed\) does preserve and reflect bisimilarity.

% APPENDIX -------------------------------------
\section{Proofs for \cref{sec:skipfree}}

\lemlocalfiniteness*

\begin{proof}
	By induction on \(e\).
	The automaton \(\langle 0\rangle\) has a single state with no transitions and \(\langle p \rangle\) consists of a single state that accepts all of \(\At\) after \(p\).
	Write \(|\langle e\rangle|\) for the number of expressions reachable from \(e\).
	For the inductive step, simply notice that \(|\langle e_1 +_b e_2\rangle|\), \(|\langle e_1e_2\rangle|\), and \(|\langle e_1^{(b)}e_2\rangle|\) are bounded above by \(|\langle e_1\rangle| + |\langle e_2\rangle|\).
\end{proof}

\lemlanguagecomputation*

\begin{proof}
    The proof proceeds by induction on $e$.
    In the base, since \(\langle 0\rangle\) consists of a single state with no outgoing transitions, clearly \(\L(0) = \emptyset\).
	Similarly, the only outgoing transitions of each \(p \in \Sigma\) are \(p \tr{\alpha \mid p} \checkmark\), so \(\L(p) = \{\alpha p \mid \alpha \in \At\}\).

    For the inductive step, first note that every successfully terminating path out of \(e_1 +_b e_2\) is of the form
	\(
	e_1 +_b e_2 \tr{\alpha p} g_1 \tr{\alpha_1 \mid p_1} \cdots \tr{} g_n \tr{\alpha_n \mid p_n} \checkmark
	\)
	where either \(\alpha \le b\) and
	\(
	e_1\tr{\alpha \mid p} g_1 \tr{\alpha_1 \mid p_1} \cdots \tr{} g_n \tr{\alpha_n \mid p_n} \checkmark
	\)
	or \(\alpha \le \bar b\) and
	\(
	e_2\tr{\alpha \mid p} g_1 \tr{\alpha_1 \mid p_1} \cdots \tr{} g_n \tr{\alpha_n \mid p_n} \checkmark
	\).
    This property implies the equality for $\L(e_1 +_b e_2)$.
    Furthermore, every successfully terminating path out of \(e_1e_2\) is of the form
	\(
	e_1e_2 \tr{\alpha_1\mid p_1} g_1e_2 \tr{} \cdots \tr{\alpha_n \mid p_n} e_2 \tr{\beta_1\mid q_1} \cdots \tr{\beta_m \mid q_m} \checkmark
	\)
	where
	\(
	e_1\tr{\alpha_1 \mid p_1} g_1 \tr{\alpha_1 \mid p_1} \cdots \tr{} g_n \tr{\alpha_n \mid p_n} \checkmark
	\)
	and
	\(
	e_2\tr{\beta_1 \mid q_1} h_1 \tr{\beta_2 \mid q_2} \cdots \tr{} h_n \tr{\beta_m \mid q_m} \checkmark
	\).
    This proves the equality for $\L(e_1 e_2)$.

	Lastly, we consider \(e_1^{(b)} e_2\).
	Call a cycle in a \(G\)-coalgebra \emph{minimal} if every state appears at most once in the cycle.
	Note that every cycle is a composition of minimal cycles.
	Minimal cycles containing \(e_1^{(b)} e_2\) are of the form
	\(
	e_1 ^{(b)} e_2 \tr{\alpha p} g_1(e_1 ^{(b)} e_2) \tr{\alpha_1 \mid p_1} \cdots \tr{\alpha_n\mid p_n} e_1 ^{(b)} e_2
	\)
	where \(\alpha \le b\) and
	\(
	e_1\tr{\alpha \mid p} g_1 \tr{\alpha_1 \mid p_1} \cdots \tr{} g_n \tr{\alpha_n \mid p_n} \checkmark
	\).
	Successfully terminating paths from \(e_1^{(b)}e_2\) that do not contain cycles are of the form
	\(
	e_1^{(b)}e_2 \tr{\beta \mid q} h_1 \tr{\beta_1 \mid q_1} \cdots \tr{} h_n \tr{\beta_m \mid q_m} \checkmark
	\)
	where
	\(
	e_2 \tr{\beta \mid q} h_1 \tr{\beta_1 \mid q_1} \cdots \tr{} h_n \tr{\beta_m \mid q_m} \checkmark
	\).
	Putting these together, a successfully terminating path from \(e_1^{(b)}e_2\) is a composition of minimal cycles followed by a successfully terminating path coming from \(e_2\).
	It follows that the guarded traces accepted by \(e_1^{(b)}e_2\) are those of the form \(w_1\cdots w_n u\), where each \(w_i \in \L(e_1)\) and starts with an atomic test below \(b\) and \(u \in \L(e_2)\) and starts with an atomic test below \(\bar b\).
	In symbols, \(\L(e_1^{(b)}e_2) = \bigcup_{n \in \mathbb N} (b \L(e_1))^n\cdot \bar b \L(e_2)\).
\end{proof}

\thmsoundness*

\begin{proof}
	We begin by sketching the proof that \(\equiv_\dagger\) is a bisimulation on \((\GExp^-, \partial)\).
	Let \(e_1 \equiv_\dagger e_2\).
	We show that \((e_1, e_2)\) satisfies \emph{1.}-\emph{3.} of \cref{def:g bisim} by induction on the proof of \(e_1 = e_2\).
	In the base, we need to consider the equational axioms (including those of equational logic).
	\emph{1.}-\emph{3.} of \cref{def:g bisim} are clearly reflexive and symmetric, so it suffices to consider the equational axioms listed in \cref{fig:gkat axioms}.
	\begin{itemize}
		\item Consider \((e, e +_b e)\).
		If \(e \tr{\alpha \mid p} e'\), then since either \(\alpha \le b\) or \(\alpha \le \bar b\) for each atomic test \(\alpha\), \(e +_b e \tr{\alpha \mid p} e'\), and vice versa.
		Since \(\equiv_\dagger\) is an equivalence relation, \(e' \equiv_\dagger e'\).
		Similarly, \(e \downarrow \alpha\) if and only if \(e +_b e\downarrow \alpha\) and \(e \tr{\alpha \mid p} \checkmark\) if and only if \(e +_b e \tr{\alpha \mid b} \checkmark\).
		\item In the \((e, e +_1 f)\), \((e +_b f, f +_{\bar b} e)\), \((e +_b (f +_c g), (e +_b f) +_{b\vee c} g)\), \((0, 0e)\), \(((e +_b f)g, eg +_b fg)\), and \((e^{(b)}f, e\cdot (e^{(b)}f) +_{b} f)\) cases, the two expressions have all the same outgoing transitions. See the previous case.
		\item In the \((e\cdot (f\cdot g), (e\cdot f)\cdot g)\) case, let \(\alpha \in \At\).
		There are a few subcases:
		\begin{itemize}
			\item[(i)] Since \(e \downarrow \alpha\) if and only if \(e\cdot f \downarrow \alpha\), it follows that \(e \cdot (f \cdot g) \downarrow \alpha\) if and only if \((e\cdot f) \cdot g \downarrow \alpha\). That is, \emph{1.} in \cref{def:g bisim} is satisfied by this \(\alpha\).
			\item[(ii)] If \(e \tr{\alpha \mid p}\checkmark\), then \(e \cdot f \tr{\alpha \mid p} f\). Therefore, \(e\cdot (f \cdot g) \tr{\alpha\mid p} f \cdot g\) and \((e \cdot f)\cdot g \tr{\alpha \mid p} f \cdot g\).
            This establishes \emph{2.} of \cref{def:g bisim}.
			\item[(iii)] If \(e \tr{\alpha \mid p} e'\), then \(e\cdot (f\cdot g) \tr{\alpha \mid p} e'\cdot (f\cdot g)\). We also know that \(e\cdot f \tr{\alpha \mid p} e'\cdot f\), so \((e\cdot f)\cdot g\tr{\alpha \mid p} (e'\cdot f)\cdot g\). Since \(e'\cdot (f\cdot g) \equiv_\dagger (e'\cdot f)\cdot g\), in conjunction with (ii) we see that \emph{3.} of \cref{def:g bisim} is satisfied for this \(\alpha\).
		\end{itemize}
	\end{itemize}
	In the induction step, we consider the Horn rules.
	We skip the transitivity case, since the conditions in \cref{def:g bisim} are clearly transitive.
	\begin{itemize}
		\item We are going to show that the congruence rules preserve the properties in \cref{def:g bisim}.
		Suppose \(e_1 \equiv_\dagger e_2\) and \((e_1, e_2)\) satisfies the conditions of \cref{def:g bisim}.
		\begin{itemize}
			\item It suffices to check \(\alpha \le b\) in the \((e_1 +_b g, e_2 +_b g)\) subcase. These expressions have the same pairs of outgoing transitions out of \((e_1, e_2)\) for \(\alpha \le b\). The induction hypothesis concludes this case.
			\item In the \((e_1\cdot g, e_2\cdot g)\) subcase, \(e_1 \downarrow \alpha\) if and only if \(e_2 \downarrow \alpha\) by the induction hypothesis, and \(e_i \cdot g \downarrow \alpha\) if and only if \(e_i \downarrow \alpha\). Similarly, \(e_1 \tr{\alpha \mid p} \checkmark\) if and only if \(e_2 \tr{\alpha \mid p} \checkmark\), so \(e_1\cdot g \tr{\alpha \mid p} g\) if and only if \(e_2 \cdot g \tr{\alpha \mid p} g\). If \(e_1' \equiv_\dagger e_2'\), then \(e_1 \tr{\alpha \mid p} e_1'\) if and only if \(e_2 \tr{\alpha \mid p} e_2'\). It follows that \(e_1\cdot g \tr{\alpha \mid p} e_1'\cdot g\) if and only if \(e_2\cdot g \tr{\alpha \mid p} e_2'\cdot g\).
			\item The case \((g\cdot e_1, g \cdot e_2)\) is easy because the \((g, g)\) case was covered above.
			\item Consider \((e_1^{(b)}g, e_2^{(b)}g)\).
			We only cover the \(\alpha \le b\) case here because \(e_1^{(b)}g\) and \(e_2^{(b)}g\) have the same outgoing transitions for \(\alpha \le \bar b\). The condition \emph{1.} of \cref{def:g bisim} follows from the fact that \(e_i ^{(b)} g \downarrow \alpha\) if and only if \(e_i \downarrow \alpha\). Similarly, \(e_i ^{(b)} g \tr{\alpha \mid p} e_i ^{(b)} g\) if and only if \(e_i \tr{\alpha \mid p} \checkmark\) for \(i=1,2\). Finally, if \(e_i \tr{\alpha \mid p} e_i'\) for \(i=1,2\) and \(e_1' \equiv_\dagger e_2'\), then \(e_1'\cdot (e_1 ^{(b)} g) \equiv_\dagger e_2'\cdot (e_2 ^{(b)} g)\) and \(e_i ^{(b)} g \tr{\alpha\mid p} e_i' \cdot (e_i ^{(b)} g)\) for \(i = 1,2\).
            \item Consider $(g^{(b)}e_1, g^{(b)}e_2)$.
            If $\alpha \leq b$, then the transitions out of $g^{(b)}e_1$ and $g^{(b)}e_2$ must either both reject, or go to expressions again related by $\equiv_\dagger$.
            Otherwise, if $\alpha \leq b$, then the $\alpha$-transitions of $g^{(b)}e_1$ and $g^{(b)}e_2$ are determined by $e_1$ and $e_2$ respectively, which are equivalent by assumption; the claim then follows.
		\end{itemize}
		\item Let \(g \equiv_\dagger eg +_b f\) and assume for an induction hypothesis that \((g, eg +_b f)\) satisfies \emph{1.}-\emph{3.} in \cref{def:g bisim}.
		Then \(g \downarrow \alpha\) if and only if either \(\alpha \le b\) and \(e \downarrow \alpha\) or \(\alpha \le \bar b\) and \(f \downarrow \alpha\).
		Since the latter crashes are the same as those of \(e^{(b)}f \downarrow \alpha\), we have verified \emph{1.} in \cref{def:g bisim}.

		We also know that \(g \tr{\alpha \mid p} \checkmark\) if and only if \(\alpha \le \bar b\) and \(f \tr{\alpha \mid p} \checkmark\).
		The latter conditions are equivalent to \(e^{(b)}f \tr{\alpha \mid p} \checkmark\), thus satisfying \emph{2.} in \cref{def:g bisim}.

        For \emph{3.} from \cref{def:g bisim}, let \(g \tr{\alpha \mid p} g'\).
        We should show that \(e^{(b)}f\) takes a similarly labeled transition to an expression equivalent to \(g'\).
        By induction, \(eg +_b f \tr{\alpha \mid p} e'\) such that \(g' \equiv_\dagger e'\).
        This gives us three cases to consider.
		\begin{itemize}
			\item
            If \(\alpha \le b\) and \(e \tr{\alpha \mid p} e''\), then \(e' = e''g\), and thus \(g' \equiv_\dagger e'' g\).
            Since \(g \equiv_\dagger e^{(b)}f\), it follows that \(g' \equiv_\dagger e''(e^{(b)}f)\).
            Given that \(e^{(b)} f \tr{\alpha \mid p} e''(e^{(b)}f)\), this was exactly what we needed to prove.
			\item
            If \(\alpha \le b\) and \(e \tr{\alpha \mid p} \checkmark\), then \(e' = g\).
			But then \(g' \equiv_\dagger g \equiv_\dagger e^{(b)}f\), which suffices since \(e^{(b)}f \tr{\alpha \mid p} e^{(b)}f\) in this case.
			\item
            If \(\alpha \le \bar b\) and \(f \tr{\alpha \mid p} f'\), then \(e' = f'\).
            But then \(g' \equiv_\dagger f'\), which suffices because \(e^{(b)}f \tr{\alpha \mid p} f'\) in this case.
		\end{itemize}
	\end{itemize}

	Now we verify that \(\equiv\) is sound with respect to language equivalence.
	Again, we proceed to show that \(e_1 \equiv e_2\) implies \(e_1 \sim_\L e_2\) by induction on the proof of \(e_1 = e_2\) from the rules in \cref{fig:gkat axioms}.
	We have already seen that if \(e_1 \equiv_\dagger e_2\) implies \(e_1 \bisim e_2\), so by \cref{lem:finer} we know that \(e_1 \equiv_\dagger e_2\) implies \(e_1 \sim_\L e_2\).
	This handles most of the base case.
	Reflexivity, symmetry, and transitivity are handled by the fact that \(\sim_\L\) is the kernel of \(\L\).
	This leaves us with \(e\cdot 0 \sim_\L 0\), which can be seen from \cref{lem:language}:
	\[
	\L(e\cdot 0) = \L(e) \cdot \L(0) = \L(e) \cdot \emptyset = \emptyset = \L(0)
	\]
	In the inductive step, we need to consider the transitivity, congruence, and Horn rules.
	Congruence is a consequence of \cref{lem:language}.
	Transitivity follows from the fact that \(\sim_\L\) is the kernel of \(\L\).
	Soundness of the Horn rule \(z = xz +_b y \implies z = x^{(b)}y\) follows from \cref{lem:language} and the fact that for any languages \(U, L, W \subseteq (A \cdot \Sigma)^+\),
	\[
	L = UL \cup W \implies L = U^*W  \tag{\(\heartsuit\)}
	\]
	Indeed, suppose \(\L(g) = {\L(e\cdot g +_b f)}\).
	Then because \(\L(e\cdot g +_b f) = b\L(e)\cdot \L(g) \cup \bar b\L(f)\), \(\L(g) = (b\L(e))^*\cdot \bar b\L(f) = \L(e^{(b)} f)\) by (\(\heartsuit\)).
\end{proof}

We end this section of the appendix by stating a number of provable equivalences that are useful in other sections.
The following lemma is necessary in our proof of \Cref{lem:prune-preserves-equiv}, for example.
Recall that for $e \in \GExp^-$ and $b \in \BExp$, we use $be$ as a shorthand for $e +_b 0$.

\begin{lemma}%
	\label{lemma:useful-equalities}
	Let $e, f \in \GExp^-$ and $b, c \in \BExp$.
	The following hold:
	\begin{mathpar}
		e +_b \overline{b} f \equiv_\dagger e +_b f
		\and
		be +_b f \equiv_\dagger e +_b f
		\and
		(b \wedge c)e \equiv_\dagger b(ce)
		\and
		b(e +_c f) \equiv_\dagger be +_c bf
		\and
		b(e +_c f) \equiv_\dagger b(e +_{b \wedge c} f)
		\and
		b(e \cdot f) \equiv_\dagger (be) \cdot f
		\and
		e \while{b} 0 \equiv 0
		\and
		e \while{b} f \equiv_\dagger e \while{b} (\overline{b}f)
		\and
		(be) \while{b} f \equiv_\dagger e\while{b} f
	\end{mathpar}
\end{lemma}
\begin{proof}
	We will also refer to the equivalence \((x +_b y) +_c z \equiv_\dagger x +_{b\wedge c} (y +_c z)\) below as (\textsf{G2},\textsf{G3}), as it can be derived as follows:
	\[
		(x +_b y) +_c z
		\stackrel{\text{\tiny(\textsf{G2})}}{\equiv_\dagger} z +_{\bar c} (y +_{\bar b} x)
		\stackrel{\text{\tiny(\textsf{G3})}}{\equiv_\dagger}  (z +_{\bar c} y) +_{\bar c \vee \bar b} x
		\stackrel{\text{\tiny(BA)}}{\equiv_\dagger}  (z +_{\bar c} y) +_{\overline{b\wedge c}} x
		\stackrel{\text{\tiny(\textsf{G2})}}{\equiv_\dagger}  x +_{b\wedge c} (y +_{c} z)
	\]
	We prove the desired equalities in order of appearance.
	\begin{itemize}
		\item
		For $e +_b \overline{b}f \equiv e +_b f$, we derive:
		\begin{align*}
			e +_b \overline{b}f
			&= e +_b (f +_{\overline{b}} 0) \\
			&\equiv_\dagger (e +_b f) +_{b \vee \overline{b}} 0 \tag{\textsf{G3}}\\
			&\equiv_\dagger (e +_b f) +_1 0 \tag{BA}\\
			&\equiv_\dagger e +_b f \tag{G0}
			\intertext{
				\item
				For $be +_b f \equiv e +_b f$, we derive:
			}
			be +_b f
			&\equiv_\dagger f +_{\overline{b}} be \tag{\textsf{G2}}\\
			&\equiv_\dagger f +_{\overline{b}} e \tag{prev.}\\
			&\equiv_\dagger e +_b f \tag{\textsf{G2}}
			\intertext{
                \item
				For $(b \wedge c)e \equiv b(ce)$, we derive:
			}
			(b \wedge c)e
			&\equiv_\dagger (c \wedge b)e \tag{BA}\\
			&= e +_{c \wedge b} 0 \\
			&\equiv_\dagger e +_{c \wedge b} (0 +_b 0) \tag{\textsf{G1}}\\
			&\equiv_\dagger (e +_c 0) +_b 0 \tag{\textsf{G3}}\\
			&= b(ce)
			\intertext{
                \item
				For $b(e +_c f) \equiv be +_c bf$, we derive:
			}
			b(e +_c f)
			&\equiv_\dagger ((b \vee \overline{c}) \wedge (b \vee c))(e +_c f) \tag{BA}\\
			&\equiv_\dagger (b \vee \overline{c})((b \vee c)(e +_c f)) \tag{prev.}\\
			&= (b \vee \overline{c})((e +_c f) +_{b \vee c} 0) \\
			&\equiv_\dagger (b \vee \overline{c})(e +_c (f +_b 0)) \tag{\textsf{G2},\textsf{G3}}\\
			&= (b \vee \overline{c})(e +_c bf) \\
			&= (e +_c bf) +_{b \vee \overline{c}} 0 \\
			&\equiv_\dagger (bf +_{\overline{c}} e) +_{b \vee \overline{c}} 0 \tag{\textsf{G2}}\\
			&\equiv_\dagger bf +_{\overline{c}} (e +_b 0) \tag{\textsf{G2},\textsf{G3}} \\
			&= bf +_{\overline{c}} be \\
			&\equiv_\dagger be +_{c} bf \tag{\textsf{G2}}
		\end{align*}

		\item For \(b(e +_c f) \equiv_\dagger b(e +_{b \wedge c} f)\),
		\begin{align*}
			b(e +_c f)
			&= (e +_c f) +_b 0 \\
			&\equiv_\dagger e +_{b\wedge c} (f +_b 0) \tag{\textsf{G2},\textsf{G3}}\\
			&\equiv_\dagger (b\wedge c)e +_{b\wedge c} bf \tag{prev.}\\
			&\equiv_\dagger b(ce) +_{b\wedge c} bf \tag{prev.}\\
			&\equiv_\dagger b(ce +_{b\wedge c} f) \tag{prev.}\\
			&\equiv_\dagger b((b\wedge c)e +_{b\wedge c} f) \tag{prev., BA}\\
			&\equiv_\dagger b(e +_{b\wedge c} f) \tag{prev.}
		\end{align*}

		\item
		For $b(e \cdot f) \equiv (be) \cdot f$, we derive
		\[
		b(e \cdot f)
		= e \cdot f +_b 0
		\stackrel{\text{\tiny(\textsf{G6})}}{\equiv_\dagger} e \cdot f +_b 0 \cdot f
		\stackrel{\text{\tiny(\textsf{G8})}}{\equiv_\dagger} (e +_b 0) \cdot f
		= (be) \cdot f
		\]

		\item
		For $e \while{b} 0 \equiv 0$, it suffices to prove the following, by (\textsf{RSP}):
		\[
		0
		\stackrel{\text{\tiny(\textsf{G1})}}{\equiv} 0 +_b 0
		\stackrel{\text{\tiny(\(\dagger\))}}{\equiv} e \cdot 0 +_b 0
		\]

		\item
		For $e \while{b} f \equiv e \while{b} \overline{b}f$, it suffices to prove the following, by (\textsf{RSP}):
		\[
		e \while{b} \overline{b}f
		\stackrel{\text{\tiny(\textsf{FP})}}{\equiv_\dagger} e \cdot e \while{b} \overline{b}f +_b \overline{b}f
		\stackrel{\text{\tiny(prev.)}}{\equiv_\dagger} e \cdot e \while{b} \overline{b}f +_b f
		\]

		\item
		For \((be) \while b f \equiv e \while b f\),
		\begin{align*}
			(be) \while b f
			&\equiv_\dagger (be)\cdot (be)\while b f +_b f \tag{\textsf{FP}}\\
			&\equiv_\dagger b(e\cdot (be)\while b f) +_b f \tag{prev.}\\
			&\equiv_\dagger e\cdot (be)\while b f +_b f \tag{prev.} \\
			&\equiv_\dagger e \while b f \tag{\textsf{RSP}}
		\end{align*}
	\end{itemize}
\end{proof}

\section{Proofs for \cref{sec:bisim complete}}

\lemmadeterministiccharacterisation*

\begin{proof}
	For any set \(X\), let \(P_{\det}X = \{U \in \P(\At \cdot \Sigma \times (\checkmark + X)) \mid U~\text{is graph-like}\}\) and define \(\func_X : P_{\det}X \to (\bot + \Sigma \times (\checkmark + X))^\At\) by
 \[
	 \func_X(U)(\alpha) =
        \begin{cases}
        (p, x) & (\alpha p, x) \in U \\
        \bot & \text{otherwise}
        \end{cases}
 \]
 Let \((S, t)\) be a deterministic LTS\@. Then \(\func_*(S, t) = (S, \func_S \circ t)\) is a skip-free automaton such that \(\grph_*\func_*(S, t) = (S, t)\), because \(\grph_S\circ \func_S = \id_{P_{\det} (S)}\).
\end{proof}

\lemmagtrahomom*

\begin{proof}
	Recall that for any function \(f \colon X \to Y\), \(P(f)\colon P(X) \to P(Y)\) is given by
	\begin{gather*}
		P(f)(\{(\alpha_i p_i, x_i) \mid i \in I\} \cup \{(\alpha_j p_j, \checkmark) \mid j \in J\})
		\\=
		\{(\alpha_i p_i, f(x_i)) \mid i \in I\}
		\cup \{(\alpha_j p_j, \checkmark) \mid j \in J\}
	\end{gather*}
	for any index sets \(I\) and \(J\).
    By \Cref{lem:same as homom} it suffices to show that $\gtr$ is a $P$-coalgebra homomorphism from \(\grph_*(\GExp^-, \partial)\) to $(\SExp, \tau)$, i.e., that \(\tau \circ \gtr(e) = P(\gtr) \circ \grph_{\GExp^-} \circ \partial(e)\).
    We prove this by induction on \(e \in \GExp^-\).
	\begin{itemize}
		\item Since \(\gtr(0) = 0\), \[
			\tau \circ \gtr(0) = \tau(0) = \emptyset = P(\gtr)\circ \grph_{\GExp^-} \circ \partial(0)
		\]
		\item Now let \(p \in \Sigma\).
		\begin{align*}
			\tau \circ \gtr(p)
			&= \tau\Big(\sum_{\alpha \in \At} \alpha p\Big) \\
			&= \{(\alpha p, \checkmark) \mid \alpha \in \At\} \\
			&= P(\gtr) (\{(\alpha p, \checkmark) \mid \alpha \in \At\}) \\
			&= P(\gtr) \circ \grph_{\GExp^-} \circ \partial(p)
		\end{align*}
		since \(\grph_{\GExp^-} \circ \partial(p) = \{(\alpha p, \checkmark) \mid \alpha \in \At\}\).
	\end{itemize}
	This covers the base case.
	For the induction step, assume that \(\tau\circ\gtr(e_i) = P(\gtr)\circ\grph_{\GExp^-}\circ\partial(e_i)\) for \(i = 1, 2\).
	\begin{itemize}
		\item Consider \(e_1 +_b e_2\).
		Observe that
		\begin{equation}
			\label{graph eq 1}
			\begin{aligned}
				\grph_{\GExp^-}\circ\partial(e_1 +_b e_2)
				&= \{(\alpha p, \xi) \mid \alpha \le b, e_1 \tr{\alpha \mid p} \xi\} \\&\hspace{4em}\cup
				\{(\alpha p, \xi) \mid \alpha \le \bar b, e_2 \tr{\alpha \mid p} \xi\}
			\end{aligned}
		\end{equation}
		Now, \(\gtr(e_1 +_b e_2) = b \cdot \gtr(e_1) + \bar b \cdot \gtr(e_2)\), and by the induction hypothesis, \(\gtr(e_i) \tr{\alpha p} r \in \SExp\) if and only if there is a \(g \in \GExp^-\) such that \(e_i \tr{\alpha \mid p} g\) and \(r = \gtr(g)\).
		This implies that \(b \cdot \gtr(e_1) \tr{\alpha p} r\) if and only if \(\alpha \le b\) and \(\gtr(e_i) \tr{\alpha p} \gtr(g)\) where \(e_i \tr{\alpha \mid p} g\).
		Therefore,
		\begin{align*}
			&\tau\circ\gtr(e_1 +_b e_2) \\
			&= \tau(b \cdot \gtr(e_1) + \bar b \cdot \gtr(e_2)) \\
			&= \{(\alpha p, \gtr(g)) \mid \alpha \le b, e_1 \tr{\alpha\mid p} g\} \cup \{(\alpha p, \checkmark) \mid \alpha \le b, e_1 \tr{\alpha \mid p} \checkmark\} \tag{IH}\\
			&\hspace*{4em}\cup \{(\alpha p, \gtr(g)) \mid \alpha \le \bar b, e_2 \tr{\alpha \mid p} g\} \cup \{(\alpha p, \checkmark) \mid \alpha \le \bar b, e_2 \tr{\alpha \mid p} \checkmark\} \\
			&= P(\gtr)(\{(\alpha p, g) \mid \alpha \le b, e_1 \tr{\alpha\mid p} g\} \cup \{(\alpha p, \checkmark) \mid \alpha \le b, e_1 \tr{\alpha \mid p} \checkmark\} \\
			&\hspace*{4em}\cup \{(\alpha p, g) \mid \alpha \le \bar b, e_2 \tr{\alpha \mid p} g\} \cup \{(\alpha p, \checkmark) \mid \alpha \le \bar b, e_2 \tr{\alpha \mid p} \checkmark\}) \\
			&= P(\gtr) \circ \grph_{\GExp^-} \circ \partial(e_1 +_b e_2) \tag{by~\eqref{graph eq 1}}
		\end{align*}

		\item Now consider \(e_1e_2\).
		Recall \(\gtr(e_1e_2) = \gtr(e_1)\gtr(e_2)\).
		If \(\gtr(e_1)\gtr(e_2) \tr{\alpha p} r \in \SExp\), then either \(r = r' \gtr(e_2)\) such that \(\gtr(e_1) \tr{\alpha p} r'\), or \(\gtr(e_1) \tr{\alpha p} \checkmark\) and \(r = \gtr(e_2)\).
		By the induction hypothesis, either there is a \(g \in \GExp^-\) such that \(e_1 \tr{\alpha \mid p} g\) and \(r' = \gtr(g)\), or \(e_1 \tr{\alpha \mid p} \checkmark\).
		Therefore,
		\begin{align*}
			&\tau \circ \gtr(e_1e_2) \\
			&= \tau(\gtr(e_1)\gtr(e_2)) \\
			&= \{(\alpha p, r \gtr(e_2)) \mid \gtr(e_1) \tr{\alpha p} r\} \cup \{(\alpha p, \gtr(e_2)) \mid \gtr(e_1) \tr{\alpha p} \checkmark\} \\
			&= \{(\alpha p, \gtr(g) \gtr(e_2)) \mid e_1 \tr{\alpha \mid p} g\} \cup \{(\alpha p, \gtr(e_2)) \mid e_1 \tr{\alpha \mid p} \checkmark\} \tag{IH} \\
			&= P(\gtr)(\{(\alpha p, g e_2) \mid e_1 \tr{\alpha \mid p} g\} \cup \{(\alpha p, e_2) \mid e_1 \tr{\alpha \mid p} \checkmark\}) \\
			&= P(\gtr)\circ\grph_{\GExp^-}\circ\partial(e_1e_2)
		\end{align*}

		\item Finally, consider \(e_1\while{b}e_2\).
		Since \(\gtr(e_1\while{b}e_2) = (b\cdot \gtr(e_1))*(\bar b \cdot \gtr(e_2))\), if \(\gtr(e_1\while{b}e_2) \tr{\alpha p} r \in \SExp\), there are three possibilities: either (1) \(\alpha \le b\) and \(\gtr(e_1) \tr{\alpha p} \checkmark\) and \(\gtr(e_1\while{b}e_2)\tr{\alpha p} \gtr(e_1\while{b}e_2)\); (2) \(\alpha \le b\) and \(\gtr(e_1) \tr{\alpha p} s\) and \(\gtr(e_1\while{b}e_2) \tr{\alpha p} s(b \cdot \gtr(e_1))*(\bar b \cdot \gtr(e_2))\); or (3) \(\alpha \le \bar b\) and \(\gtr(e_2) \tr{\alpha p} r\).
		In case (1), by the induction hypothesis, \(e_1 \tr{\alpha \mid p} \checkmark\).
		In case (2), by the induction hypothesis, there is a \(g \in \GExp^-\) such that \(e_1 \tr{\alpha \mid p} g\) and \(s = \gtr(g)\).
		This would make \(r = \gtr(g(e_1\while{b}e_2))\).
		In case (3), by the induction hypothesis, there is a \(g \in \GExp^-\) such that \(e_2 \tr{\alpha \mid p} g\) and \(r = \gtr(g)\).
		Therefore,
		\begin{align*}
			\tau \circ \gtr(e_1\while{b}e_2)
			&= \tau((b \cdot \gtr(e_1))*(\bar b \cdot \gtr(e_2))) \\
			&= \{(\alpha p, \gtr(e_1 \while b e_2)) \mid \alpha\le b, \gtr(e_1) \tr{\alpha p} \checkmark\} \\
			&\hspace*{4em}\cup \{(\alpha p, r\gtr(e_1 \while b e_2)) \mid \alpha\le b, \gtr(e_1) \tr{\alpha p} r\} \\
			&\hspace*{4em}\cup \{(\alpha p, r) \mid \alpha\le \bar b, \gtr(e_2) \tr{\alpha p} r\} \\
			&= \{(\alpha p, \gtr(e_1 \while b e_2)) \mid \alpha\le b, e_1 \tr{\alpha \mid p} \checkmark\}\tag{IH} \\
			&\hspace*{4em}\cup \{(\alpha p, \gtr(g)\gtr(e_1 \while b e_2)) \mid \alpha\le b, e_1 \tr{\alpha \mid p} g\} \\
			&\hspace*{4em}\cup \{(\alpha p, \gtr(g)) \mid \alpha\le \bar b, e_2 \tr{\alpha \mid p} g\} \\
			&= \{(\alpha p, \gtr(e_1 \while b e_2)) \mid \alpha\le b, e_1 \tr{\alpha \mid p} \checkmark\} \\
			&\hspace*{4em}\cup \{(\alpha p, \gtr(g(e_1 \while b e_2))) \mid \alpha\le b, e_1 \tr{\alpha \mid p} g\} \\
			&\hspace*{4em}\cup \{(\alpha p, \gtr(g)) \mid \alpha\le \bar b, e_2 \tr{\alpha \mid p} g\} \\
			&= P(\gtr)( \{(\alpha p, e_1 \while b e_2) \mid \alpha\le b, e_1 \tr{\alpha \mid p} \checkmark\}\\
			&\hspace*{4em}\cup \{(\alpha p, g(e_1 \while b e_2)) \mid \alpha\le b, e_1 \tr{\alpha \mid p} g\} \\
			&\hspace*{4em}\cup \{(\alpha p, g) \mid \alpha\le \bar b, e_2 \tr{\alpha \mid p} g\}) \\
			&= P(\gtr) \circ \grph_{\GExp^-} \circ \partial(e_1\while b e_2)
		\end{align*}
	\end{itemize}
	This concludes the proof.
\end{proof}

Given \(r, s \in \SExp\), write \(r \perp_b s\) if \(b\) is a test separating \(r\) from \(s\).

\begin{lemma}\label{lem:arbitrary seperation}
	Let \(b\) be any test.
	If \(r \perp_b s\), then \(\rtg(r + s) \equiv_\dagger \rtg(r) +_b \rtg(s)\) and \(\rtg(r * s) \equiv_\dagger \rtg(r) ^{(b)} \rtg(s)\).
\end{lemma}

\begin{proof}
	We are going to begin by showing that \(\rtg(r) \equiv_\dagger b\rtg(r)\) when \(r \equiv_* b\cdot r\), by induction on \(r \in \Det\).
	\begin{itemize}
		\item
		In the first base case, \(r = 0\).
		Of course, \(b \cdot 0 = 0\) for any \(b \in \BExp\).
		By assumption, \(b\rtg(r) = b0 \equiv_\dagger 0 = \rtg(0)\).

		\item
		For the second base case, consider $\alpha p$ for some \(\alpha \in \At\) and \(p \in \Sigma\).
		If \(\alpha \le \bar b\), then \(b\cdot \alpha p = 0 \not\equiv_* \alpha p\), so we can rule out this possibility by soundness.
		Otherwise, \(\alpha p \equiv_* b\cdot \alpha p\) and
		\begin{align*}
			b\rtg(\alpha p)
			&= b(p +_\alpha  0) \\
			&\equiv_\dagger bp +_\alpha  b0 \tag{\cref{lemma:useful-equalities}}\\
			&\equiv_\dagger p+_\alpha b0 \tag{\cref{lemma:useful-equalities}} \\
			&\equiv_\dagger p+_\alpha 0 \tag{\textsf{G1}} \\
			&= \rtg(\alpha p)
		\end{align*}

		\item
		For the first inductive case, let \(r_1 \perp_d r_2\) and \(b\cdot (r_1 + r_2) \equiv_* r_1 + r_2\).
		Then
		\begin{align*}
			(b\wedge d)\cdot (r_1 + r_2)
			&\equiv_* b \cdot (d \cdot r_1 + d \cdot r_2) \\
			&\equiv_* b \cdot (d \cdot r_1 + d \cdot (\overline{d} \cdot r_2)) \\
			&\equiv_* b \cdot (d \cdot r_1 + (d \wedge \overline{d}) \cdot r_2) \\
			&\equiv_* b \cdot (d \cdot r_1 + 0 \cdot r_2) \\
			&\equiv_* b \cdot (d \cdot r_1 + 0) \\
			&\equiv_* b \cdot (d \cdot r_1) \\
			&\equiv_* b \cdot r_1 \\
			&\equiv_* r_1
		\end{align*}
		and similarly \(r_2 \equiv_* (b \wedge \bar d)\cdot (r_1 + r_2)\).
		Hence \(r_i \equiv_* b \cdot r_i\) for \(i\in\{1,2\}\), so
		\begin{align*}
			b\rtg(r_1 + r_2)
			&\equiv_\dagger b(\rtg(r_1) +_{b(r_1, r_2)} \rtg(r_2)) \\
			&\equiv_\dagger b\rtg(r_1) +_{b(r_1,r_2)} b\rtg(r_2) \tag{\cref{lemma:useful-equalities}} \\
			&\equiv_\dagger \rtg(r_1) +_{b(r_1, r_2)} \rtg(r_2) \tag{ind.~hyp.} \\
			&= \rtg(r_1 + r_2)
		\end{align*}

		\item For the sequential composition case, if \(r_1r_2 \equiv_*b\cdot (r_1 r_2)\), then we first argue that \(r_1 \equiv_* b\cdot r_1\).
		To see this, we first make two claims.
        \begin{itemize}
            \item
            \(r_1 \tr{\alpha \mid p} \checkmark\) if and only if \(b \cdot r_1 \tr{\alpha \mid p} \checkmark\).
            For the direction from left to right, note that if \(r_1 \tr{\alpha \mid p} \checkmark\), then \(r_1r_2 \tr{\alpha \mid p} r_2\).
            But then, by the assumption that \(r_1r_2 \equiv_* b \cdot (r_1r_2) = (b \cdot r_1)r_2\) and soundness, we have that \((b \cdot r_1)r_2 \tr{\alpha \mid p} r_2'\) for some \(r_2'\).
            If \(r_2' = r_2\) and \(b \cdot r_1 \tr{\alpha \mid p} \checkmark\), then we are done.
            Otherwise, if \(r_2' = r_1' \cdot r_2\) such that \(b \cdot r_1 \tr{\alpha \mid p} r_1'\), then a simple inductive proof shows that \(r_1 \tr{\alpha \mid p} r_1'\) as well --- we can therefore rule out this case.
            Conversely, if \(b \cdot r_1 \tr{\alpha \mid p} \checkmark\), then \(r_1 \tr{\alpha \mid p} \checkmark\), by another induction on \(r_1\).
            \item
            If \(r_1 \tr{\alpha \mid p} r_1'\), then \(b \cdot r_1 \tr{\alpha \mid p} r_1''\) such that \(r_1' \equiv_* r_1''\); conversely, if \(b \cdot r_1 \tr{\alpha \mid p} r_1'\), then \(r_1 \tr{\alpha\mid p} r_1'\).
            The first property follows by an argument similar to the previous claim; the second can be shown by induction on \(r_1\).
        \end{itemize}

		We can now use the \emph{fundamental theorem} for one-free regular expressions~\cite[Proposition 2.9]{regexlics}, which states that for any \(r \in \SExp\),
		\begin{equation}\label{eq:fundamental theorem}
			r \equiv_* \sum_{r \tr{\alpha p} \checkmark} \alpha p + \sum_{r \tr{\alpha p} s} \alpha p~s
		\end{equation}
		This allows us to derive that
		\begin{align*}
			r_1
			&\equiv_* \sum_{r_1 \tr{\alpha p} \checkmark} \alpha p + \sum_{r_1 \tr{\alpha p} s} \alpha p s \tag{fund.~thm.} \\
			&\equiv_* \sum_{b\cdot r_1 \tr{\alpha p} \checkmark} \alpha p + \sum_{b\cdot r_1 \tr{\alpha p} s} \alpha p s \tag{observations above} \\
			&\equiv_* b\cdot r_1 \tag{fund.~thm.}
		\end{align*}
		Hence, \(r_1 \equiv_*b \cdot r_1\).
		This in turn allows us to deduce
		\begin{align*}
			\rtg(r_1r_2)
			= \rtg(r_1)\rtg(r_2)
			&\equiv_\dagger (b\rtg(r_1))\rtg(r_2) \tag{ind.~hyp.}\\
			&\equiv_\dagger b(\rtg(r_1)\rtg(r_2)) \tag{\cref{lemma:useful-equalities}} \\
			&= b\rtg(r_1r_2)
		\end{align*}

		\item For the star case, assume \(b \cdot (r_1*r_2) \equiv_* r_1 * r_2\) for \(i = 1,2\) and that \(r_1\) and \(r_2\) are separated by \(d\).
		Notice that this entails
		\[
			r_1 (r_1*r_2) + r_2 \equiv_* r_1 * r_2 \equiv_* b\cdot (r_1 * r_2) \equiv_* b \cdot (r_1 (r_1*r_2) + r_2)
		\]
		It follows that \(b \cdot r_2 \equiv_* r_2\) and \(b \cdot (r_1(r_1 * r_2)) \equiv_* r_1(r_1*r_2)\), like in the sum case.
		It then follows that \(b \cdot r_1 \equiv_* r_1\), by the same reasoning in the sequential composition case.
		The induction hypothesis tells us \(\rtg(r_i) \equiv_\dagger b\rtg(r_i)\) for \(i \in \{1,2\}\).
		So,
		\begin{align*}
			\rtg(r_1*r_2)
			&\equiv_\dagger \rtg(r_1)\rtg(r_1 * r_2) +_{b(r_1,r_2)} \rtg(r_2) \\
			&\equiv_\dagger (b\rtg(r_1))\rtg(r_1 * r_2) +_{b(r_1,r_2)} b\rtg(r_2) \tag{ind.~hyp.}\\
			&\equiv_\dagger b(\rtg(r_1)\rtg(r_1 * r_2)) +_{b(r_1,r_2)} b\rtg(r_2) \tag{\cref{lemma:useful-equalities}}\\
			&\equiv_\dagger b(\rtg(r_1)\rtg(r_1 * r_2) +_{b(r_1,r_2)} \rtg(r_2)) \tag{\cref{lemma:useful-equalities}}\\
			&\equiv_\dagger b\rtg(r_1 * r_2)
		\end{align*}
	\end{itemize}
	This concludes the proof of the intermediate property, that \(\rtg(r) \equiv_\dagger b\rtg(r)\) when \(r \equiv_* b\cdot r\).

	Now let \(c = b(r, s)\), the maximal test separating \(r\) from \(s\).
	The key insight here is that \(b \le c\), so that \(x +_c y \equiv_{\dagger} (x +_b z) +_c y\).
	With this observation in hand, we can further calculate
	\begin{align*}
		\rtg(r + s)
		&= \rtg(r) +_c \rtg(s)  \\
		&\equiv_\dagger b\rtg(r) +_c \rtg(s) \tag{inter.~prop.} \\
		&= (\rtg(r) +_b 0) +_c \rtg(s)\\
		&\equiv_\dagger \rtg(r) +_b (0 +_c \rtg(s)) \tag{\textsf{G3} and $b \leq c$}\\ % chktex 1
		&\equiv_\dagger \rtg(r) +_b \bar c\rtg(s) \tag{\textsf{G2}}\\
		&\equiv_\dagger \rtg(r) +_b \rtg(s) \tag{inter.~prop.}
	\end{align*}
	as well as
	\begin{align*}
		\rtg(r * s)
		&= \rtg(r)^{(c)} \rtg(s)  \\
		&\equiv_\dagger \rtg(r)\rtg(r)^{(c)}\rtg(s) +_c \rtg(s)  \tag{\textsf{FP}}\\
		&\equiv_\dagger (b\rtg(r))(\rtg(r)^{(c)} \rtg(s)) +_c \rtg(s) \tag{inter.~prop.}\\
		&\equiv_\dagger b(\rtg(r)\rtg(r)^{(c)} \rtg(s)) +_c \rtg(s) \tag{\cref{lemma:useful-equalities}}\\
		&= (\rtg(r)\rtg(r)^{(c)} \rtg(s) +_b 0) +_c \rtg(s) \\
		&\equiv_\dagger \rtg(r)(\rtg(r)^{(c)} \rtg(s)) +_b (0 +_c \rtg(s))\tag{\textsf{G3} and $b \leq c$}\\ % chktex 1
		&\equiv_\dagger \rtg(r)(\rtg(r)^{(c)} \rtg(s)) +_b \bar c\rtg(s) \tag{\textsf{G2}}\\
		&\equiv_\dagger \rtg(r)(\rtg(r)^{(c)} \rtg(s)) +_b \rtg(s) \tag{inter.~prop.}\\
		&= \rtg(r)(\rtg(r * s)) +_b \rtg(s)\\
		&\equiv_\dagger \rtg(r) \while b \rtg(s) \tag{\textsf{RSP}}
	\end{align*}
	This concludes the proof overall.
\end{proof}

The next result is the key step to proving \cref{lem:rtg is also good stuff}, which establishes the most important property of \(\rtg\) used in the completeness proof for skip-free bisimulation \(\GKAT\).

\restatedetproofthm*

The proof of this theorem is the subject of \cref{app:proof of fix}.

\lemconservativity*

\begin{proof}
	Let \(r, s \in \Det\), and suppose \(r \equiv_* s\).
	Then, by \cref{thm:sf_gkat_fix}, \(r \equiv_*^{\det} s\).
	We proceed by induction on the deterministic derivation of \(r \equiv_*^{\det} s\).
	In the base case, we verify the one-free star behaviour axioms directly.
	\begin{description}
		\item[(Idem.)] Every \(r \in \Det\) such that \(r + r \in \Det\) satisfies \(r \equiv_* 0\), for if $r \perp_b r$ then
        \[
            r
                \equiv_* b \cdot (\overline{b} \cdot r)
                \equiv_* (b \wedge \overline{b}) \cdot r
                \equiv_* 0 \cdot r
                \equiv_* 0
        \]
        so by (G0) we find
		\(\rtg(r + r)
		= \rtg(r) +_1 \rtg(r)
		\equiv_\dagger \rtg(r)\).

		\item[(\(+\)-Zero)] The maximal test separating \(r\) from \(0\) is \(1\), so \(\rtg(r + 0) = \rtg(r) +_1 \rtg(0) \equiv_\dagger \rtg(r)\) by (G0).

		\item[(Comm.)] If \(r \perp_b s\), then
		\(
		\rtg(r + s)
		\equiv_\dagger \rtg(r) +_b \rtg(s)
		\equiv_\dagger \rtg(s) +_{\bar b} \rtg(r)
		\equiv_\dagger \rtg(s + r)
		\)
		by (\textsf{G2}) and \cref{lem:arbitrary seperation} (because \(s \perp_{\bar b} r\)).

		\item[(\(+\)-Asso.)] Letting \(r \perp_b s\) and \((r + s) \perp_c t\), we have \(r \perp_{b\wedge c} t\) and \(s \perp_c t\) as well.
		Using \cref{lem:arbitrary seperation}, we derive
		\begin{align*}
			\rtg((r + s) + t)
			&\equiv_\dagger (\rtg(r) +_b \rtg(s)) +_c \rtg(t) \\
			&\equiv_\dagger \rtg(r) +_{b\wedge c} (\rtg(s) +_c \rtg(t))  \tag{\cref{lemma:useful-equalities}}\\
			&\equiv_\dagger \rtg(r) +_{b\wedge c} \rtg(s + t) \\
			&\equiv_\dagger \rtg(r + (s + t))
		\end{align*}

		\item[(L-Zero)] In the left zero case we find that by (\textsf{G6}),
		\[ \rtg(0r) = \rtg(0)\rtg(r) = 0\rtg(r) \equiv_\dagger 0 = \rtg(0)\]

		\item[(\(\cdot\)-Asso.)] No further considerations in the associativity case:
		\begin{align*}
			\rtg(r(st))
			&= \rtg(r)\rtg(st)   \\
			&= \rtg(r)(\rtg(s)\rtg(t))   \\
			&\equiv_\dagger (\rtg(r)\rtg(s))\rtg(t) \tag{G7}\\
			&= \rtg((rs)t)
		\end{align*}

		\item[(L-Dist.)] If \(r \perp_b s\), then \(rt\perp_b st\), so
		\begin{align*}
			\rtg((r + s)t)
			&= \rtg(r + s)\rtg(t)  \\
			&\equiv_\dagger (\rtg(r) +_b \rtg(s))\rtg(t) \tag{\cref{lem:arbitrary seperation}}\\
			&\equiv_\dagger \rtg(r)\rtg(t) +_b \rtg(s)\rtg(t) \tag{\textsf{G8}} \\
			&\equiv_\dagger \rtg(rt + st) \tag{\cref{lem:arbitrary seperation}}
		\end{align*}

		\item[(\textsf{FP})] Suppose \(r \perp_b s\).
		Then \(r(r * s) \perp_b s\) as well, so
		\begin{align*}
			\rtg(r * s)
			&\equiv_\dagger \rtg(r) \while{b} \rtg(s) \tag{\cref{lem:arbitrary seperation}}\\
			&\equiv_\dagger \rtg(r) (\rtg(r) \while{b} \rtg(s)) +_b \rtg(s) \\
			&\equiv_\dagger \rtg(r) \rtg(r*s) +_b \rtg(s) \tag{\cref{lem:arbitrary seperation}}\\
			&= \rtg(r(r*s)) +_b \rtg(s) \\
			&\equiv_\dagger \rtg(r(r*s) + s) \tag{\cref{lem:arbitrary seperation}}
		\end{align*}
	\end{description}
	For the inductive step, assume \(r_i \equiv_*^{\det} s_i\) implies \(\rtg(r_i) \equiv_\dagger \rtg(s_i)\) for \(i = 1,2\).
	We consider the \(+\) and \(*\) congruence rules (the \(\cdot\) congruence rule is trivial), recursive specification rule (\textsf{RSP}), symmetry (Sym) and transitivity (Tra).
	\begin{description}
		\item[(\(+\))] Suppose the deterministic proof ends with
		\[
		\prftree{r_1 = s_1 \qquad r_2 = s_2}{r_1 + r_2 = s_1 + s_2}
		\]
		and assume that \(\rtg(r_i) \equiv_\dagger \rtg(s_i)\).
		We are also assuming that \(r_1 + r_2\) and \(s_1 + s_2\) are deterministic, so we must have \(b_r,b_s \in \BExp\) such that \(r_1 \perp_{b_r} r_2\), and \(s_1 \perp_{b_s} s_2\).
		However, since \(r_1 \equiv_* s_1\) and \(r_2 \equiv_* s_2\), \(r_1 \perp_b s_1\) if and only if \(r_2 \perp_b s_2\), so we may as well take \(b_r = b_s = b\).
		We have
		\begin{align*}
			\rtg(r_1 + r_2)
			\equiv_\dagger \rtg(r_1) +_b \rtg(r_2)
			\equiv_\dagger \rtg(s_1) +_b \rtg(s_2)
			\equiv_\dagger \rtg(s_1 + s_2)
		\end{align*}

		\item[(\(*\))] Suppose the deterministic proof ends with
		\[
		\prftree{r_1 = s_1 \qquad r_2 = s_2}{r_1 * r_2 = s_1 * s_2}
		\]
		and assume that \(\rtg(r_i) \equiv_\dagger \rtg(s_i)\), \(i = 1,2\).
		Again, we obtain a common $b \in \BExp$ such that \(r_1 \perp_b r_2\) and \(s_1 \perp_b s_2\).
		Using \cref{lem:arbitrary seperation}, we have
		\begin{align*}
			\rtg(r_1 * r_2)
			\equiv_\dagger \rtg(r_1)\while{b} \rtg(r_2)
			\equiv_\dagger \rtg(s_1)\while{b} \rtg(s_2)
			\equiv_\dagger \rtg(s_1 * s_2)
		\end{align*}

		\item[(\textsf{RSP})] Suppose the deterministic proof ends with the rule
		\(
		t = rt + s \Rightarrow t = r * s
		\)
		and assume that \(\rtg(t) \equiv_\dagger \rtg(rt + s)\).
		Since we are also assuming that \(r*s \in \Det\), there is a \(b \in \BExp\) such that \(r \perp_b s\).
		We also have \(rt \perp_b s\), so
		\(
		\rtg(t) \equiv_\dagger \rtg(r)\rtg(t) +_b \rtg(s)
		\) by \cref{lem:arbitrary seperation}.
		It follows from (\textsf{RSP}) that
		\begin{align*}
			\rtg(t) &\equiv_\dagger \rtg(r)\while{b}\rtg(s) \equiv_\dagger \rtg(r * s)
		\end{align*}

		\item[(Sym.)] Suppose the deterministic proof ends with
		\[
		\prftree{r_1 = r_2}{r_2 = r_1}
		\]
		Then by symmetry of \(\equiv_\dagger\) and the induction hypothesis, \(\rtg(r_2) \equiv_\dagger \rtg(r_1)\).

		\item[(Tra.)] Suppose the deterministic proof ends with
		\[
		\prftree{r_1 = s \qquad s = r_2}{r_1 = r_2}
		\]
		Then by assumption, \(r_1 \equiv_*^{\det} s\) and \(s \equiv_*^{\det} r_2\).
		By the induction hypothesis, \(\rtg(r_1) \equiv_\dagger \rtg(s) \equiv_\dagger \rtg(r_2)\).
		This implies that \(\rtg(r_1) \equiv_\dagger \rtg(r_2)\) by transitivity of \(\equiv_\dagger\).
	\end{description}
\end{proof}

%} % Comment out for final

\noindent
We will need the following lemma in the sequel.

\begin{lemma}\label{lem:gkat proofs to regex}
	If \(e \equiv_\dagger f\), then \(\gtr(e) \equiv_* \gtr(f)\).
\end{lemma}

\begin{proof}
	By induction on the proof of \(e \equiv_* f\).
	In the base case, we consider the skip-free \GKAT axioms directly.
	\begin{itemize}
		\item
		For $e +_b e \equiv_\dagger e$, we derive
		\[\gtr(e +_b e) = b\cdot\gtr(e) + \bar b\cdot\gtr(e) \equiv_* 1\cdot \gtr(e) \equiv_* \gtr(e)\]
		The second equivalence follows from the fact that \(b \cdot r + c \cdot r \equiv_* (b \vee c) \cdot r\) for any \(b,c \in \BExp\) and any one-free regular expression \(r\), which can be shown via a straightforward induction on \(r\).
		\item
		For $e +_1 f \equiv_\dagger e$, we derive
		\[\gtr(e +_1 f) \equiv_* 1\cdot \gtr(e) + 0\cdot \gtr(f) \equiv_* \gtr(e)\]
		\item
		For $e +_b f \equiv_\dagger f +_{\bar b} e$, we derive
		\begin{align*}
			\gtr(e +_b f)
			&= b\cdot\gtr(e) + \bar b\cdot\gtr(f) \\
			&\equiv_* \bar b\cdot\gtr(f) + b\cdot\gtr(e) \\
			&= \gtr(f +_{\bar b} e)
			\intertext{
				\item
				For $(e +_b f)g \equiv_\dagger eg +_b fg$, we derive:
			}
			\gtr((e +_b f)g)
			&= (b\cdot\gtr(e) + \bar b\cdot\gtr(f))\gtr(g) \\
			&\equiv_* b\cdot\gtr(e)\gtr(g) + \bar b\cdot\gtr(f)\gtr(g) \\
			&= \gtr(eg +_b fg)
			\intertext{
				\item
				For $e(fg) \equiv_\dagger (ef)g$, we derive:
			}
			\gtr(e(fg))
			&= \gtr(e)(\gtr(f)\gtr(g)) \\
			&\equiv_* (\gtr(e)\gtr(f))\gtr(g) \\
			&= \gtr((ef)g)
			\intertext{
				\item
				For $0 \cdot e \equiv_\dagger 0$, we derive
				\[\gtr(0e) = \gtr(0)\gtr(e) = 0 \cdot \gtr(e) \equiv_* 0\]
				\item $e \while{b}f \equiv_\dagger e(e\while{b}f) +_b f$ requires some calculation:
			}
			\gtr(e^{(b)}f)
			&= (b\cdot\gtr(e)) *(\bar b\cdot\gtr(f))\\
			&\equiv_* (b\cdot\gtr(e))((b\cdot\gtr(e)) *(\bar b\cdot\gtr(f))) + \bar b\cdot\gtr(f) \\
			&\equiv_* b\cdot(\gtr(e)((b\cdot\gtr(e)) *(\bar b\cdot\gtr(f)))) + \bar b\cdot\gtr(f) \\
			&= b\cdot\gtr(e(e ^{(b)} f)) + \bar b\cdot\gtr(f) \\
			&= \gtr(e(e ^{(b)} f) +_b f)
		\end{align*}
	\end{itemize}
	For the inductive step, we need to consider the congruence properties and the Horn rule; the cases for transitivity and symmetry are trivial.
	\begin{itemize}
		\item[(\(+_b\))] Suppose the proof ends with
		\[
		\prftree{e_1 = f_1 \qquad e_2 = f_2}{e_1 +_b f_1 = e_2 +_b f_2}
		\]
		and assume that \(\gtr(e_i) \equiv_* \gtr(f_i)\).
		Then
		\begin{align*}
			\gtr(e_1 +_b f_1)
			&= b\cdot\gtr(e_1) + \bar b\cdot\gtr(f_1) \\
			&\equiv_* b\cdot\gtr(e_2) + \bar b\cdot\gtr(f_2) \\
			&= \gtr(e_2 +_b f_2)
		\end{align*}

		\item[(\(^{(b)}\))] Suppose the proof ends with
		\[
		\prftree{e_1 = f_2 \qquad e_2 = f_2}{e_1 ^{(b)} f_1 = e_2 ^{(b)} f_2}
		\]
		and assume that \(\gtr(e_i) \equiv_* \gtr(f_i)\).
		Then
		\begin{align*}
			\gtr(e_1 ^{(b)} f_1)
			&= b\cdot\gtr(e_1) * \bar b\cdot\gtr(f_1) \\
			&\equiv_* b\cdot\gtr(e_2) * \bar b\cdot\gtr(f_2) \\
			&= \gtr(e_2 ^{(b)} f_2)
		\end{align*}
		\item Suppose the proof ends with
		\[
		\prftree{g = eg +_b f}{g = e ^{(b)} f}
		\]
		and assume that \(\gtr(g) \equiv_* \gtr(eg +_b f)\).
		Then
		\[
		\gtr(g) \equiv_* (b\cdot\gtr(e))\gtr(g) + \bar b\cdot\gtr(f)
		\]
		By the Horn rule in skip-free \textsf{RegEx},
		\[
		\gtr(g) \equiv_* (b\cdot\gtr(e)) * (\bar b\cdot\gtr(f)) = \gtr(e^{(b)} f)
		\]
	\end{itemize}
\end{proof}

We are now ready to prove the key lemma in our first completeness theorem.

\lemprovableretract*

\begin{proof}
	We begin by proving the following intermediary result: for any \(e \in \GExp^-\) and \(b \in \BExp\), \(\rtg(b\cdot r) \equiv_\dagger b\rtg(r) = \rtg(r) +_b 0\) (*).
	We proceed by induction on the one-free regular expression \(r\).
	\begin{itemize}
		\item In the first base case, we have \(\rtg(b\cdot 0) = 0 \equiv_\dagger 0 +_b 0 = b0 = b\rtg(0)\).
		\item In the second base case, if \(\alpha \le b\),
		\begin{align*}
			\rtg(b\cdot \alpha p) &= \rtg(\alpha p)  \\
							&= p +_\alpha 0 \\
							&\equiv_\dagger p +_\alpha (0 +_b 0) \tag{\textsf{G1}}\\
							&\equiv_\dagger (p +_\alpha 0) +_{b \vee \alpha} 0 \tag{\textsf{G3}}\\
							&\equiv_\dagger (p +_\alpha 0) +_b 0 \tag{BA}\\
							&= \rtg(\alpha p) +_b 0 \\
							&= b\rtg(\alpha p)
		\end{align*}
		If \(\alpha \le \bar b\),
		\begin{align*}
			\rtg(b\cdot \alpha p) &= 0 \\
						&\equiv_\dagger b0 \tag{\textsf{G1}}\\
						&\equiv_\dagger b(p +_\alpha 0) \tag{\cref{lemma:useful-equalities}}\\
						&= b\rtg(\alpha p)
		\end{align*}
		\item In the first inductive step, let \(r \perp_c s\).
		Since \(b\cdot r \perp_{c\wedge b} b \cdot s\),
		\begin{align*}
			\rtg(b\cdot (r + s)) &\equiv_\dagger \rtg(b\cdot r + b\cdot s) \\
						&\equiv_\dagger \rtg(b\cdot r) +_{c\wedge b} \rtg(b\cdot s) \tag{\cref{lem:arbitrary seperation}}\\
						&\equiv_\dagger b\rtg(r) +_{c\wedge b} b\rtg(s)\tag{ind.~hyp.}\\
						&\equiv_\dagger b(\rtg(r) +_{c\wedge b} \rtg(s)) \tag{\cref{lemma:useful-equalities}}\\
						&\equiv_\dagger b(\rtg(r) +_{c} \rtg(s)) \tag{\cref{lemma:useful-equalities}}\\
						&\equiv_\dagger b\rtg(r + s) \tag{\cref{lem:arbitrary seperation}}
		\end{align*}
		\item In the second inductive case,
		\begin{align*}
			\rtg(b\cdot (rs)) &= \rtg((b\cdot r)s) \\
						&= \rtg(b\cdot r)\rtg(s) \\
						&\equiv_\dagger (b\rtg(r))\rtg(s) \tag{ind.~hyp.} \\
						&\equiv_\dagger b(\rtg(r)\rtg(s)) \tag{\textsf{G8},\textsf{G6}}  \\
						&= b\rtg(rs)
		\end{align*}
		\item In the final inductive step, let \(r \perp_c s\).
		Since \(b\cdot r \perp_{c\wedge b} b \cdot s\) and $r(r * s) \perp_c s$,
		\begin{align*}
			\rtg(b\cdot (r*s)) &= \rtg((b\cdot r)(r*s) + b\cdot s) \\
						&\equiv_\dagger \rtg(b\cdot r)\rtg(r*s) +_{b\wedge c}  \rtg(b\cdot s) \tag{\cref{lem:arbitrary seperation}} \\
						&\equiv_\dagger b\rtg(r)\rtg(r*s) +_{b\wedge c} b\rtg(s) \tag{ind.~hyp.}\\
						&\equiv_\dagger b(\rtg(r)\rtg(r*s) +_{b \wedge c}  \rtg(s)) \tag{\cref{lemma:useful-equalities}}\\
						&= b(\rtg(r(r*s)) +_{c} \rtg(s)) \tag{BA} \\
						&= b\rtg(r(r*s) + s) \tag{\cref{lem:arbitrary seperation}}\\
						&\equiv_\dagger b\rtg(r*s)  \tag{\cref{lem:rtg is also good stuff}}
		\end{align*}
	\end{itemize}
	This intermediary fact (*) lets us establish the main claim by induction on \(e\).
	\begin{itemize}
		\item In the first base case, \(\rtg(\gtr(0)) = 0\) by definition.
		\item In the second base case, let \(\At = \{\alpha_1, \dots, \alpha_n\}\). Then inductively,
		\begin{align*}
			\rtg(\gtr(p))
			&= \rtg\left(\sum_{\alpha \in \At} \alpha p\right)  \\
			&= \rtg\left(\alpha_1 p + \sum_{\alpha \in \At\setminus \alpha_1} \alpha p\right) \\
			&= p +_{\alpha_1} \rtg\left(\sum_{\alpha \in \At\setminus \alpha_1} \alpha p\right) \\
			&= p +_{\alpha_1} (p +_{\alpha_2} \cdots +_{\alpha_{n-1}} (p +_{\alpha_n} 0)) \\
			&\equiv_\dagger p +_{1} 0 \tag{\textsf{G3}\(\times n\), BA} \\
			&\equiv_\dagger p \tag{G0}
		\end{align*}
		\item
		In the first inductive step,
		\begin{align*}
			\rtg(\gtr(e +_b f))
				&= \rtg(b\cdot \gtr(e) + \bar b\cdot \gtr(f)) \\
				&\equiv_\dagger \rtg(b\cdot \gtr(e)) +_b \rtg(\bar b\cdot \gtr(f)) \tag{\cref{lem:arbitrary seperation}} \\
				&\equiv_\dagger b\rtg(\gtr(e)) +_b \bar b\rtg(\gtr(f)) \tag{*}\\
				&\equiv_\dagger be +_b \bar b f \tag{ind.~hyp.}\\
				&\equiv_\dagger e +_b f \tag{\cref{lemma:useful-equalities}}
		\end{align*}
		\item
		In the inductive step,
		\begin{align*}
			\rtg(\gtr(e f)) &= \rtg(\gtr(e) \gtr(f))\\
						&= \rtg(\gtr(e)) \rtg(\gtr(f))\\
						&\equiv e f \tag{ind.~hyp.}
		\end{align*}
		\item
		In the inductive step,
		\begin{align*}
			&\rtg(\gtr(e^{(b)}f))\\
			&= \rtg((b\cdot \gtr(e))*(\bar b\cdot \gtr(f)))
			\\
			&\equiv_\dagger \rtg(b\cdot \gtr(e))^{(b)}(\rtg(\bar b\cdot \gtr(f))) \tag{\cref{lem:arbitrary seperation}}
			\\
			&\equiv_\dagger (b\rtg(\gtr(e)))^{(b)}(\bar b\rtg(\gtr(f)) ) \tag{*}
			\\
			&\equiv_\dagger (be)^{(b)}(\bar b f) \tag{ind.~hyp.}
			\\
			&\equiv_\dagger e^{(b)}f  \tag{\cref{lemma:useful-equalities}}
% %
% 			&\equiv_\dagger \gtr(e)(b\gtr(e))^{(b)}(\bar b\gtr(f)) +_b \gtr(f) \tag{\cref{lemma:useful-equalities}}
% 			\\
% %
% 			&\equiv_\dagger (\gtr(e)(\gtr(e) +_b 0)^{(b)}(\gtr(f) +_{\bar b} 0) +_b 0) +_b (\gtr(f) +_{\bar b} 0) \tag{\textsf{G8}, \textsf{G6}}
% 			\\
% 			&\equiv_\dagger \gtr(e)(\gtr(e) +_b 0)^{(b)}(\gtr(f) +_{\bar b} 0) +_b (0 +_b (\gtr(f) +_{\bar b} 0)) \tag{\cref{lemma:useful-equalities}, BA}
% 			\\
% 			&\equiv_\dagger \gtr(e)(\gtr(e) +_b 0)^{(b)}(\gtr(f) +_{\bar b} 0) +_b (0 +_b (0 +_{b} \gtr(f))) \tag{\textsf{G2}}
% 			\\
% 			&\equiv_\dagger \gtr(e)(\gtr(e) +_b 0)^{(b)}(\gtr(f) +_{\bar b} 0) +_b ((0 +_{\bar b} 0) +_b \gtr(f)) \tag{\textsf{G3}, BA}
% 			\\
% 			&\equiv_\dagger \gtr(e)(\gtr(e) +_b 0)^{(b)}(\gtr(f) +_{\bar b} 0) +_b (\gtr(f) +_{\bar b} 0) \tag{\textsf{G3}}
% 			\\
% 			&\equiv_\dagger (\gtr(e)(\gtr(e) +_b 0)^{(b)}(\gtr(f) +_{\bar b} 0) +_b \gtr(f)) +_1 0 \tag{\textsf{G3}, BA}
% 			\\
% 			&\equiv_\dagger \gtr(e)(\gtr(e) +_b 0)^{(b)}(\gtr(f) +_{\bar b} 0) +_b \gtr(f) \tag{G0}
% 			\\
% 			&\equiv_\dagger \gtr(e)^{(b)}\gtr(f) \tag{\textsf{FP}}
		\end{align*}
	\end{itemize}
\end{proof}

\section{Proofs for \Cref{sec:completeness for skip-free GKAT}}%
\label{appendix:completeness for skip-free gkat}

\lemdeadequalszero*
\begin{proof}
	It suffices to prove that if $a \in \BExp$ and $a\L(e) = \emptyset$, then $ae \equiv 0$.
	After all, if this is true and $e$ is dead then $1\L(e) = \L(e) = \emptyset$, and hence $e = 1e \equiv 0$.

	We proceed by induction on $e$.
	In the base, there are two cases:
	\begin{itemize}
		\item
		If $e = p \in \Sigma$, then for all $\alpha \leq a$ it holds that $\alpha{} p \in a\L(e)$.
		We thus find that $a \equiv 0$, and therefore $0p = p +_0 0 \equiv 0 +_1 p \equiv 0$.

		\item
		If $e = 0$, then $b0 = 0 +_b 0 \equiv 0$ immediately.
	\end{itemize}
	For the inductive step, there are three more cases.
	\begin{itemize}
		\item
		If $e = e_1 +_b e_2$ then $(a \wedge b)\L(e_1) = \emptyset$ and $(a \wedge \overline{b})\L(e_2) = \emptyset$.
		By induction, we have $(a \wedge b)e_1 \equiv 0$ and $(a \wedge \overline{b})e_2 \equiv 0$.
		We then derive using \Cref{lemma:useful-equalities}:
		\[
		a(e_1 +_b e_2)
		\equiv ae_1 +_b ae_2
		\equiv bae_1 +_b \overline{b}ae_2
		\equiv (a \wedge b)e_1 +_b (a \wedge \overline{b})e_2
		\equiv 0 +_b 0
		\equiv 0
		\]

		\item
		If $e = e_1 \cdot e_2$, then $a\L(e_1) = \emptyset$ or $\L(e_2) = \emptyset$.
		In the former case $ae_1 \equiv 0$, so by \Cref{lemma:useful-equalities} we find
		\(
		a(e_1 \cdot e_2)
		= (ae_1) \cdot e_2
		\equiv 0 \cdot e_2
		\equiv 0
		\).
		In the latter case $e_2 \equiv 0$, and thus $a(e_1 \cdot e_2) \equiv a(e_1 \cdot 0) \equiv a0 = 0$.

		\item
		If $e = e_1 \while{b} e_2$, then $\overline{b}\L(e_2) = \emptyset$, and so $\overline{b}e_2 \equiv 0$ by induction.
		By \Cref{lemma:useful-equalities}:
		\[
		a(e_1 \while{b} e_2)
		\equiv a(e_1 \while{b} \overline{b}e_2)
		\equiv a(e_1 \while{b} 0)
		\equiv a0
		= 0
		\]
	\end{itemize}
\end{proof}

\restateprunepreservesequiv*
\begin{proof}
	We proceed by induction on $e$.
	In the base, the claim holds immediately, whether $e = 0$ or $e = p \in \Sigma$.
	For the inductive step, there are three cases.
	\begin{itemize}
		\item
		If $e = e_1 +_b e_2$, then $e_1 +_b e_2 \equiv \floor{e_1} +_b \floor{e_2} = \floor{e_1 +_b e_2}$.

		\item
		Suppose $e = e_1 \cdot e_2$.
		If $e_2$ is dead, then $e_2 \equiv 0$ by \Cref{lemma:dead-equals-zero}, and so $e_1 \cdot e_2 \equiv e_1 \cdot 0 \equiv 0 = \floor{e_1 \cdot e_2}$.
		Otherwise $e_1 \cdot e_2 \equiv \floor{e_1} \cdot \floor{e_2} = \floor{e_1 \cdot e_2}$.

		\item
		Suppose $e = e_1 \while{b} e_2$.
		If $\overline{b}e_2$ is dead, then $\overline{b}e_2 \equiv 0$ by \Cref{lemma:dead-equals-zero}, and so $e_1 \while{b} e_2 \equiv e_1 \while{b} \overline{b}e_2 \equiv e_1 \while{b} 0 \equiv 0 = \floor{e_1 \while{b} e_2}$ by \Cref{lemma:useful-equalities}.
		Otherwise, we derive that $e_1 \while{b} e_2 \equiv \floor{e_1} \while{b} \floor{e_2} = \floor{e_1 \while{b} e_2}$.
	\end{itemize}
\end{proof}

\noindent
To prove \Cref{lemma:prune-bisimilar}, we need an auxiliary lemma.

\begin{lemma}%
	\label{lemma:empty-dead}
	Let $(X, h)$ be a $G$-coalgebra and $x \in X$.
	Now $\L(x, (X, h)) = \emptyset$ if and only if $x \in D(X, h)$.
\end{lemma}
\begin{proof}
	For the direction from left to right, it suffices to show that $D = \{ x \in X : \L(x, (X, h)) = \emptyset \}$ satisfies the rules for dead states.
	To see this, note that if $x \in D$ then certainly the case where $h(x)(\alpha) = (p, \checkmark)$ is excluded, for then $\alpha{}p \in \L(x, (X, h))$.
	Furthermore, if $h(x)(\alpha) = (p, x')$ for some $p \in \Sigma$, then $\L(x', (X, h)) = \emptyset$, for otherwise $L(x, (X, h))$ cannot be empty; hence $x' \in D$.

	For the converse, suppose $w \in \L(x, (X, h))$.
	Then we can find, by induction on $w$, an $x' \in D(X, h)$ with $h(x')(\alpha) = (p, \checkmark)$ for some $p \in \Sigma$ --- a contradiction.
\end{proof}

\restateprunebisimilar*
\begin{proof}
	It suffices to prove that $R = \{ (x, y) : \L(x, (X, h)) = \L(y, (X, h)) \}$ is a bisimulation in $(X, \floor{h})$.
	To see this, suppose that $\L(x, (X, h)) = \L(y, (X, h))$.
	\begin{itemize}
		\item
		If $\floor{h}(x)(\alpha) = \bot$, then $h(x)(\alpha) = \bot$ or $h(x)(\alpha) = (p, y')$ with $x' \in D(X, h)$.
		Therefore there is no word in $\L(x, (X, h)) = \L(x, (Y, h))$ of the form $\alpha{}w$, by \Cref{lemma:empty-dead}.
		Thus $h(y)(\alpha) = \bot$ or $h(y)(\alpha) = (p, y')$ where $\L(y', (X, h)) = \emptyset$, whence $y' \in D(X, h)$ by \Cref{lemma:empty-dead}.
		In either case, $\floor{h}(y)(\alpha) = \bot$.

		\item
		If $\floor{h}(x)(\alpha) = (p, \checkmark)$ for some $p \in \Sigma$, then $\alpha{}p \in \L(x, (X, h))$, and therefore $\alpha{}p \in \L(y, (Y, h))$.
		But then $\floor{h}(y)(\alpha) = (p, \checkmark)$ as well.

		\item
		If $\floor{h}(x)(\alpha) = (p, x') \in \Sigma \times X$, then $x'$ is not dead in $(X, h)$.
        By \Cref{lemma:empty-dead}, there exists some $w$ s.t.\ $\alpha{}pw \in \L(x, (X, h)) = \L(y, (X, h))$.
		We then have $\floor{h}(y)(\alpha) = (p, y')$ for some $y' \in X$.
		A straightforward argument then shows that $\L(x', (X, h)) = \L(y', (X, h))$, and so $x' \mathrel{R} y'$.
	\end{itemize}
\end{proof}

\noindent
The claim about the graph of $\floor{-}: \GExp^- \to \GExp^-$ being a bisimulation of skip-free automata can also be shown by induction on $e$ and exhaustive case analysis, as follows.

\restateprunecommutes*
\begin{proof}
	It suffices to show that \(\partial \circ \floor{-} = G(\floor{-}) \circ \floor{\partial}\)
	We proceed by induction on $e$.
	In the base, $\floor{e} = e$ and $G(\floor{-}) \circ \floor{\partial}(e) = \partial(e)$, and so the claim holds trivially.
	For the inductive step, there are three cases.
	\begin{itemize}
		\item
		Suppose $e = e_1 +_b e_2$.
		Without loss of generality, we assume $\alpha \leq b$.
		\begin{itemize}
			\item
			If $\floor{\partial}(e_1 +_b e_2) = \bot$, then either $\partial(e_1)(\alpha) = \bot$ or $\partial(e_1)(\alpha) = (p, e_1')$ with $e_1'$ dead.
			In either case, $\floor{\partial}(e_1)(\alpha) = \bot$, and so $\partial(\floor{e_1})(\alpha) = \bot$ by induction.
			We then conclude that $\partial(\floor{e_1 +_b e_2})(\alpha) = \partial(\floor{e_1})(\alpha) = \bot$.

			\item
			If $\floor{\partial}(e_1 +_b e_2) = (p, \checkmark)$ for some $p \in \Sigma$, then $\partial(e_1)(\alpha) = (p, \checkmark)$ as well.
			It follows that $\floor{\partial}(e_1)(\alpha) = (p, \checkmark)$, and so $\partial(\floor{e_1})(\alpha) = (p, \checkmark)$ by induction.
			We then conclude that $\partial(\floor{e_1 +_b e_2})(\alpha) = \partial(\floor{e_1})(\alpha) = (p, \checkmark)$.

			\item
			If $\floor{\partial}(e_1 +_b e_2) = (p, e_1')$, then $\partial(e_1)(\alpha) = (p, e_1')$ for some live $e_1' \in \GExp$.
			Hence $\floor{\partial}(e_1)(\alpha) = (p, e_1')$, and so $\partial(\floor{e_1})(\alpha) = (p, \floor{e_1'})$ by induction.
			We then conclude that $\partial(\floor{e_1 +_b e_2})(\alpha) = \partial(\floor{e_1})(\alpha) = (p, \floor{e_1'})$.
		\end{itemize}

		\item
		Suppose $e = e_1 \cdot e_2$.
		If $e_2$ is dead, then so is $e_1 \cdot e_2$; hence
		\[
		\floor{\partial}(e_1 \cdot e_2)(\alpha)
		= \bot
		= \partial(0)(\alpha)
		= \partial(\floor{e_1 \cdot e_2})(\alpha)
		\]
		Otherwise, if $e_2$ is live, then we have three cases to consider.
		\begin{itemize}
			\item
			If $\floor{\partial}(e_1 \cdot e_2)(\alpha) = \bot$, then we have two more subcases to consider.
			\begin{itemize}
				\item
				If $\partial(e_1 \cdot e_2)(\alpha) = \bot$, then $\partial(e_1)(\alpha) = \bot$, meaning $\floor{\partial}(e_1)(\alpha) = \bot$.
				By induction, $\partial(\floor{e_1})(\alpha) = \bot$, and so $\partial(\floor{e_1 \cdot e_2})(\alpha) = \bot$.
				\item
				If $\partial(e_1 \cdot e_2)(\alpha) = (p, e')$ with $e'$ dead, then we can exclude the case where $\partial(e_1)(\alpha) = (p, \checkmark)$, for then $e' = e_2$, which contradicts that $e_2$ is live.
				We then know that $e' = e_1' \cdot e_2$ such that $\partial(e_1)(\alpha) = (p, e_1')$.
				Furthermore, since $e_2$ is live and $e'$ is dead, it must be the case that $e_1'$ is dead.
				We then find that $\floor{\partial}(e_1)(\alpha) = \bot$, and so $\partial(\floor{e_1})(\alpha) = \bot$ by induction.
				We conclude that $\partial(\floor{e_1 \cdot e_2})(\alpha) = \partial(\floor{e_1} \cdot \floor{e_2})(\alpha) = \bot$.
			\end{itemize}

			\item
			If $\floor{\partial}(e_1 \cdot e_2) = (p, \checkmark)$, then $\partial(e_1 \cdot e_2)(\alpha) = (p, \checkmark)$, which is impossible.
			We can therefore exclude this case.

			\item
			If $\floor{\partial}(e_1 \cdot e_2) = (p, e')$, then $\partial(e_1 \cdot e_2)(\alpha) = (p, e')$ with $e'$ live.
			This gives us two more subcases.
			\begin{itemize}
				\item
				If $\partial(e_1)(\alpha) = (p, \checkmark)$ and $e' = e_2$, then $\floor{\partial}(e_1)(\alpha) = (p, \checkmark)$, and so by induction $\partial(\floor{e_1})(\alpha) = (p, \checkmark)$.
				Thus $\partial(\floor{e_1 \cdot e_2})(\alpha) = (p, \floor{e'})$.

				\item
				If $\partial(e_1)(\alpha) = (p, e_1')$ and $e' = e_1' \cdot e_2$, then $e_1'$ must be live.
				We then have $\floor{\partial}(e_1)(\alpha) = (p, e_1')$, and so $\partial(\floor{e_1})(\alpha) = (p, \floor{e_1'})$ by induction.
				We conclude that $\partial(\floor{e_1 \cdot e_2})(\alpha) = (p, \floor{e_1'} \cdot \floor{e_2}) = (p, \floor{e'})$.
			\end{itemize}
		\end{itemize}

		\item
		Suppose $e = e_1 \while{b} e_2$.
		If $\overline{b}e_2$ is dead, then so is $e_1 \while{b} e_2$; hence
		\[
		\floor{\partial}(e_1 \while{b} e_2)(\alpha) = \bot = \partial(0)(\alpha) = \partial(\floor{e_1 \while{b} e_2})(\alpha)
		\]
		Otherwise, if $\overline{b}e_2$ is live, then we first consider the case where $\alpha \leq \overline{b}$.
		\begin{itemize}
			\item
			If $\floor{\partial}(e_1 \while{b} e_2)(\alpha) = \bot$, then $\partial(e_2)(\alpha) = \bot$ or $\partial(e_2)(\alpha) = (p, e_2')$ with $e_2'$ dead.
			In either case, $\floor{\partial}(e_2)(\alpha) = \bot$, and so $\partial(\floor{e_2})(\alpha) = \bot$ by induction.
			We conclude that $\partial(\floor{e_1 \while{b} e_2})(\alpha) = \partial(\floor{e_2})(\alpha) = \bot$.

			\item
			If $\floor{\partial}(e_1 \while{b} e_2)(\alpha) = (p, \checkmark)$, then $\partial(e_2)(\alpha) = (p, \checkmark)$.
			Thus $\floor{\partial}(e_2)(\alpha) = (p, \checkmark)$ and so we have $\partial(\floor{e_2})(\alpha) = (p, \checkmark)$ by induction.
			We then conclude that $\partial(\floor{e_1 \while{b} e_2})(\alpha) = \partial(\floor{e_2})(\alpha) = (p, \checkmark)$.

			\item
			If $\floor{\partial}(e_1 \while{b} e_2)(\alpha) = (p, e_2')$, then $\partial(e_2)(\alpha) = (p, e_2')$ with $e_2'$ live.
			It then follows that $\floor{\partial}(e_2)(\alpha) = (p, e_2')$, and so $\partial(\floor{e_2})(\alpha) = (p, \floor{e_2'})$ by induction.
			We then conclude that $\partial(\floor{e_1 \while{b} e_2})(\alpha) = \partial(\floor{e_2})(\alpha) = (p, \floor{e_2'})$.
		\end{itemize}
		It remains to consider the case where $\alpha \leq b$.
		\begin{itemize}
			\item
			If $\floor{\partial}(e_1 \while{b} e_2)(\alpha) = \bot$, then we have two more subcases to consider.
			\begin{itemize}
				\item
				If $\partial(e_1 \while{b} e_2)(\alpha) = \bot$, then $\partial(e_1)(\alpha) = \bot$ as well, since $\alpha \leq b$.
				This means that $\floor{\partial}(e_1)(\alpha) = \bot$, and so $\partial(\floor{e_1})(\alpha) = \bot$ by induction.
				We conclude that $\partial(\floor{e_1 \while{b} e_2})(\alpha) = \bot$.

				\item
				If $\partial(e_1 \while{b} e_2)(\alpha) = (p, e')$ with $e'$ dead, then since $\alpha \leq b$ we must have that $e' = e_1' \cdot e_1 \while{b} e_2$, with $\partial(e_1)(\alpha) = (p, e_1')$.
				Since $\overline{b}e_2$ is live, so is $e_1 \while{b} e_2$.
				It then follows that $e_1'$ is dead, and so $\floor{\partial}(e_1)(\alpha) = \bot$.
				By induction, $\partial(\floor{e_1})(\alpha) = \bot$, meaning that $\partial(\floor{e_1 \while{b} e_2})(\alpha) = \bot$.
			\end{itemize}

			\item
			If $\floor{\partial}(e_1 \while{b} e_2)(\alpha) = (p, \checkmark)$, then $\partial(e_1 \while{b} e_2)(\alpha) = (p, \checkmark)$, which is impossible when $\alpha \leq b$; we can therefore exclude this case.

			\item
			If $\floor{\partial}(e_1 \while{b} e_2)(\alpha) = (p, e')$, then $\partial(e_1 \while{b} e_2)(\alpha) = (p, e')$ with $e'$ live.
			Since $\alpha \leq b$, we have that $\partial(e_1)(\alpha) = (p, e_1')$ and $e' = e_1' \cdot e_1 \while{b} e_2$.
			Since $e'$ and $e_1 \while{b} e_2$ are live, so is $e_1'$.
			We then find that $\floor{\partial}(e_1)(\alpha) = (p, e_1')$, meaning $\partial(\floor{e_1})(\alpha) = (p, \floor{e_1'})$ by induction.
			We conclude by deriving
			\[
			\partial(\floor{e_1 \while{b} e_2})(\alpha)
			= (p, \floor{e_1'} \cdot \floor{e_1} \while{b} \floor{e_2})
			= (p, \floor{e'})
			\]
		\end{itemize}
	\end{itemize}
\end{proof}

\section{Proofs for \Cref{sec:relation to GKAT}}%
\label{appendix:relation to GKAT}

We start by proving that our tweaked version of the syntactic \GKAT automaton does not make any difference with regard to bisimilarity of expressions.
To make this precise, we first recall the definition of the syntactic \GKAT automaton~\cite{gkaticalp}.
\begin{definition}
	We define $\delta': \GExp \to (\bot + \checkmark + \Sigma \times \GExp)^\At$ in the same way as $\delta$, except for the following two cases:
	\begin{itemize}
		\item
		When $e, f \in \GExp$ such that $\delta(e)(\alpha) = (p, e')$, we set $\delta(e \cdot f)(\alpha) = (p, e' \cdot f)$.
		\item
		When $e \in \GExp$ and $\alpha \leq b$, we set $\delta(e^{(b)})(\alpha) = (p, e' \cdot e^{(b)})$.
	\end{itemize}
	Essentially, $\delta$ and $\delta'$ differ only by the use of $\cdot$ versus $\fatsemi$ in these two cases.
\end{definition}

We can now show that an expression is bisimilar to its representation in the syntactic automaton \GKAT as defined in the literature; this then tells us that two states are bisimilar in $(\GExp, \delta)$ if and only if they are bisimilar in $(\GExp, \delta')$.

\begin{restatable}{lemma}{restatesmartconstructorbisimilar}%
	\label{lemma:smart-constructor-bisimilar}
	Let $e \in \GExp$.
	Now $(\GExp, \delta'), e \bisim (\GExp, \delta), e$.
\end{restatable}
\begin{proof}
	Let $R$ be the smallest relation on $\GExp$ such that all of the following hold:
	\begin{mathpar}
		\inferrule{~}{%
			e \mathrel{R} e
		}
		\and
		\inferrule{%
			e \mathrel{R} e' \\
			f \in \mathbb{E}
		}{%
			e \cdot f \mathrel{R} e' \fatsemi f
		}
	\end{mathpar}
	Let $e, f \in \GExp$ such that $e \mathrel{R} f$.
	We proceed by induction on the construction of $R$.
	In the base, $e \mathrel{R} f$ because $e = f$; we proceed by induction on $e$.
	In the (inner) base, the definitions of $\delta'$ and $\delta$ coincide, and so those cases go through immediately.
	For the (inner) inductive step, we have three more cases.
	\begin{itemize}
		\item
		If $e = e_1 +_b e_2$, then let $\alpha \in \At$ and assume without loss of generality that $\alpha \leq b$.
		In that case, $\delta'(e_1 +_b e_2)(\alpha) = \delta'(e_1)(\alpha)$ and $\delta(e_1 +_b e_2)(\alpha) = \delta(e_1)(\alpha)$.
		The claim then follows by induction.

		\item
		If $e = e_1 \cdot e_2$, then there are three more subcases to consider.
		\begin{itemize}
			\item
			If $\delta'(e_1)(\alpha) = (p, e_1')$, then by induction $\delta(e_1)(\alpha) = (p, e_1'')$ such that $e_1' \mathrel{R} e_1''$.
			Now, $\delta'(e_1 \cdot e_2)(\alpha) = (p, e_1' \cdot e_2)$ and $\delta(e_1 \cdot e_2) = (p, e_1'' \fatsemi e_2)$, while $e_1' \cdot e_2 \mathrel{R} e_1'' \fatsemi e_2$, as desired.

			\item
			If $\delta'(e_1)(\alpha) = \bot$, then $\delta(e_1)(\alpha) = \bot$ by induction, and so $\delta'(e_1 \cdot e_2)(\alpha) = \bot = \delta(e_1 \cdot e_2)(\alpha)$.

			\item
			If $\delta'(e_1)(\alpha) = \checkmark$, then $\delta(e_1)(\alpha) = \checkmark$, so $\delta'(e_1 \cdot e_2)(\alpha) = \delta'(e_2)(\alpha)$ and $\delta(e_1 \cdot e_2)(\alpha) = \delta(e_1)(\alpha)$; the claim then follows by induction.
		\end{itemize}

		\item
		If $e = e_1^{(b)}$, then let $\alpha \in \At$.
		There are three more subcases to consider.
		\begin{itemize}
			\item
			If $\alpha \leq b$ and $\delta'(e_1)(\alpha) = (p, e_1')$, then by induction $\delta(e_1)(\alpha) = (p, e_1'')$ such that $e_1' \mathrel{R} e_1''$.
			We then find that $\delta'(e_1^{(b)})(\alpha) = (p, e_1' \cdot e_1^{(b)})$ and $\delta(e_1^{(b)})(\alpha) = (p, e_1'' \fatsemi e_1^{(b)})$, while $e_1' \cdot e_1^{(b)} \mathrel{R} e_1'' \fatsemi e_1^{(b)}$.
            \item
            If $\alpha \leq b$ and $\delta'(e_1)(\alpha) = \bot$, then $\delta(e_1)(\alpha) = \bot$ by induction, so $\delta'(e_1 \while{b} e_2)(\alpha) = \bot = \delta(e_1 \while{b} e_2)(\alpha)$ as desired.
            \item
            If $\alpha \leq \overline{b}$, then $\delta'(e_1 \while{b} e_2)(\alpha) = \checkmark = \delta(e_1 \while{b} e_2)(\alpha)$.
		\end{itemize}
	\end{itemize}
	For the (outer) inductive step, $e \mathrel{R} f$ because $e = e' \fatsemi g$ and $f = f' \cdot g$ such that $e' \mathrel{R} f'$.
	We distinguish two cases.
	\begin{itemize}
		\item
		If $e' = 1$, then because $e' \mathrel{R} f'$ we have that $\delta'(e')(\alpha) = \checkmark$ for all $\alpha \in \At$, so $\delta(f')(\alpha) = \checkmark$ for all $\alpha \in \At$ by induction.
		But in that case we have that $\delta'(e)(\alpha) = \delta'(g)(\alpha)$, and similarly $\delta(f)(\alpha) = \delta(g)(\alpha)$.
		The claim follows by an argument similar to the (outer) base.

		\item
		Otherwise $e = e' \cdot g$; there are now three cases.
		\begin{itemize}
			\item
			If $\delta'(e')(\alpha) = \bot$, then $\delta(f')(\alpha) = \bot$ by induction, and so $\delta'(e)(\alpha) = \bot = \delta(f)(\alpha)$, as desired.

			\item
			If $\delta'(e')(\alpha) = \checkmark$, then $\delta(f')(\alpha) = \checkmark$ by induction, and so $\delta'(e)(\alpha) = \delta'(g)(\alpha)$ and $\delta(f)(\alpha) = \delta(g)(\alpha)$; the claim then follows by an argument similar to the one in the outer base.

			\item
			If $\delta'(e')(\alpha) = (p, e'')$, then $\delta'(f')(\alpha) = (p, f'')$ with $e'' \mathrel{R} f''$ by induction.
			In that case, $\delta'(e)(\alpha) = (p, e'' \cdot g)$ and $\delta(f)(\alpha) = (p, f'' \fatsemi g)$ while $e'' \cdot g \mathrel{R} f'' \fatsemi g$, as desired.
		\end{itemize}
	\end{itemize}
\end{proof}

\restateembedbisimilar*
\begin{proof}
	For the forward direction, let $R$ be a $G$-bisimulation on $(X, d)$.
	We claim that $R \cup \{ (\top, \top) \}$ is an $F$-bisimulation on $\embed(X,d) = (X+\top, \tilde{d})$.
	To see this, first note that the pair $(\top, \top)$ immediately satisfies the conditions put on an $F$-bisimulation.
	Furthermore, if $x \mathrel{R} y$, then we check the three conditions.
	\begin{itemize}
		\item
		If $\tilde{d}(x)(\alpha) = \bot$, then $d(x)(\alpha) = \bot$ as well, which means that $d(y)(\alpha) = \bot$ and hence $\tilde{d}(y)(\alpha) = \bot$; this condition holds.
		\item
		If $\tilde{d}(x)(\alpha) = \checkmark$, then $x = \top$, but this contradicts that $x \mathrel{R} y$; we can therefore disregard this case.
		\item
		If $\tilde{d}(x)(\alpha) = (p, x')$, then there are two possibilities.
		\begin{itemize}
			\item
			If $x' = \top$ and $d(x)(\alpha) = (p, \checkmark)$, then $d(y)(\alpha) = (p, \checkmark)$ as well.
			But then $\tilde{d}(y)(\alpha) = (p, \top)$.
			Since $\top$ is related to $\top$ by $R \cup \{ (\top, \top) \}$, we are done.

			\item
			If $d(x)(\alpha) = (p, x')$, then $d(y)(\alpha) = (p, y')$ such that $x' \mathrel{R} y'$.
			But then $\tilde{d}(y)(\alpha) = (p, y')$ as well, and we are done.
		\end{itemize}
	\end{itemize}
	For the converse claim, let $R$ be a bisimulation on $\embed(X, d)$.
	First, we note that if $x \in X + \top$ and $x \mathrel{R} \top$ or $\top \mathrel{R} x$, then necessarily $x = \top$ --- after all, this tells us that $\tilde{d}(x)(\alpha) = \checkmark$ for all $\alpha$, and $\top$ is the only element of $X + \top$ that fits that description.
	We now claim that $R' = R \cap X \times X$ is a bisimulation on $(X, d)$.
	To this end, let $x \mathrel{R'} y$; we check the three conditions.
	\begin{itemize}
		\item
		If $d(x)(\alpha) = \bot$, then $\tilde{d}(x)(\alpha) = \bot$ as well, which means $\tilde{d}(x)(\alpha) = \bot$, and hence $d(x)(\alpha) = \bot$; this condition is covered.
		\item
		If $d(x)(\alpha) = (p, \checkmark)$, then $\tilde{d}(x)(\alpha) = (p, \top)$, which means that $\tilde{d}(y)(\alpha) = (p, y')$ such that $\top \mathrel{R} y'$.
		But by the considerations above, this means that $y' = \top$, and thus $d(y)(\alpha) = (p, \checkmark)$ as well.
		\item
		If $d(x)(\alpha) = (p, x')$, then $\tilde{d}(x)(\alpha) = (p, x')$ as well.
		But then $\tilde{d}(y)(\alpha) = (p, y')$ such that $x' \mathrel{R} y'$.
		Now, since $x' \neq \top$, also $y' \neq \top$ by the considerations above, and thus $x' \mathrel{R'} y'$.
	\end{itemize}
\end{proof}

\restateembedexpmorphism*
\begin{proof}
    Let \(h: \GExp^- + \top \to \GExp\) be given by \(h(e) = e\) when \(e \in \GExp^-\), and \(h(\top) = 1\).
    Clearly, the graph of \(h\) is the relation claimed to be a bisimulation.
    By \Cref{lem:same as homom}, it suffices to show that \(h\) is a coalgebra homomorphism.

	For the special case of $\top$, we have that $F(h)(\tilde{\partial}(\top))(\alpha) = \checkmark = \delta(1)(\alpha) = \delta(h(\top))(\alpha)$ for all $\alpha \in \At$.
	It remains to check that $F(h)(\tilde{\partial}(e))(\alpha) = \delta(e)(\alpha)$ for all $e \in \GExp$ and $\alpha \in \At$, which we do by induction on $e$.
    There are two base cases:
	\begin{itemize}
		\item
		If $e = 0$, then $F(h)(\tilde{\partial}(0))(\alpha) = \bot = \delta(0)(\alpha)$.

		\item
		If $e = p$, then $F(h)(\tilde{\partial}(p))(\alpha) = (p, h(\top)) = (p, 1) = \delta(p)(\alpha)$.
    \end{itemize}
    For the inductive step, we distinguish the following cases:
    \begin{itemize}
		\item
		If $e = e_0 +_b e_1$, then there are two cases to consider.
		First, if $\alpha \leq b$, then
		\[
		F(h)(\tilde{\partial}(e_0 +_b e_1))(\alpha)
		= F(h)(\tilde{\partial}(e_0))(\alpha)
		\stackrel{\text{IH}}{=} \delta(e_0)(\alpha)
		= \delta(e_0 +_b e_1)(\alpha)
		\]
		The case where $\alpha \not\leq b$ is similar.

		\item
		If $e = e_0 \cdot e_1$, then there are three cases to consider.
		\begin{itemize}
			\item
			If $\partial(e_0)(\alpha) = (p, e_0')$, then $\tilde{\partial}(e_0)(\alpha) = (p, e_0')$ as well.
			By induction, $\delta(e_0)(\alpha) = (p, e_0')$, and since $\tilde{\partial}(e_0 \cdot e_1)(\alpha) = (p, e_0' \cdot e_1)$, we have
			\[
			F(h)(\tilde{\partial}(e_0 \cdot e_1))(\alpha)
			= (p, e_0' \cdot e_1)
			= \delta(e_0 \cdot e_1)(\alpha)
			\]
			\item
			If $\partial(e_0)(\alpha) = (p, \checkmark)$, then $\tilde{\partial}(e_0)(\alpha) = (p, \top)$.
			By induction, $\delta(e_0)(\alpha) = (p, 1)$, and since $\tilde{\partial}(e_0 \cdot e_1)(\alpha) = (p, e_1)$ we have
			\[
			F(h)(\tilde{\partial}(e_0 \cdot e_1))(\alpha)
			= (p, e_1)
			= (p, 1 \fatsemi e_1)
			= \delta(e_0 \cdot e_1)(\alpha)
			\]
			\item
			If $\partial(e_0)(\alpha) = \bot$, then $\tilde{\partial}(e_0)(\alpha) = \bot$ as well.
			By induction $\delta(e_0)(\alpha) = \bot$, and so
			\[
			F(h)(\tilde{\partial}(e_0 \cdot e_1))(\alpha)
			= \bot
			= \delta(e_0 \cdot e_1)(\alpha)
			\]
		\end{itemize}

		\item
		If $e = e_0 \while{b} e_1$, then there are four cases to consider.
		\begin{itemize}
			\item
			If $\alpha \leq b$ and $\partial(e_0)(\alpha) = (p, e_0')$, then $\tilde{\partial}(e_0)(\alpha) = (p, e_0')$ as well.
			By induction $\delta(e_0)(\alpha) = (p, e_0')$, so
			\[
			F(h)(\tilde{\partial}(e_0 \while{b} e_1))(\alpha)
			= (p, e_0' \cdot e_0 \while{b} e_1)
			= \delta(e_0 \while{b} e_1)(\alpha)
			\]

			\item
			If $\alpha \leq b$ and $\partial(e_0)(\alpha) = (p, \checkmark)$, then $\tilde{\partial}(e_0)(\alpha) = (p, \top)$.
			By induction, $\delta(e_0)(\alpha) = (p, 1)$, and so
			\[
			F(h)(\tilde{\partial}(e_0 \while{b} e_1))(\alpha)
			= (p, e_0 \while{b} e_1)
			= (p, 1 \fatsemi e_0 \while{b} e_1)
			= \delta(e_0 \while{b} e_1)(\alpha)
			\]

			\item
			If $\alpha \leq b$ and $\partial(e_0)(\alpha) = \bot$, then $\tilde{\partial}(e_0)(\alpha) = \bot$.
			By induction $\delta(e_0)(\alpha) = \bot$, and so
			\[
			F(h)(\tilde{\partial}(e_0 \while{b} e_1))(\alpha)
			= \bot
			= \delta(e_0 \while{b} e_1)(\alpha)
			\]

			\item
			If $\alpha \not\leq b$, then $\tilde{\partial}(e_0 \while{b} e_1)(\alpha) = \tilde{\partial}(e_1)(\alpha)$, so we derive as follows:
			\[
			F(h)(\tilde{\partial}(e_0 \while{b} e_1))(\alpha)
			= F(h)(\tilde{\partial}(e_1))(\alpha)
			\stackrel{\text{IH}}{=} \delta(e_1)(\alpha)
			= \delta(e_0 \while{b} e_1)(\alpha)
			\]
		\end{itemize}
	\end{itemize}
\end{proof}

\restatelanguagerecover*
\begin{proof}
	We proceed by induction on $e$.
	In the base, there are two cases.
	\begin{itemize}
		\item
		First, if $e = 0$, then $\L(e) = \emptyset$ by definition, and $\widehat{\L}(e) = \{ \alpha \in \At \mid \alpha \leq 0 \} = \emptyset$, meaning that $\widehat{\L}(e) = \emptyset = \emptyset \cdot \At = \L(e) \cdot \At$.

		\item
		Second, if $e = p$, then $w \in \widehat{\L}(e)$ if and only if $w = \alpha{}p\beta$ for some $\alpha,\beta \in \At$, which is true if and only if $\alpha{}p\beta \in \L(e) \cdot \At$.
	\end{itemize}
	For the inductive step, we first note that for $b \in \BExp$ and $e \in \GExp$, we have $\widehat{\L}(b) \diamond (\L(e) \cdot \At) = b\L(e) \cdot \At$.
	We now distinguish three cases.
	\begin{itemize}
		\item
		If $e = e_1 +_b e_2$, then we derive
		\begin{align*}
			\widehat{\L}(e_1 +_b e_2)
			&= \widehat{\L}(b) \diamond \widehat{\L}(e_1) \cup \widehat{\L}(\overline{b}) \diamond \widehat{\L}(e_2) \\
			&= \widehat{\L}(b) \diamond (\L(e_1) \cdot \At) \cup \widehat{\L}(\overline{b}) \diamond (\L(e_2) \cdot \At) \tag{IH} \\
			&= b\L(e_1) \cdot \At \cup \overline{b}\L(e_2) \cdot \At \\
			&= (b\L(e_1) \cup \overline{b}\L(e_2)) \cdot \At \\
			&= \L(e_1 +_b e_2) \cdot \At \tag{\Cref{lem:language}}
			\intertext{
				\item
				If $e = e_1 \cdot e_2$, then we derive
			}
			\widehat{\L}(e_1 \cdot e_2)
			&= \widehat{\L}(e_1) \diamond \widehat{\L}(e_2) \\
			&= (\L(e_1) \cdot \At) \diamond (\L(e_2) \cdot \At) \tag{IH} \\
			&= \L(e_1) \cdot \L(e_2) \cdot \At \\
			&= \L(e_1 \cdot e_2) \cdot \At \tag{\Cref{lem:language}}
			\intertext{
				\item
				If $e = e_1 \while{b} e_2$, then we derive
			}
			\widehat{\L}(e_1 \while{b} e_2)
			&= (\widehat{\L}(b) \diamond \widehat{\L}(e_1))^{(*)} \diamond \widehat{\L}(\overline{b}) \diamond \widehat{\L}(e_2) \\
			&= (\widehat{\L}(b) \diamond (\L(e_1) \cdot \At))^{(*)} \diamond \widehat{\L}(\overline{b}) \diamond (\L(e_2) \cdot \At) \tag{IH} \\
			&= (b\L(e_1) \cdot \At)^{(*)} \diamond (\overline{b}\L(e_2) \cdot \At) \\
			&= \Bigl(\bigcup_{n \in \mathbb{N}} b\L(e_1)^n \cdot \At\Bigr) \diamond (\overline{b}\L(e_2) \cdot \At) \\
			&= \bigcup_{n \in \mathbb{N}} b\L(e_1)^n \cdot \overline{b}\L(e_2) \cdot \At \\
			&= \L(e_1 \while{b} e_2) \cdot \At \tag{\Cref{lem:language}}
		\end{align*}
	\end{itemize}
\end{proof}

% PASTE HERE
\section{The proof of \cref{thm:sf_gkat_fix}}\label{app:proof of fix}

In order to show that every provable equivalence between deterministic one-free regular expressions is obtainable from a deterministic proof, we need to take a detour through the completeness proof of Grabmayer and Fokkink~\cite{regexlics} for one-free regular expressions modulo bisimilarity.

Grabmayer and Fokkink's completeness proof revolves around a notion of \emph{solution to an operational model}.
The operational models of one-free regular expressions they solve are LTSs with the so-called \emph{LLEE-property}, defined in~\cite{regexlics}.
The following is a slightly simpler but equivalent condition on LTSs called \emph{well-layeredness}, found in~\cite{starmfps}.

\begin{definition}
	Let \((X, \tau)\) be a LTS\@.
	An \emph{entry/body labelling} \((X, \eo, \bo)\) of \((X, \tau)\) is a labelling of each transition \(x \tr{\alpha p}_\tau y\) as either an \emph{entry transition}, written \(x \eo y\), or as a \emph{body transition}, written \(x \bo y\).
	Given an entry/body labelling \((X, \eo, \bo)\) of \((X, \tau)\) and \(x,y \in X\), we write \(x \diredge y\) to denote that there is a path
	\[
		x \eo x_1 \bo x_2 \bo \cdots \bo x_n \bo y
	\]
	such that \(x \notin \{x_1, \dots, x_n,y\}\).
\end{definition}

We use \(x \to y\) to denote that \(x \tr{\alpha p} y\) for some \(\alpha p \in \At\cdot\Sigma\), and use \((-)^*\) and \((-)^+\) to denote reflexive-transitive and transitive closures respectively.

\begin{definition}[see~\cite{starmfps}]\label{def:well-layeredness}
	An LTS \((X, \tau)\) is said to be \emph{well-layered} if it has an entry/body labelling \((X, \eo, \bo)\) that satisfies the following conditions.
	\begin{enumerate}
		\item (\emph{local finiteness}) For any \(x \in X\), \(\langle x \rangle\) is finite.
		\item (\emph{flatness}) For any \(x,y \in X\), if \(x \eo y\), then \(\neg (x \bo y)\).
		\item (\emph{full specification}) For any \(x,y \in X\),
		\begin{enumerate}
			\item \(\neg(x \bo^+ x)\) (there are no body loops), and
			\item if \(x \eo y\) and \(y \neq x\), then \(y \to^+ x\) (every entry transition is a loop entry transition).
		\end{enumerate}
		\item (\emph{layeredness}) The graph \((X, \diredge)\) is acyclic.
		\item (\emph{goto-free}) If \(x \diredge y\), then \(\neg(y \to \checkmark)\).
	\end{enumerate}
	If \((X, \eo, \bo)\) satifies these conditions, we say that \((X, \eo, \bo)\) is a \emph{layering witness} for \((X, \tau)\).
\end{definition}

Grabmayer and Fokkink's completeness proof technique revolves around being able to \emph{solve} certain systems of equations represented by LTSs.
This is captured abstractly by the following notion~\cite[Definition 2.8]{regexlics}.

\begin{definition}
	A \emph{solution} to an LTS \((X, \tau)\) is a function \(\varphi\colon X \to \SExp\) such that for any \(x \in X\),
	\[
		\varphi(x) \equiv_* \sum_{x \tr{\alpha p} \checkmark} \alpha p + \sum_{x \tr{\alpha p} y} \alpha p~\varphi(y)
	\]
\end{definition}

Grabmayer and Fokkink's completeness proof technique can now be summarized as follows.
\begin{description}
	\item[Step 1.] Prove that for any \(r \in \SExp\), \(\langle r \rangle\) is well-layered.
	\item[Step 2.] Prove that every  well-layered LTS admits a unique solution up to \(\equiv_*\).
	That is, if \(\varphi_1,\varphi_2\) are two solutions to the well-layered LTS \((X, \tau)\), then \(\varphi_1(x) \equiv_* \varphi_2(x)\) for all \(x \in X\).
	\item[Step 3.] Prove that if \(h \colon (Y, \tau_Y) \to (X, \tau_X)\) is a coalgebra homomorphism and \(\varphi\) is a solution to \((X, \tau)\), then \(\varphi \circ h\) is a solution to \((Y, \tau)\).
	\item[Step 4.] Show that the bisimulation collapse of a well-layered LTS is well-layered.
\end{description}

The following theorem is a strengthening of the statement in Step 4.
The bisimulation collapse case is originally due to Grabmayer and Fokkink~\cite{regexlics}.
It was observed in~\cite[Theorem 4.1]{starmfps} that a slight modification to Grabmayer and Fokkink's proof establishes the version below.

\begin{theorem}\label{thm:closure under homomorphic images}
	Let \((X, \tau_X)\) be a well-layered LTS and \(h\colon (X, \tau_X) \to (Y, \tau_Y)\) be a surjective coalgebra homomorphism.
	Then \((Y, \tau_Y)\) is also well-layered.
\end{theorem}

Roughly speaking, what we are going to do now is rework Steps 1 through 4 of Grabmayer and Fokkink's completeness proof for deterministic well-layered LTSs and their corresponding notion of \emph{deterministic solution}.

\begin{definition}
	Let \((X,\tau)\) be a deterministic LTS and \(\varphi \colon X \to \Det\).
	We say that \(\varphi\) is a \emph{deterministic solution} if for any \(x \in X\),
	\[
	\varphi(x) \equiv_*^{\det} \sum_{x \tr{\alpha p} \checkmark} \alpha p + \sum_{x \tr{\alpha p} y} \alpha p~\varphi(y)
	\]
	Given two deterministic solutions \(\varphi_1\) and \(\varphi_2\), we write \(\varphi_1 \equiv_*^{\det} \varphi_2\) if for any \(x \in X\), \(\varphi_1(x) \equiv_*^{\det} \varphi_2(x)\).
\end{definition}

%In \cite{regexlics}, Grabmayer and Fokkink provide a sufficient condition for guaranteeing that a LTS admits a unique solution up to (not necessarily deterministic) provable equivalence.
%The condition is the existence of a so-called \emph{layering witness}, an auxiliary transition system with edges labelled either by \(\eo\), for \emph{loop entry}, or \(\bo\), for \emph{body} (see \cref{def:layering witness} for details).
The unique solution to a well-layered LTS is given by a formula.

\begin{definition}\label{def:deterministic_canonical_solution}
	Define the two quantities below for any \(x \in X\)
	\begin{align*}
		|x|_{en} &= \max \{n \mid (\exists x_1,\dots,x_n)~x \diredge x_1 \diredge \cdots \diredge x_n\} \\
		|x|_{bo} &= \max \{n \mid (\exists x_1,\dots,x_n)~x \bo x_1 \bo \cdots \bo x_n\}
	\end{align*}
	These are finite because in a layering witness \((X,\eo, \bo)\), \((X,\diredge)\) and \((X, \bo)\) do not contain infinite paths.

	Let \(x,z \in X\), and let
	\begin{equation}
		\label{eq:abc}
		\begin{gathered}
			a = \bigvee\{\alpha \mid x \etr{\alpha p} y,y \in X\}
			\qquad
			b = \bigvee\{\alpha \mid x \tr{\alpha p} x\}
			\\
			c = \bigvee\{\alpha \mid x \tr{\alpha p} \checkmark\}
		\end{gathered}
	\end{equation}
	The \emph{canonical solution} \(\varphi_X\) to \((X,\tau)\) (given the layering witness \((X,\eo, \bo)\)) is defined recursively on \(|x|_{bo}\) as follows:
	\[
		\varphi_X(x)
		= \Big(\sum_{x \etr{\alpha p} x} \alpha p +^b \sum_{\substack{x \etr{\alpha p} y\\x\neq y}} \alpha p~t_X(y,x)\Big) *^a \Big(\sum_{x \tr{\alpha p} \checkmark} \alpha p +^c \sum_{\substack{x \btr{\alpha p} y\\x\neq y}} \alpha p~\varphi_X(y)\Big)
	\]
	where we write \(r +^b s\) and \(r *^b s\) to denote the terms \(r + s\) and \(r * s\) respectively, as well as the statement \(r \perp_b s\).
	Above, the expression \(t_X(y,x)\) is defined for every pair \((y,x)\) such that \(x \diredge y\) by recursion on \((|x|_{en}, |y|_{bo})\) in the lexicographical ordering of \(\N^2\) as follows:
	where
	\begin{equation}\
		\label{eq:a'b'c'}
		\begin{gathered}
		a' = \bigvee \{\alpha\mid y \etr{\alpha p} z, z \in X\}
		\qquad
		b' = \bigvee \{\alpha \mid y \tr{\alpha p} y\}
		\\
		c' = \bigvee \{\alpha \mid y \btr{\alpha p} x\}
		\end{gathered}
	\end{equation}
	we define
	\begin{align*}
		t_X(y, x)
		&= \Big(\sum_{y \etr{\alpha p} y} \alpha p
		+^{b'} \sum_{\substack{y \etr{\alpha p} z\\z\neq y}} \alpha p~t_X(z,y)\Big)
		\\&\hspace{4em}*^{a'} \Big(\sum_{y \btr{\alpha p} x} \alpha p
		+^{c'} \sum_{\substack{y \btr{\alpha p} z\\z\neq x}} \alpha p~t_X(z, x)\Big)
	\end{align*}
\end{definition}

The well-definedness of the formulas appearing in the definition above deserves further explanation.

The expression \(\varphi_X(x)\) is defined by induction on \(|x|_{bo}\).
In the base case of this induction, \(x\) has no outgoing body transitions.
This gives the equivalent formula,
\begin{align*}
	\varphi_X(x)
	&= \Big(\sum_{x \etr{\alpha p} x} \alpha p +^b \sum_{\substack{x \etr{\alpha p} y\\x\neq y}} \alpha p~t_X(y,x)\Big) *^a \Big(\sum_{x \tr{\alpha p} \checkmark} \alpha p\Big)
\end{align*}
The full recursive formula is well-defined (if each \(t_X(y, x)\) is), because if \(x \bo y\), then \(|y|_{bo} < |x|_{bo}\).

Both of these formulas depend on the expression \(t_X(y, x)\), which is further defined by induction on \((|x|_{bo}, |y|_{en})\).
In the base case of this induction, \(|x|_{en} = 1\) and \(|y|_{bo} = 0\), because we require \(x \diredge y\).
So, \(y\) has no outgoing body transitions.
Since \(x \diredge y\) in the expression \(t_X(y, x)\), \(|y|_{en} < |x|_{en} = 1\), so \(|y|_{en} = 0\), and \(y\) has no outgoing loop entry transitions.
This gives the equivalent expression
\begin{align*}
	t_X(y, x)
	&= \Big(\sum_{y \etr{\alpha p} y} \alpha p\Big)*^{a'} 0
\end{align*}
which has no further dependencies.
For well-definedness of \(t_X(y, x)\), we must check well-definedness of the subterms \(t_X(z, y)\) and \(t_X(z,x)\).
The former appears in a context where $y \etr{\alpha p} z$, and so $y \diredge z$; furthermore, $(|y|_{en}, |z|_{bo}) < (|x|_{en}, |y|_{bo})$ because $x \diredge y$, which tells us that $|y|_{en} < |x|_{en}$.
As for the subterm $t_X(z, x)$, note that because $y \btr{\alpha p} z$ we have that $x \diredge z$ and $|z|_{bo} < |y|_{bo}$, meaning that $(|x|_{en}, |z|_{bo}) < (|x|_{en}, |y|_{bo})$.

We are specifically interested in solving deterministic LTSs.
The following terminology will be useful in proofs.

\begin{definition}
	A state \(x\) in a LTS \((X, \tau)\) is \emph{operationally deterministic} if \(\tau(x)\) is graph-like.
	That is, for any \(\xi,\xi' \in \checkmark + X\), \(x \tr{\alpha p} \xi\) and \(x \tr{\alpha q} \xi'\) implies \(p = q\) and \(\xi = \xi'\).
\end{definition}

Thus, an LTS is deterministic precisely when each of its states is operationally deterministic.
Determinism is preserved by homomorphisms.

\begin{lemma}\label{lem:determinism homom}
	If \(x\) is an operationally deterministic state of the LTS \((X,\tau_X)\) and \(h \colon (X,\tau_X) \to (Y,\tau_Y)\) is a homomorphism, then \(h(x)\) is operationally deterministic.
\end{lemma}

\begin{proof}
	There are three cases to consider.
	\begin{enumerate}
		\item Suppose \(h(x) \tr{\alpha p} \checkmark\) and \(h(x) \tr{\alpha q} \checkmark\).
		Then \((\alpha p, \checkmark),(\alpha q,\checkmark) \in \tau_X(x)\) because \(h\) is a homomorphism.
		It follows that \(p = q\), because \(x\) is operationally deterministic.
		\item If \(h(x) \tr{\alpha p} y_1\) and \(h(x) \tr{\alpha q} y_2\) for some \(y,y' \in Y\), then there are \(x_1,x_2 \in X\) such that \(h(x_1) = y_1\), \(h(x_2) = y_2\), \(x \tr{\alpha p} x_1\), and \(x \tr{\alpha q} x_2\).
		Since \(x\) is operationally deterministic, \(p = q\) and \(x_1 = x_2\).
		Hence, \(h(x_1) = h(x_2)\).
		\item If \(h(x) \tr{\alpha p} \checkmark\) and \(h(x) \tr{\alpha q} y\) for some \(y \in Y\), then \(x \tr{\alpha p} \checkmark\) and \(x \tr{\alpha q} x'\) for some \(x'\) such that \(h(x') = y\).
		This is not possible because \(x\) is operationally deterministic, which would then require \(p = q\) and \(\checkmark = x'\), despite \(X \cap \checkmark = \emptyset\).
	\end{enumerate}
\end{proof}

We are now ready to describe the structure of the proof of \cref{thm:sf_gkat_fix}.
The proof requires the following four properties of deterministic regular expressions, deterministic well-layered LTSs, and their deterministic solutions.
\begin{enumerate}
	\item (\Cref{lem:deterministic_subcoalgebra}) \(\Det\) is a deterministic subcoalgebra of \((\SExp,\ell)\).
	\item (\Cref{lem:deterministic_fundamental_theorem}) For any \(r \in \Det\), the inclusion map \(\incl_{\langle r\rangle} \colon \langle r \rangle \to \SExp\) is a deterministic solution.
	This is equivalent to saying that
	\[
	r \equiv_*^{\det} \sum_{r\tr{\alpha p} \checkmark} \alpha p + \sum_{r\tr{\alpha p} s} \alpha p~s
	\]
	\item (\Cref{thm:deterministic_canonical_solution}) Fix a layering witness \((X,\eo,\bo)\) for a deterministic LTS \((X,\tau)\).
    The following properties hold:
	\begin{enumerate}
		\item[(a)]\(\varphi_X\) is a deterministic solution to \((X,\tau)\). That is, for any \(x \in X\),\[
		\varphi_X(x) \equiv_*^{\det} \sum_{x\tr{\alpha p}\checkmark}  \alpha p + \sum_{x\tr{\alpha p} y} \alpha p~\varphi_X(y)
		\]
		\item[(b)] For any deterministic solution \(\psi\) to \((X,\tau)\) and for all \(x \in X\), \[
		\varphi_X(x) \equiv_*^{\det} \psi(x)
		\]
	\end{enumerate}
	\item (\Cref{lem:homoms_preserve_det_solutions}) Let \(h\colon (X,\tau_X) \to (Y,\tau_Y)\) be a homomorphism between deterministic LTSs and let \(\varphi\colon Y \to \Det\) be a deterministic solution to \((Y,\tau)\).
	Then \(\varphi\circ h\) is a deterministic solution to \((X,\tau_X)\).
\end{enumerate}
We now provide the proof of \cref{thm:sf_gkat_fix}, showing how these four properties collectively imply the theorem, and will then later present the proofs of all above properties.

\begin{proof}[Proof of \Cref{thm:sf_gkat_fix}.]
	Suppose \(r,s \in \Det\).
	By \cref{lem:deterministic_subcoalgebra}, \(\langle r \rangle\) and \(\langle s\rangle\) are deterministic subcoalgebras of \((\SExp, \ell)\).
    Because the latter is well-layered~\cite{starmfps}, and subcoalgebras of well-layered LTSs are again well-layered (this follows easily from the definition), we know that \(\langle r \rangle\) and \(\langle s \rangle\) are also well-layered.

	Now, if \(r \equiv_* s\), then \(r \bisim s\) by soundness.
	By~\cite[Theorem 4.2]{coalgebra}, there is a minimal LTS \((X,\tau)\) and homomorphisms \(\langle r\rangle \xrightarrow{h} (X,\tau) \xleftarrow{k} \langle s\rangle\) such that \(h(r) = k(s)\).
	Since well-layeredness and determinism are preserved by homomorphisms (\cref{thm:closure under homomorphic images,lem:determinism homom}), \((X,\tau)\) is a deterministic well-layered LTS\@.
	It follows from \Cref{thm:deterministic_canonical_solution} that \((X,\tau)\) has a deterministic solution \(\varphi_X\).

	By \Cref{lem:homoms_preserve_det_solutions}, \(\varphi_X\circ h\) and \(\varphi_X\circ k\) are deterministic solutions to \(\langle r\rangle\) and \(\langle s\rangle\) respectively.
	\Cref{lem:deterministic_subcoalgebra,lem:deterministic_fundamental_theorem} tell us that \(\incl_{\langle r\rangle}\) and \(\incl_{\langle s\rangle}\) are also deterministic solutions to \(\langle r\rangle\) and \(\langle s\rangle\) respectively.
	Therefore, by \cref{thm:deterministic_canonical_solution},
	\begin{align*}
		r &= \incl_{\langle r\rangle} (r) \equiv_*^{\det} \varphi_X\circ h(r) \\
		s &= \incl_{\langle s\rangle} (s) \equiv_*^{\det} \varphi_X\circ k(s)
	\end{align*}
	Since \(h(r) = k(s)\), we see from the derivations above that \(r \equiv_*^{\det} s\).
\end{proof}

In the proof of \cref{thm:sf_gkat_fix} we used the four properties outlined above (\cref{lem:deterministic_subcoalgebra,lem:deterministic_fundamental_theorem,thm:deterministic_canonical_solution,lem:homoms_preserve_det_solutions}), as depicted in \cref{fig:dependencies}. We will now prove all the needed properties individually.

\begin{figure}[t]
\begin{center}
	\begin{tikzpicture}[scale=1.1]
		\node[state, rectangle] at (0,0.5) (main) {\cref{thm:sf_gkat_fix}};
		\node[state, rectangle] at (-3.5,-1.2) (24) {\cref{lem:deterministic_subcoalgebra}};
		\node[state, rectangle] at (-1.3,-1.2) (25) {\cref{lem:deterministic_fundamental_theorem}};
		\node[state, rectangle] at (1.3,-1.2) (26) {\cref{thm:deterministic_canonical_solution}};
		\node[state, rectangle] at (4.25,-1.2) (28) {\cref{lem:homoms_preserve_det_solutions}};
		\node[state, rectangle] at (1.3, -2.2) (27) {\cref{lem:deterministic_solution_properties}};
		\node[state, rectangle] at (2.5, 1.5) (12) {\cref{thm:closure under homomorphic images}};
		\node[state, rectangle] at (-2.5, 1.5) (22) {\cref{lem:determinism homom}};

		\draw (main) edge[-stealth, thick] (24);
		\draw (main) edge[-stealth, thick] (25);
		\draw (main) edge[-stealth, thick] (12);
		\draw (main) edge[-stealth, thick] (22);
		\draw (main) edge[-stealth, thick] (26);
		\draw (26) edge[-stealth, thick] (27);
		\draw (main) edge[-stealth, thick] (28);

		%\draw[dotted] (-5.1, -0.4) rectangle (-2.8, -2.5);
		%\draw[dotted] (-2.1, -0.4) rectangle (0.2, -2.5);
		%\draw[dotted] (0.7, -0.4) rectangle (3.3, -2.5);
		%\draw[dotted] (3.8, -0.4) rectangle (6.2, -2.5);

		\node[fill=white] at (-0.6, -0.4) (0) {\it Property 2};
		\node at (-3.5, -0.4) (0) {\it Property 1};
		\node[fill=white] at (.9, -0.4) (0) {\it Property 3};
		\node at (4.5, -0.4) (0) {\it Property 4};

	\end{tikzpicture}
\end{center}
\caption{Results needed for the proof of \cref{thm:sf_gkat_fix}. An arrow $A \to B$ denotes that \(B\) is used in the proof of \(A\).}\label{fig:dependencies}.
\end{figure}

\paragraph*{Property 1.} Let us start with the first property.

\begin{lemma}\label{lem:deterministic_subcoalgebra}
	\(\Det\) is a deterministic subcoalgebra of \((\SExp,\ell)\).
\end{lemma}

\begin{proof}
	We show by induction on \(r \in \Det\) that \(r\) is operationally deterministic and that if \(r \tr{\alpha p} s\), then \(s \in \Det\).
	There are two base cases.
	\begin{itemize}
		\item The deterministic expression \(0\) has no outgoing transitions at all, and so the two properties we are trying to show hold vacuously.
		\item The expression \(\alpha p\) has one outgoing transition, \(\alpha p \tr{\alpha p} \checkmark\), which does not land on another state.
		Every state with at most one outgoing transition is operationally deterministic, so we are done here.
	\end{itemize}
	For the inductive step, suppose \(r_1,r_2 \in \Det\), \(r_1,r_2\) are both operationally deterministic, and \(r_i \tr{\alpha p} s\) implies that \(s \in \Det\) for either \(i=1,2\).
	\begin{itemize}
		\item Suppose \(r_1 \perp_b r_2\) for some \(b \in \BExp\).
		If \(r_1 + r_2 \tr{\alpha p} \xi\) and \(r_1 + r_2 \tr{\alpha q} \xi'\) for \(\alpha \le b\), then \(r_1 \tr{\alpha p} \xi\) and \(r_1 \tr{\alpha q} \xi'\), because \(\bar b\cdot r_2 \bisim 0\).
		By the induction hypothesis, \(p = q\) and \(\xi = \xi'\).
		Also, by the induction hypothesis, if \(\xi = t \in \SExp\), then \(t \in \Det\).
		The case for \(\alpha \le \bar b\) and \(r_2\) is similar.

		\item Now suppose we have \(r_1r_2 \tr{\alpha p} s\) and \(r_1r_2 \tr{\alpha q} s'\) (it is not possible for \(r_1r_2 \to \checkmark\)).
		Then for some \(\xi,\xi' \in \checkmark + \SExp\), \(r_1 \tr{\alpha p} \xi\) and \(r_1 \tr{\alpha q} \xi'\) are two outgoing transitions from \(r_1\).
		By the induction hypothesis, \(r_1\) is operationally deterministic, so \(p = q\) and \(\xi = \xi'\).
		If \(\xi = \checkmark\), then we must have \(r_1r_2 \tr{\alpha p} r_2\).
		Hence, \(s = s' = r_2 \in \Det\), as desired.
		If \(\xi = t \in \SExp\), then \(s = tr_2 = s'\).
		By the induction hypothesis, \(t \in \Det\), so \(t r_2 \in \Det\), as desired.

		\item Now suppose \(r_1 \perp_b r_2\) for some \(b \in \BExp\).
		If \(r_1 * r_2 \tr{\alpha p} \xi\) and \(r_1 * r_2 \tr{\alpha q} \xi'\) for some \(\alpha \le \bar b\), then \(r_2 \tr{\alpha p} \xi\) and \(r_2 \tr{\alpha q} \xi'\), because \(\bar b \cdot r_1 \bisim 0\).
		By the induction hypothesis, \(p = q\) and \(\xi = \xi' \in \checkmark + \Det\).
		Now suppose \(\alpha \le b\).
		Then \(r_1 \tr{\alpha p} s\) and \(r_1 \tr{\alpha q} s'\) as well as \(\xi = s(r_1*r_2)\) and \(\xi' = s'(r_1*r_2)\).
		By the induction hypothesis applied to \(r_1\), \(p = q\) and \(s = s' \in \Det\).
		Hence, \(s(r_1*r_2) = s'(r_1*r_2) \in \Det\).
	\end{itemize}
\end{proof}

\paragraph*{Property 2.} Now we turn our attention to the second property.

\begin{lemma}%
	\label{lem:deterministic_fundamental_theorem}
	Given \(r \in \Det\), \(\incl_{\langle r\rangle} \colon \langle r \rangle \to \SExp\) is a deterministic solution.
\end{lemma}

Recall that, given \(r,s \in \SExp\), we write \(r +^b s\) and \(r *^b s\) for the expressions \(r + s\) and \(r * s\) respectively, as well as the statement \(r \perp_b s\).
Note that we will refer to these superscripts as \emph{separation markers}.
Also note that if \(r,s \in \Det\), then \(r +^b s\) and \(r *^b s\) indicate that \(r + s\) and \(r * s\) are deterministic expressions.

\begin{proof}[Proof of \cref{lem:deterministic_fundamental_theorem}.]
	We are going to prove that
	\[
	r \equiv_*^{\det} \sum_{r\tr{\alpha p} \checkmark} \alpha p + \sum_{r\tr{\alpha p} s} \alpha ps
	\]
	by induction on \(r \in \Det\).
	We have yet to see why this equation is between two deterministic expressions, however:
	Notice that on the right-hand side of the equation above, \(\sum_{r\tr{\alpha p} \checkmark} \alpha p\) and \(\sum_{r\tr{\alpha p} s} \alpha ps\) are deterministic expressions because there is at most one transition of the form \(r \tr{\alpha p} \xi\) per \(\alpha \in \At\) by operational determinism of \(r\).
	Moreover, if \(b = \bigvee\{\alpha \mid r \tr{\alpha p} \checkmark\}\), then \(\sum_{r\tr{\alpha p} \checkmark} \alpha p \perp_b \sum_{r\tr{\alpha p} s} \alpha ps\).
	This shows that the right-hand side of the equation we are about to derive is deterministic.

	There are two base cases.
	\begin{itemize}
		\item Since \(0\) has no outgoing transitions,
		\[
		\sum_{0\tr{\alpha p} \checkmark} \alpha p + \sum_{0\tr{\alpha p} s} \alpha ps  = 0 + 0 \equiv_* 0
		\]
		where the first equality is literal syntactic equality and the second is the axiom \(x + x = x\).
		We have already seen that the first expression is deterministic, and the first equality is literal syntactic equality, so both are deterministic.
		%The expression \(0 + 0\) is deterministic because \(0 \in \Det\) and \(0 \perp_1 0\).

		\item Since \(\alpha p\) has only one outgoing transition,
		\[
			\sum_{\alpha p\tr{\beta q} \checkmark} \beta q + \sum_{\alpha p\tr{\beta q} s} \beta qs
			\equiv_* \alpha p + 0 \equiv_* \alpha p
		\]
		We have already seen that the first expression is deterministic.
		The second is deterministic because \(0,\alpha p \in \Det\) and \(\alpha p \perp_\alpha 0\).
	\end{itemize}
	For the inductive step, assume that for \(i=1,2\),
	\[
	r_i \equiv_*^{\det} \sum_{r_i\tr{\alpha p} \checkmark} \alpha p + \sum_{r_i\tr{\alpha p} s} \alpha ps
	\]
	\begin{itemize}
		\item Let \(r_1 \perp_c r_2\), and define \(b_i = \bigvee \{r_i \tr{\alpha p} \checkmark\}\) for \(i=1,2\).
		Then %\vspace{2em}
		\begin{align*}
			&r_1 +^c r_2\\
			&\equiv_*^{\det}
			\Big(\sum_{r_1 \tr{\alpha  p} \checkmark} \alpha p
			+^{b_1} \sum_{r_1 \tr{\alpha  p} s} \alpha p s\Big)
			+^c
			\Big(\sum_{r_2 \tr{\alpha  p} \checkmark} \alpha p
			+^{b_2} \sum_{r_2 \tr{\alpha  p} s} \alpha p s\Big)
			\tag{ind.~hyp.}\\
			&\equiv_*^{\det}
			\sum_{r_1 \tr{\alpha  p} \checkmark} \alpha p +^{b_1\wedge c} \Big(\sum_{r_1 \tr{\alpha  p} s} \alpha p s
			+^c
			\Big(\sum_{r_2 \tr{\alpha  p} \checkmark} \alpha p +^{b_2} \sum_{r_2 \tr{\alpha  p} s} \alpha p s\Big)\Big)
			\tag{assoc.}\\
			&\equiv_*^{\det}
			\sum_{r_1 \tr{\alpha  p} \checkmark} \alpha p
			+^{b_1\wedge c} \Big(\Big(\sum_{r_1 \tr{\alpha  p} s} \alpha p s
			+^c \sum_{r_2 \tr{\alpha  p} \checkmark} \alpha p \Big)
			+^{c\vee b_2} \sum_{r_2 \tr{\alpha  p} s} \alpha p s\Big)
			\tag{assoc.}\\
			&\equiv_*^{\det}
			\sum_{r_1 \tr{\alpha  p} \checkmark} \alpha p
			+^{b_1\wedge c} \Big(\Big(\sum_{r_2 \tr{\alpha  p} \checkmark} \alpha p
			+^{\bar c} \sum_{r_1 \tr{\alpha  p} s} \alpha p s \Big)
			+^{c\vee b_2} \sum_{r_2 \tr{\alpha  p} s} \alpha p s\Big)
			\tag{comm., \(r\perp_c s\) iff \(s \perp_{\bar c} r\)}\\
			&\equiv_*^{\det}
			\sum_{r_1 \tr{\alpha  p} \checkmark} \alpha p
			+^{b_1\wedge c} \Big(\sum_{r_2 \tr{\alpha  p} \checkmark} \alpha p
			+^{\bar c \wedge b_2} \Big(\sum_{r_1 \tr{\alpha  p} s} \alpha p s
			+^{c\vee b_2} \sum_{r_2 \tr{\alpha  p} s} \alpha p s\Big)\Big)
			\tag{assoc.}\\
			&\equiv_*^{\det}
			\Big(\sum_{r_1 \tr{\alpha  p} \checkmark} \alpha p
			+^{b_1\wedge c} \sum_{r_2 \tr{\alpha  p} \checkmark} \alpha p\Big)
			+^{b} \Big(\sum_{r_1 \tr{\alpha  p} s} \alpha p s
			+^{c\vee b_2} \sum_{r_2 \tr{\alpha  p} s} \alpha p s\Big)
			\tag{assoc.,\(\bullet\)}\\
			&\equiv_*^{\det} \sum_{r_1+ r_2 \tr{\alpha  p}\checkmark} \alpha p
			+^{b} \sum_{r_1 + r_2 \tr{\alpha  p} s} \alpha p s
		\end{align*}
		\(^\bullet\)In the last two steps, we define \(b = (b_1 \wedge c) \vee (\bar c \wedge b_2)\).
		The separation markers in the derivation above indicate how to construct each expression that appears as a deterministic expression.
		This establishes that there is a deterministic proof of \(r_1 + r_2 \equiv_* \sum_{r_1+ r_2 \tr{\alpha  p}\checkmark} \alpha p
+ \sum_{r_1 + r_2 \tr{\alpha  p} s} \alpha p s\).

		\item In the inductive case for sequential composition, we have
		\begin{align*}
			r_1r_2
			&\equiv_*^{\det}
			\Big(\sum_{r_1 \tr{\alpha  p} \checkmark} \alpha p
			+^{b_1} \sum_{r_1 \tr{\alpha  p} s} \alpha p s\Big)
			r_2 \tag{ind.~hyp.}\\
			&\equiv_*^{\det}
			\sum_{r_1 \tr{\alpha  p} \checkmark} \alpha p r_2
			+^{b_1} \sum_{r_1 \tr{\alpha  p} s} \alpha p sr_2
			\tag{dist.}\\
			&\equiv_*^{\det}
			0 +^{0}\Big(\sum_{r_1 \tr{\alpha  p} \checkmark} \alpha p r_2
			+^{b_1} \sum_{r_1 \tr{\alpha  p} s} \alpha p sr_2\Big)
			\tag{\(\neg(r_1r_2 \tr{\alpha  p} \checkmark)\)}\\
			&\equiv_*^{\det}
			\sum_{r_1r_2 \tr{\alpha  p}\checkmark} \alpha p
			+^{0} \sum_{r_1r_2 \tr{\alpha  p} s} \alpha p s
			\tag{rearr.~summands}
		\end{align*}

		\item In the inductive case for the star, let \(r_1 \perp_c r_2\) and compute
		\begin{align*}
			&r_1* r_2 \\
			&\equiv_*^{\det} r_1(r_1*r_2) +^c r_2 \tag{fixed-point}\\
			&\equiv_*^{\det}
			\Big(\sum_{r_1 \tr{\alpha  p} \checkmark} \alpha p
			+^{b_1} \sum_{r_1 \tr{\alpha  p} s} \alpha p s\Big) (r_1* r_2) +^c r_2
			\tag{ind.~hyp.}\\
			&\equiv_*^{\det}
			\Big(\sum_{r_1 \tr{\alpha  p} \checkmark} \alpha p(r_1* r_2)
			+^{b_1} \sum_{r_1 \tr{\alpha  p} s} \alpha p s(r_1* r_2)\Big)  +^c r_2
			\tag{dist.}\\
			&\equiv_*^{\det}
			\Big(\sum_{r_1 \tr{\alpha  p} \checkmark} \alpha p(r_1* r_2)
			+^{b_1} \sum_{r_1 \tr{\alpha  p} s} \alpha p s(r_1* r_2)\Big)
			\\&\hspace{8em}+^c \Big(\sum_{r_2 \tr{\alpha  p} \checkmark} \alpha p
			+^{b_2} \sum_{r_2 \tr{\alpha  p} s} \alpha p s\Big)
			\tag{ind.~hyp.} \\
			&\equiv_*^{\det}
			\sum_{r_2 \tr{\alpha  p} \checkmark} \alpha p
			+^{\bar c \wedge b_2}
			\Big(\sum_{r_1 \tr{\alpha  p} \checkmark} \alpha p(r_1* r_2)
			\\&\hspace{8em} +^{b_1\wedge c} \Big(\sum_{r_1 \tr{\alpha  p} s} \alpha p s(r_1* r_2)
			+^{c \wedge \bar{b_1}} \sum_{r_2 \tr{\alpha  p} s} \alpha p s\Big)\Big)
			\tag{assoc.}\\
			&\equiv_*^{\det} \sum_{r_1* r_2 \tr{\alpha  p}\checkmark} \alpha p
			+^{b} \sum_{r_1* r_2 \tr{\alpha  p} s} \alpha p s
		\end{align*}
		where \(b = \bar c \wedge b_2\).
	\end{itemize}
\end{proof}

\paragraph*{Property 3.} Now we turn our attention to the third property, and recall from \cref{fig:dependencies} that this will need an extra result (\cref{lem:deterministic_solution_properties}).

Fix a locally finite layering witness \((X, \eo, \bo)\) for a deterministic LTS \((X,\tau)\).
Write \(\varphi_X\) for the canonical solution to \((X,\tau)\) given \((X, \eo, \bo)\).
We need the following lemma.

\begin{lemma}\label{lem:deterministic_solution_properties}
	Let \(x,z \in X\).
	\begin{enumerate}
		\item \(\varphi_X(x)\) and \(t_X(x,z)\) are deterministic expressions.
		\item If \(x \diredge y\), then \[
			\varphi_X(y) \equiv_*^{\det} t_X(y,x)\varphi_X(x)
		\]
		\item If \(x \diredge y\) and \(\psi\) is any deterministic solution to \((X,\tau)\), then \[
			\psi(y) \equiv_*^{\det} t_X(y,x)\psi(x)
		\]
	\end{enumerate}
\end{lemma}

\begin{proof}
	We begin by showing item 1, that \(\varphi_X(x)\) and \(t_X(x,y)\) are deterministic expressions for any \(x,z \in X\).
	This can be seen by induction on \(|x|_{bo}\) and \((|x|_{en}, |y|_{bo})\) in \(\N\) and the lexicographical ordering of \(\N\times\N\) respectively.
	Recall the definitions of \(a,b,c,a',b',c'\) from~\eqref{eq:abc} and~\eqref{eq:a'b'c'}.

	Let us start with \(\varphi_X(x)\).
	In the base case, \(|x|_{bo} = 0\), so \(x\) makes no outgoing \(\bo\) transitions.
	\begin{align*}
		\varphi_X(x)
		&= \Big(\sum_{x \etr{\alpha p} x} \alpha p +^b \sum_{x \etr{\alpha p} y} \alpha p~t_X(y, x)\Big) *^a \Big(\sum_{x \tr{\alpha p} \checkmark} \alpha p +^c 0\Big)
	\end{align*}
	This is deterministic, as can be seen from the separation markers \(a,b,c\), and because \(x\) is operationally deterministic.

	In the inductive step for \(\varphi_X(x)\), observe from the separation markers in \Cref{def:deterministic_canonical_solution} that \(\varphi_X(x)\) is constructed from sums and stars of separated one-free regular expressions.
	Therefore, it suffices to check that the summands are deterministic.
	By induction, the \(\varphi_X(y)\) that appears in this expression is deterministic, because \(|y|_{bo} < |x|_{bo}\).
	It now suffices to show that \(t_X(y, x)\) is deterministic.

	In the base case, because \(x \diredge y\), \(|x|_{en} = 1\) and \(|y|_{en} = |y|_{bo} = 0\), so \(y\) makes no outgoing \(\eo\) transitions or \(\bo\) transitions.
	This leaves us with the expression
	\[
		t_X(y, x) = \Big(\sum_{y \etr{\alpha p} y} \alpha p +^{b'} 0\Big) *^{a'} \Big(0 +^{c'} 0\Big)
	\]
	This expression is deterministic, as can be seen from the separation markers \(a',b',c'\) from~\eqref{eq:a'b'c'}, and because \(y\) is operationally deterministic.

	In the inductive case for \(t_X(y, x)\), again observe that the separation markers indicate how to construct \(t_X(y, x)\) from separated one-free regular expressions.
	It now suffices to see that the sub-expressions \(t_X(z,y),t_X(z,x)\) are deterministic.
	This follows from the induction hypothesis:
	\begin{itemize}
		\item Where \(t_X(z,y)\) appears in \(t_X(y,x)\), \(x \diredge y\).
		This means that \(|y|_{en} < |x|_{en}\) and therefore \((|y|_{en}, |z|_{bo}) < (|x|_{en}, |y|_{bo})\) in the lexicographical ordering.
		By the induction hypothesis, \(t_X(z, y)\) is deterministic.
		\item Where \(t_X(z,x)\) appears in \(t_X(y,x)\), \(y \bo z\).
		This means that \(|z|_{bo} < |y|_{bo}\), so that \((|x|_{en}, |z|_{bo}) < (|x|_{en}, |y|_{bo})\) in the lexicographical ordering.
		By the induction hypothesis, \(t_X(z,x)\) is deterministic.
	\end{itemize}
	Thus, \(t_X(y, x)\) is deterministic.

	We show item 2 by induction on \(|y|_{bo}\).
	Let
	\[
	r = \sum_{y \etr{\alpha p} y} \alpha p
	+^{b'} \sum_{\substack{y \etr{\alpha p} z \\ y \neq z}} \alpha p~t_X(z, y)
	\]
	(i.e., $r$ is the first part of $t_X(y, x)$ and $\varphi_X(y)$).
	Note that we have already seen that \(r\) is deterministic.

	In the base case, \(|y|_{bo} = 0\), so \(t_X(y,x) = r*0 = \varphi_X(y)\).
	Observe that the following derivation is a deterministic proof whenever \(e,f,g \in \Det\) and \(e \perp f\):
	\begin{equation}
		\label{eq:star assoc det}
		(e * f)g
		\equiv_*^{\det} (e(e * f) + f)g
		\equiv_*^{\det} e(e * f)g + fg
		\equiv_*^{\det} e * (fg)
	\end{equation}
	where the last equality is an instance of the recursive specification principle.
	Applying this to the situation at hand, we have \[
	t_X(y,x) \varphi_X(x)
    = (r*0)\varphi_X(x)
	\equiv_*^{\det} r * (0 \varphi_X(x))
	\equiv_*^{\det} r * 0
	= \varphi_X(y)
	\]
	This concludes the base case.

	Let \(d = \bigvee \{\alpha \mid y \tr{\alpha  p} \checkmark\}\).
	In the inductive step, we have
	\begin{align*}
		\varphi_X(y) \
		&\equiv_*^{\det} r*^{a'} \Big(\sum_{y \tr{\alpha  p} \checkmark} \alpha p +^{d} \sum_{y \btr{\alpha  p} z} \alpha p~\varphi_X(z) \Big)
		\tag{def.}\\
		&\equiv_*^{\det} r*^{a'} \Big(0 +^{0} \sum_{y \btr{\alpha  p} z} \alpha p~\varphi_X(z) \Big)
		\tag{\(\neg(y \to \checkmark)\)}\\
		&\equiv_*^{\det} r*^{a'} \Big(\sum_{y \btr{\alpha  p} z} \alpha p~\varphi_X(z) \Big)
		\tag{zero} \\
		&\equiv_*^{\det} r*^{a'} \Big(\sum_{y \btr{\alpha  p} x} \alpha p~\varphi_X(x) +^{c'} \sum_{\substack{y \btr{\alpha  p} z\\z \neq x}} \alpha p~\varphi_X(z) \Big)
		\tag{rearr.~summands} \\
		&\equiv_*^{\det} r*^{a'} \Big(\sum_{y \btr{\alpha  p} x} \alpha p~\varphi_X(x) +^{c'} \sum_{\substack{y \btr{\alpha  p} z\\z \neq x}} \alpha p~t_X(z,x)\varphi_X(x) \Big)
		\tag{ind.~hyp., \(|z|_{bo} < |y|_{bo}\)}\\
		&\equiv_*^{\det} r*^{a'} \Big(\sum_{y \btr{\alpha  p} x} \alpha p +^{c'} \sum_{\substack{y \btr{\alpha  p} z\\z \neq x}} \alpha p~t_X(z,x) \Big)\varphi_X(x)
		\tag{dist.} \\
		&\equiv_*^{\det} t_X(y,x)\varphi_X(x) \tag{def.~\(t_X(y,x)\), \eqref{eq:star assoc det}} % chktex 2
	\end{align*}
	Indeed, by item 1 and as indicated by the separation markers above, each of the expressions in this derivation is deterministic.

	Finally, we prove item 3 by induction on \((|x|_{en},|y|_{bo})\) in the lexicographical ordering of \(\N\times\N\), assuming \(x \diredge y\).
	First, derive
	\begin{align*}
		\psi(y)
		&\equiv_*^{\det} 0 +^0 \sum_{y \tr{\alpha p} z} \alpha p ~\psi(z) \tag{$\psi$ is a solution, \(\neg(y \to \checkmark)\)} \\
	    & \equiv_*^{\det} \sum_{y \tr{\alpha p} z} \alpha p ~\psi(z) \\
		&\equiv_*^{\det} \sum_{y \etr{\alpha p} z} \alpha p ~\psi(z) +^{a'} \sum_{y \btr{\alpha p} z} \alpha p ~\psi(z)
		\tag{rearr.~summands} \\
		&\equiv_*^{\det} \Big(\sum_{y \etr{\alpha p} y} \alpha p ~\psi(y) +^{b'} \sum_{\substack{y \etr{\alpha p} z\\y \neq z}} \alpha p ~\psi(z)\Big) \\ & \hspace{8em}+^{a'} \Big(
			\sum_{y \btr{\alpha p} x} \alpha p ~\psi(z)
			+^{c'} \sum_{\substack{y \btr{\alpha p} z\\z \neq x}} \alpha p ~\psi(z)
		\Big)
		\tag{rearr.~summands}
	\end{align*}

	In the base case, \(|x|_{en} = 1\) and \(|y|_{bo}=0\), and we also know \(|y|_{en}=0\) since \(x \diredge y\).
	The above then becomes
	\begin{align*}
		\Big(\sum_{y \etr{\alpha p} y} \alpha p ~\psi(y) +^{b'} 0\Big) +^{a'} 0
		&\equiv_*^{\det} \Big(\sum_{y \etr{\alpha p} y} \alpha p ~\psi(y) +^{b'} 0~\psi(y)\Big) +^{a'} 0 \\
		&\equiv_*^{\det} \Big(\sum_{y \etr{\alpha p} y} \alpha p +^{b'} 0\Big)\psi(y) +^{a'} 0
	\end{align*}
	so by (\textsf{RSP}),
    \begin{align*}
		\psi(y)
		&\equiv_*^{\det} (\sum_{y \etr{\alpha p} y} \alpha p +^{b'} 0) *^{a'} 0 \\
		&\equiv_*^{\det} \Big((\sum_{y \etr{\alpha p} y} \alpha p +^{b'} 0) *^{a'} 0\Big) \psi(x) \\
		&\equiv_*^{\det} t_X(y,x) \psi(x)
	\end{align*}
	In the inductive step,
	\begin{align*}
		\psi(y)
		&\equiv_*^{\det} \Big(
				\sum_{y \etr{\alpha p} y} \alpha p ~\psi(y)
				+^{b'} \sum_{\substack{y \etr{\alpha p} z\\y \neq z}} \alpha p ~\psi(z)
			\Big) \\ & \hspace{8em}+^{a'} \Big(
				\sum_{y \btr{\alpha p} x} \alpha p ~\psi(z)
				+^{c'} \sum_{\substack{y \btr{\alpha p} z\\z \neq x}} \alpha p ~\psi(z)
			\Big)\\
		&\equiv_*^{\det} \Big(
				\sum_{y \etr{\alpha p} y} \alpha p ~\psi(y) +^{b'} \sum_{\substack{y \etr{\alpha p} z\\y \neq z}} \alpha p ~t_X(z,y)\psi(y)
			\Big) \\ & \hspace{8em}+^{a'} \Big(
				\sum_{y \btr{\alpha p} x} \alpha p~\psi(x)
				+^{c'} \sum_{\substack{y \btr{\alpha p} z\\z \neq x}} \alpha p ~t_X(z,x)\psi(x)
			\Big)
			\tag{ind.~hyp. with \(|y|_{en} < |x|_{en}\) and ind.~hyp with \(|z|_{bo}<|y|_{bo}\)}\\
		&\equiv_*^{\det} \Big(
				\sum_{y \etr{\alpha p} y} \alpha p
				+^{b'} \sum_{\substack{y \etr{\alpha p} z\\y \neq z}} \alpha p ~t_X(z,y)
			\Big)\psi(y) \\ & \hspace{8em}+^{a'} \Big(
				\sum_{y \btr{\alpha p} x} \alpha p~\psi(x)
				+^{c'} \sum_{\substack{y \btr{\alpha p} z\\z \neq x}} \alpha p ~t_X(z,x)\psi(x)
			\Big)
			\tag{dist.} \\
		&\equiv_*^{\det} \Big(
				\sum_{y \etr{\alpha p} y} \alpha p
				+^{b'} \sum_{\substack{y \etr{\alpha p} z\\y \neq z}} \alpha p ~t_X(z,y)
			\Big) \\ & \hspace{8em}*^{a'} \Big(
				\sum_{y \btr{\alpha p} x} \alpha p~\psi(x)
				+^{c'} \sum_{\substack{y \btr{\alpha p} z\\z \neq x}} \alpha p ~t_X(z,x)\psi(x)
			\Big)
			\tag{\textsf{RSP}} \\
		&\equiv_*^{\det} \Big(
				\sum_{y \etr{\alpha p} y} \alpha p
				+^{b'} \sum_{\substack{y \etr{\alpha p} z\\y \neq z}} \alpha p ~t_X(z,y)
			\Big) *^{a'}\Big(
				\sum_{\substack{y \btr{\alpha p} z\\z \neq x}} \alpha p ~t_X(z,x)
			\Big)\psi(x)
			\tag{assoc.}\\
		&\equiv_*^{\det} t_X(y,x)\psi(x) \tag{def.~of \(t_X(y, x)\), \eqref{eq:star assoc det}} % chktex 2
	\end{align*}
\end{proof}

Now that we have proved \cref{lem:deterministic_solution_properties}, we can proceed with the proof of the following theorem which is the third property.

\begin{theorem}%
	\label{thm:deterministic_canonical_solution}
	The canonical solution \(\varphi_X\) is a also deterministic solution to \((X,\tau)\).
	That is, for any \(x \in X\),
	\[
	\varphi_X(x) \equiv_*^{\det} \sum_{x\tr{\alpha p} \checkmark} \alpha p + \sum_{x\tr{\alpha p} y} \alpha p~\varphi_X(y)
	\]
	Furthermore, for any deterministic solution \(\psi\) to \((X,\tau)\) and any \(x \in X\),
	\[
	\varphi_X(x) \equiv_*^{\det} \psi(x)
	\]
\end{theorem}

\begin{proof}%[Proof of \Cref{thm:deterministic_canonical_solution}.]
	Recall the definitions of \(a,b,c\) from~\eqref{eq:abc}.
	To see that the canonical solution to \((X,\tau)\) is a deterministic solution, we derive
	% let \(x,z \in X\), and \(b = \bigvee\{\alpha \mid x \tr{\alpha p} x\}\), \(c = \bigvee\{\alpha \mid x \etr{\alpha p} y, y \in X\}\), and \(d = \bigvee\{\alpha \mid x \tr{\alpha p} \checkmark\}\).
	\begin{align*}
		&\varphi_{X} (x) \\
		&\equiv_*^{\det} \Big(
		\sum_{x \etr{\alpha p} x} \alpha p
		+^{b} \sum_{\substack{x \etr{\alpha p} y \\ x \neq y}} \alpha p~t_X(y, x)
		\Big)*^{a} \Big(
		\sum_{x \tr{\alpha p} \checkmark} \alpha p
		+^{c} \sum_{\substack{x \tr{\alpha p}_{\mathsf b} y \\ x \neq y}} \alpha p~\varphi_X(y)
		\Big)
		\tag{def.}\\
		&\equiv_*^{\det} \Big(
		\sum_{x \etr{\alpha p} x} \alpha p
		+^{b} \sum_{\substack{x \etr{\alpha p} y \\ x \neq y}} \alpha p~t_X(y, x)
		\Big)\varphi_X(x) \\&\hspace{8em}+^{a} \Big(
		\sum_{x \tr{\alpha p} \checkmark} \alpha p
		+^{c} \sum_{\substack{x \tr{\alpha p}_{\mathsf b} y \\ x \neq y}} \alpha p~\varphi_X(y)
		\Big) \\
		&\equiv_*^{\det} \Big(
		\sum_{x \etr{\alpha p} x} \alpha p~\varphi_X(x)
		+^{b} \sum_{\substack{x \etr{\alpha p} y \\ x \neq y}} \alpha p~t_X(y, x) \varphi_X(x)
		\Big) \\&\hspace{8em}+^{a} \Big(
		\sum_{x \tr{\alpha p} \checkmark} \alpha p
		+^{c} \sum_{\substack{x \tr{\alpha p}_{\mathsf b} y \\ x \neq y}} \alpha p~\varphi_X(y)
		\Big)
		\tag{dist.}\\
		&\equiv_*^{\det} \Big(
		\sum_{x \etr{\alpha p} x} \alpha p~\varphi_X(x)
		+^{b} \sum_{\substack{x \etr{\alpha p} y \\ x \neq y}} \alpha p~\varphi_X(y)
		\Big) \\&\hspace{8em}+^{a} \Big(
		\sum_{x \tr{\alpha p} \checkmark} \alpha p
		+^{c} \sum_{\substack{x \tr{\alpha p}_{\mathsf b} y \\ x \neq y}} \alpha p~\varphi_X(y)
		\Big) \tag{\cref{lem:deterministic_solution_properties} item 2}\\
		&\equiv_*^{\det} 	 \sum_{x \tr{\alpha p} \checkmark} \alpha p +^{c} \sum_{x \tr{\alpha p} y} \alpha p~\varphi_X(y) \tag{rearr.~summands}
	\end{align*}
	In the last step, we used the fact that \(\bar c \ge_{\textsf{BA}} a \ge_{\textsf{BA}} b\).

	Now let \(\psi\) be any deterministic solution to \((X,\tau)\).
	To see that \(\psi(x) \equiv_*^{\det} \varphi_X(x)\) for all \(x \in X\), we proceed by induction on \(|x|_{bo}\).
	In the inductive step:
	\begin{align*}
		&\psi(x)  \\
		&\equiv_*^{\det}  \Big(\sum_{x \etr{\alpha p} x} \alpha p~\psi(x) +^b \sum_{\substack{x\etr{\alpha p} y \\ y \neq x}} \alpha p~\psi(y) \Big) +^{a} \Big(\sum_{x \tr{\alpha p} \checkmark} \alpha p +^{c} \sum_{\substack{x \btr{\alpha p}y}} \alpha p ~\psi(y)\Big)
		\tag{def.} \\
		&\equiv_*^{\det} \Big(\sum_{x \etr{\alpha p} x} \alpha p~\psi(x) +^b \sum_{\substack{x\etr{\alpha p} y \\ y \neq x}} \alpha p~t_X(y,x)\psi(x) \Big) \\&\hspace{4em}+^{a} \Big(\sum_{x \tr{\alpha p} \checkmark} \alpha p +^{c} \sum_{\substack{x \btr{\alpha p}y}} \alpha p ~\psi(y)\Big)
		\tag{\Cref{lem:deterministic_solution_properties} item 3}\\
		&\equiv_*^{\det} \Big(\sum_{x \etr{\alpha p} x} \alpha p +^b \sum_{\substack{x\etr{\alpha p} y \\ y \neq x}} \alpha p~t_X(y,x) \Big)\psi(x) +^{a} \Big(\sum_{x \tr{\alpha p} \checkmark} \alpha p +^{c} \sum_{\substack{x \btr{\alpha p}y}} \alpha p ~\psi(y)\Big)
		\tag{dist.} \\
		&\equiv_*^{\det}  \Big(\sum_{x \etr{\alpha p} x} \alpha p +^b \sum_{\substack{x\etr{\alpha p} y \\ y \neq x}} \alpha p~t_X(y,x) \Big)*^{a} \Big(\sum_{x \tr{\alpha p} \checkmark} \alpha p +^{c} \sum_{\substack{x \btr{\alpha p}y}} \alpha p ~\psi(y)\Big)
		\tag{\textsf{RSP}} \\
		&\equiv_*^{\det}  \Big(\sum_{x \etr{\alpha p} x} \alpha p +^b \sum_{\substack{x\etr{\alpha p} y \\ y \neq x}} \alpha p~t_X(y,x) \Big)*^{a} \Big(\sum_{x \tr{\alpha p} \checkmark} \alpha p +^{c} \sum_{\substack{x \btr{\alpha p}y}} \alpha p ~\varphi_X(y)\Big) \tag{ind.~hyp., \(|y|_{bo} < |x|_{bo}\)}\\
		&= \varphi_X(x)
	\end{align*}
	The base case is similar, except that it skips the second to last equality above, because \(\sum_{x \btr{\alpha p} y}\) is an empty sum.
\end{proof}

\paragraph*{Property 4.} Finally, we prove the fourth property we needed in \cref{thm:sf_gkat_fix}.

\begin{lemma}%
	\label{lem:homoms_preserve_det_solutions}
	Let \(h\colon (X,\tau_X) \to (Y,\tau_Y)\) be a homomorphism between deterministic LTSs and let \(\psi\colon Y \to \Det\) be a deterministic solution to \((Y,\tau)\).
	Then \(\psi\circ h\) is a deterministic solution to \((X,\tau_X)\).
\end{lemma}

\begin{proof}
	The key observations here are (1) \(x \tr{\alpha p} \checkmark\) iff \(h(x) \tr{\alpha p} \checkmark\), and (2) \(h(x) \tr{\alpha p} y\) iff there is an \(x' \in X\) such that \(h(x') = y\) and \(x \tr{\alpha p} x'\).
	Since (1) implies \(\sum_{h(x)\tr{\alpha p} \checkmark} \alpha p = \sum_{x\tr{\alpha p} \checkmark} \alpha p\), and since (2) implies \(\sum_{h(x)\tr{\alpha p} y} \alpha p~\psi(y) = \sum_{x\tr{\alpha p} x'} \alpha p~\psi(h(x'))\), we find that
	\begin{align*}
		\psi\circ h(x)
		&\equiv_*^{\det} \sum_{h(x)\tr{\alpha p} \checkmark} \alpha p + \sum_{h(x)\tr{\alpha p} y} \alpha p~\psi(y) \\
		&\equiv_*^{\det} \sum_{x\tr{\alpha p} \checkmark} \alpha p + \sum_{x\tr{\alpha p} x'} \alpha p~\psi\circ h(x')
	\end{align*}
	It follows that \(\psi\circ h\) is a deterministic solution to \((X, \tau)\).
\end{proof}

%This concludes our verification of the four properties used to prove \cref{thm:sf_gkat_fix}.

\section{Corrections Made to the Current Document}%
\label{app:changes}

%We the authors believe that skip-free \(\GKAT\) offers an elegant first step to solving the problems that the original work on \(\GKAT\)~\cite{kozentseng,gkatpopl} posed.
We discovered that the original version of this paper left a gap in its main proof, the argument towards completeness of skip-free bisimulation \(\GKAT\) w.r.t.~bisimilarity.
Thankfully, we have since filled this gap.
A corrected completeness proof now appears in this updated version of the paper (\cref{app:proof of fix}), but we explain here the gap and how we proceeded to fix it. At the end of this appendix, we also list some other minor changes we took the opportunity to do which correct a few other small errors and typos in the paper.

\paragraph*{Filling the gap.}
The gap in the argument for \cref{thm:completeness I} (completeness of skip-free bisimulation \(\GKAT\)) could originally be found in an inductive step of the proof of \cref{lem:rtg is also good stuff} --- or rather, in the lack of a certain inductive step: \cref{lem:rtg is also good stuff} states that if \(r \equiv_* s\) for some \(r,s \in \Det\), then \(\rtg(r) \equiv_\dagger \rtg(s)\).
In the original paper, the proof of \cref{lem:rtg is also good stuff} proceeded by induction on the derivation of  \(r \equiv_* s\).
The inductive step covers the loop axiom
\[
	\infer[\textsf{RSP}]{t = rt + s}{t = r * s}
\]
of Grabmayer and Fokkinks's system and the inference rules of equational logic:
\[
	\infer[Con\(_\diamond\)]{r_1 = s_1 \quad r_2 = s_2}{r_1\mathbin\diamond r_2 = s_1 \mathbin\diamond s_2}
	\qquad
	\infer[Sym]{s = r}{r = s}
	\qquad
	\infer[Tra]{r = t \quad t = s}{r = s}
\]
where \(\diamond \in \{+, \cdot, *\}\).
The inductive step corresponding to \textsf{RSP} was considered in the original paper.
The inductive steps corresponding to \textsc{Con\(_\diamond\)} and \textsc{Sym} are easy.
The inductive step corresponding to \textsc{Tra} was fallaciously also thought to be easy.
As we discovered, it does not go through.

If \textsc{Tra} is the last step in the proof of \(r \equiv_* s\), the induction hypothesis states that if \(r \equiv_* t\) with \(r,t \in \Det\), then \(\rtg(r) \equiv_\dagger \rtg(t)\), and similarly for \(t\) and \(s\).
Being able to apply the induction hypothesis requires that we have \(t \in \Det\).
This may not be true.
For instance, take \(r = s = a0\) and \(t = a0 + a0\); then we may be in the (admittedly contrived) situation
\[
	\infer{
		\infer{\vdots}{a0 = a0 + a0}
		\qquad
		\infer{\vdots}{a0 + a0 = a0}
	}{
		a0 = a0
	}
\]
in which the induction hypothesis does not apply.
In the current edition of this paper, we circumvent this problem with \cref{thm:sf_gkat_fix}, which allows us to consider only \emph{deterministic proofs} of \(r \equiv_* s\).
If \(r \equiv_* s\) with \(r,s \in \Det\), a deterministic proof of \(r \equiv_* s\) is a derivation in which only deterministic one-free regular expressions appear.
\cref{thm:sf_gkat_fix} tells us that every provable equivalence \(r \equiv_* s\) with \(r,s \in \Det\) can be derived from a deterministic proof.
Thus, the proof of \cref{lem:rtg is also good stuff} can proceed by induction on the deterministic proof of \(r \equiv_* s\).
Now, in the case where \textsc{Tra} is the last step in the deterministic proof of \(r \equiv_* s\), the expression \(t\) is deterministic by assumption.
This allows us to apply the induction hypothesis and finish the completeness proof.

In order to prove \cref{thm:sf_gkat_fix}, we have had to delve into the details of Grabmayer and Fokkink's completeness proof for one-free regular expressions modulo bisimilarity~\cite{onefreeregexlics}.
The proof of \cref{thm:sf_gkat_fix} carefully goes through Grabmayer and Fokkink's proof and ensures that each step can be carried out for deterministic expressions without the use of nondeterministic ones.
The details can be found in \cref{app:proof of fix}, which is a completely new addition to this document, and which we are confident patches the hole in our original proof.

\paragraph*{Other small changes}
Aside from the change to the proof of \cref{thm:sf_gkat_fix} (\cref{app:proof of fix}), we also list below a few other changes we did to the original paper.
\begin{itemize}
	\item \cref{def:deterministic sexp} used to say that if $r_1, r_2 \in \Det$, then $b \cdot r_1 + \overline{b} \cdot r_2 \in \Det$ and $(b \cdot r_1) * (\overline{b} \cdot r_2) \in \Det$.
    This has been relaxed in our revision, to match the way deterministic one-free star expressions were already used in proofs.

	\item ``Lemma~\ref{lem:rtg is also good stuff}'' has been renamed to \cref{lem:rtg is also good stuff} to reflect the additional effort necessary to get to a proof.

	\item The list of useful equalities used throughout the appendix now appears in \cref{lemma:useful-equalities}, and has been expanded.

	\item Name tags have been added to the axioms in \cref{fig:gkat axioms} and the axiomatic derivations of skip-free GKAT equivalences, such as those from \cref{lemma:useful-equalities}.

	\item The proof of the intermediate property used in \cref{lem:arbitrary seperation} (i.e., that $\rtg(r) = b \rtg(r)$ when $r \equiv_* b \cdot r$) mistakenly assumed the premise was true on the \emph{operands} of the inductive cases.
	This proof has been corrected and clarified.
\end{itemize}

%% file: main.bbl
\begin{thebibliography}{10}
\providecommand{\url}[1]{\texttt{#1}}
\providecommand{\urlprefix}{URL }
\providecommand{\doi}[1]{https://doi.org/#1}

\bibitem{aceto94complete}
Aceto, L.: Deriving complete inference systems for a class of {GSOS} languages generation regular behaviours. In: CONCUR. pp. 449--464 (1994). \doi{10.1007/978-3-540-48654-1_33}

\bibitem{aceto11gsos}
Aceto, L., Caltais, G., Goriac, E., Ing{\'{o}}lfsd{\'{o}}ttir, A.: Axiomatizing {GSOS} with predicates. In: SOS. pp. 1--15 (2011). \doi{10.4204/EPTCS.62.1}

\bibitem{aceto11preg}
Aceto, L., Caltais, G., Goriac, E., Ing{\'{o}}lfsd{\'{o}}ttir, A.: {PREG} axiomatizer - {A} ground bisimilarity checker for {GSOS} with predicates. In: CALCO. pp. 378--385 (2011). \doi{10.1007/978-3-642-22944-2_27}

\bibitem{netkat}
Anderson, C.J., Foster, N., Guha, A., Jeannin, J.B., Kozen, D., Schlesinger, C., Walker, D.: {NetKAT}: semantic foundations for networks. In: POPL. pp. 113--126 (2014). \doi{10.1145/2535838.2535862}

\bibitem{awodey}
Awodey, S.: Category theory. Oxford university press (2010)

\bibitem{birkhoff}
Birkhoff, G.: On the structure of abstract algebras. Mathematical Proceedings of the Cambridge Philosophical Society  \textbf{31}(4),  433–454 (1935). \doi{10.1017/S0305004100013463}

\bibitem{bloomesik}
Bloom, S.L., {\'{E}}sik, Z.: Iteration Theories - The Equational Logic of Iterative Processes. {EATCS} Monographs on Theoretical Computer Science, Springer (1993). \doi{10.1007/978-3-642-78034-9}

\bibitem{brzozowski}
Brzozowski, J.A.: Derivatives of regular expressions. J. {ACM}  \textbf{11}(4),  481--494 (1964). \doi{10.1145/321239.321249}

\bibitem{DBLP:conf/pldi/ChajedTKZ19}
Chajed, T., Tassarotti, J., Kaashoek, M.F., Zeldovich, N.: Argosy: verifying layered storage systems with recovery refinement. In: PLDI. pp. 1054--1068 (2019). \doi{10.1145/3314221.3314585}

\bibitem{cohen-1994}
Cohen, E.: Hypotheses in {K}leene algebra. Tech. rep., Bellcore (1994)

\bibitem{ernie}
Cohen, E.: Weak {K}leene algebra is sound and (possibly) complete for simulation (2009). \doi{10.48550/arXiv.0910.1028}

\bibitem{hypotheses}
Doumane, A., Kuperberg, D., Pous, D., Pradic, P.: Kleene algebra with hypotheses. In: FOSSACS. pp. 207--223 (2019). \doi{10.1007/978-3-030-17127-8_12}

\bibitem{probnetkat}
Foster, N., Kozen, D., Mamouras, K., Reitblatt, M., Silva, A.: Probabilistic {NetKAT}. In: ESOP. pp. 282--309 (2016). \doi{10.1007/978-3-662-49498-1_12}

\bibitem{netkat-decision}
Foster, N., Kozen, D., Milano, M., Silva, A., Thompson, L.: A coalgebraic decision procedure for {NetKAT}. In: POPL. pp. 343--355 (2015). \doi{10.1145/2676726.2677011}

\bibitem{regexlics}
Grabmayer, C.: Milner's proof system for regular expressions modulo bisimilarity is complete: Crystallization: Near-collapsing process graph interpretations of regular expressions. In: LICS. pp. 34:1--34:13 (2022). \doi{10.1145/3531130.3532430}

\bibitem{onefreeregexlics}
Grabmayer, C., Fokkink, W.J.: A complete proof system for 1-free regular expressions modulo bisimilarity. In: LICS. pp. 465--478 (2020). \doi{10.1145/3373718.3394744}

\bibitem{kmt}
Greenberg, M., Beckett, R., Campbell, E.H.: Kleene algebra modulo theories: a framework for concrete {KAT}s. In: PLDI. pp. 594--608 (2022). \doi{10.1145/3519939.3523722}

\bibitem{gumm}
Gumm, H.P.: Functors for coalgebras. Algebra Universalis  \textbf{45} (11 1998). \doi{10.1007/s00012-001-8156-x}

\bibitem{gummelements}
Gumm, H.P.: Elements of the general theory of coalgebras. LUATCS'99, Rand Afrikaans University, Johannesburg  (1999)

\bibitem{huntington}
Huntington, E.V.: Sets of independent postulates for the algebra of logic. Transactions of the American Mathematical Society  \textbf{5}(3),  288--309 (1904). \doi{10.1090/S0002-9947-1904-1500675-4}

\bibitem{kao}
Kapp{\'{e}}, T., Brunet, P., Rot, J., Silva, A., Wagemaker, J., Zanasi, F.: Kleene algebra with observations. In: CONCUR. pp. 41:1--41:16 (2019). \doi{10.4230/LIPIcs.CONCUR.2019.41}

\bibitem{ckao}
Kapp{\'{e}}, T., Brunet, P., Silva, A., Wagemaker, J., Zanasi, F.: Concurrent {K}leene algebra with observations: From hypotheses to completeness. In: FOSSACS. pp. 381--400 (2020). \doi{10.1007/978-3-030-45231-5_20}

\bibitem{cka}
Kapp\'{e}, T., Brunet, P., Silva, A., Zanasi, F.: Concurrent {K}leene algebra: Free model and completeness. In: ESOP. pp. 856--882 (2018). \doi{10.1007/978-3-319-89884-1_30}

\bibitem{kleene}
Kleene, S.C.: Representation of events in nerve nets and finite automata. Automata studies  \textbf{34},  3--41 (1956)

\bibitem{DBLP:journals/entcs/KotK05}
Kot, L., Kozen, D.: Kleene algebra and bytecode verification. Electron. Notes Theor. Comput. Sci.  \textbf{141}(1),  221--236 (2005). \doi{10.1016/j.entcs.2005.02.028}

\bibitem{kozen}
Kozen, D.: A completeness theorem for {K}leene algebras and the algebra of regular events. Inf. Comput.  \textbf{110}(2),  366--390 (1994). \doi{10.1006/inco.1994.1037}

\bibitem{katintro}
Kozen, D.: Kleene algebra with tests and commutativity conditions. In: TACAS. pp. 14--33 (1996). \doi{10.1007/3-540-61042-1_35}

\bibitem{kaequations}
Kozen, D., Mamouras, K.: Kleene algebra with equations. In: ICALP. pp. 280--292 (2014). \doi{10.1007/978-3-662-43951-7_24}

\bibitem{DBLP:conf/cl/KozenP00}
Kozen, D., Patron, M.: Certification of compiler optimizations using {K}leene algebra with tests. In: CL. pp. 568--582 (2000). \doi{10.1007/3-540-44957-4_38}

\bibitem{katcompleteness}
Kozen, D., Smith, F.: Kleene algebra with tests: Completeness and decidability. In: CSL. pp. 244--259 (1996). \doi{10.1007/3-540-63172-0_43}

\bibitem{kozentseng}
Kozen, D., Tseng, W.D.: The {B}{\"{o}}hm-{J}acopini theorem is false, propositionally. In: MPC. pp. 177--192 (2008). \doi{10.1007/978-3-540-70594-9_11}

\bibitem{cka-precursor}
Laurence, M.R., Struth, G.: Completeness theorems for pomset languages and concurrent {K}leene algebras (2017). \doi{10.48550/arXiv.1705.05896}

\bibitem{horntheory}
Makowsky, J.A.: Why {H}orn formulas matter in computer science: Initial structures and generic examples. J. Comput. Syst. Sci.  \textbf{34}(2/3),  266--292 (1987). \doi{10.1016/0022-0000(87)90027-4}

\bibitem{milner}
Milner, R.: A complete inference system for a class of regular behaviours. J. Comput. Syst. Sci.  \textbf{28}(3),  439--466 (1984). \doi{10.1016/0022-0000(84)90023-0}

\bibitem{kah-tools}
Pous, D., Rot, J., Wagemaker, J.: On tools for completeness of {K}leene algebra with hypotheses. In: RAMICS. pp. 378--395 (2021). \doi{10.1007/978-3-030-88701-8_23}

\bibitem{ka-top}
Pous, D., Wagemaker, J.: Completeness theorems for {K}leene algebra with top. In: CONCUR. pp. 26:1--26:18 (2022). \doi{10.4230/LIPIcs.CONCUR.2022.26}

\bibitem{fizzbuzz}
Rees, J.: Fizz Buzz: 101 Spoken Numeracy Games. Learning Development Aids (2002)

\bibitem{coalgebra}
Rutten, J.J.M.M.: Universal coalgebra: a theory of systems. Theor. Comput. Sci.  \textbf{249}(1),  3--80 (2000). \doi{10.1016/S0304-3975(00)00056-6}

\bibitem{salomaa}
Salomaa, A.: Two complete axiom systems for the algebra of regular events. J. {ACM}  \textbf{13}(1),  158--169 (1966). \doi{10.1145/321312.321326}

\bibitem{orderedprocesses}
Schmid, T.: A (co)algebraic framework for ordered processes (2022). \doi{10.48550/arXiv.2209.00634}

\bibitem{gkaticalp}
Schmid, T., Kapp{\'{e}}, T., Kozen, D., Silva, A.: Guarded {K}leene algebra with tests: Coequations, coinduction, and completeness. In: ICALP. pp. 142:1--142:14 (2021). \doi{10.4230/LIPIcs.ICALP.2021.142}

\bibitem{starmfps}
Schmid, T., Rot, J., Silva, A.: On star expressions and coalgebraic completeness theorems. In: MFPS. pp. 242--259 (2021). \doi{10.4204/EPTCS.351.15}

\bibitem{processesparametrised}
Schmid, T., Rozowski, W., Silva, A., Rot, J.: Processes parametrised by an algebraic theory. In: ICALP. pp. 132:1--132:20 (2022). \doi{10.4230/LIPIcs.ICALP.2022.132}

\bibitem{gkatpopl}
Smolka, S., Foster, N., Hsu, J., Kapp{\'{e}}, T., Kozen, D., Silva, A.: Guarded {K}leene algebra with tests: verification of uninterpreted programs in nearly linear time. In: POPL. pp. 61:1--61:28 (2020). \doi{10.1145/3371129}

\bibitem{cantorscott}
Smolka, S., Kumar, P., Foster, N., Kozen, D., Silva, A.: Cantor meets {S}cott: semantic foundations for probabilistic networks. In: POPL. pp. 557--571 (2017). \doi{10.1145/3009837.3009843}

\bibitem{mcnetkat}
Smolka, S., Kumar, P., Kahn, D.M., Foster, N., Hsu, J., Kozen, D., Silva, A.: Scalable verification of probabilistic networks. In: PLDI. pp. 190--203 (2019). \doi{10.1145/3314221.3314639}

\bibitem{monodic}
Takai, T., Furusawa, H.: Monodic tree {K}leene algebra. In: RelMICS/AKA. pp. 402--416 (2006). \doi{10.1007/11828563_27}

\bibitem{ska}
Wagemaker, J., Bonsangue, M.M., Kapp{\'{e}}, T., Rot, J., Silva, A.: Completeness and incompleteness of synchronous {K}leene algebra. In: MPC. pp. 385--413 (2019). \doi{10.1007/978-3-030-33636-3_14}

\bibitem{pocka}
Wagemaker, J., Brunet, P., Docherty, S., Kapp{\'{e}}, T., Rot, J., Silva, A.: Partially observable concurrent {K}leene algebra. In: CONCUR. pp. 20:1--20:22 (2020). \doi{10.4230/LIPIcs.CONCUR.2020.20}

\bibitem{gkatlearning}
Zetzsche, S., Silva, A., Sammartino, M.: Guarded {K}leene algebra with tests: Automata learning (2022). \doi{10.48550/arXiv.2204.14153}

\end{thebibliography}
